\documentclass[opre,nonblindrev]{informs3rad}
\OneAndAHalfSpacedXI
\usepackage{comment}

\ECRepeatTheorems

\MANUSCRIPTNO{OPRE-2021-05-337} 

\usepackage{xfrac}
\usepackage{bbm}
\usepackage{tcolorbox}
\DeclareRobustCommand{\mybox}[2][gray!15]{
\begin{tcolorbox}[  
        left=0pt,
        right=0pt,
        top=0pt,
        bottom=0pt,
        colback=#1,
        colframe=#1,
        enlarge left by=0mm,
        boxsep=10pt,
        arc=2pt,outer arc=2pt,
        ]
        #2
\end{tcolorbox}
}

\usepackage{csquotes}
\makeatletter
\renewenvironment*{displayquote}
  {\begingroup\setlength{\leftmargini}{1cm}\csq@getcargs{\csq@bdquote{}{}}}
  {\csq@edquote\endgroup}
\makeatother

\usepackage{comment}

\usepackage{colortbl}
\usepackage{bbm}
\usepackage{hhline}
\usepackage{mathtools}
\usepackage{booktabs}

\usepackage{subfig}

\usepackage{ifthen}
\usepackage{algpseudocode,algorithm,algorithmicx}

\algrenewcommand\algorithmicrequire{\textbf{Precondition:}}
\algrenewcommand\algorithmicensure{\textbf{Postcondition:}}
\usepackage{graphicx,color}
\usepackage{fancybox}
\usepackage{array,float}
\usepackage{url}
\begingroup
\makeatletter
\@for\theoremstyle:=definition,remark,plain,observation\do{%
	\expandafter\g@addto@macro\csname th@\theoremstyle\endcsname{%
		\addtolength\thm@preskip\parskip
	}%
}
\endgroup
\usepackage{mathrsfs}

\definecolor{cornellred}{rgb}{0.7, 0.11, 0.11}
\definecolor{dgreen}{rgb}{0.0, 0.5, 0.0}
\definecolor{ballblue}{rgb}{0.13, 0.67, 0.8}
\definecolor{royalblue(web)}{rgb}{0.25, 0.41, 0.88}
\definecolor{bleudefrance}{rgb}{0.19, 0.55, 0.91}
\definecolor{royalazure}{rgb}{0.0, 0.22, 0.66}
\usepackage{hyperref}
\hypersetup{
	colorlinks = true,
	linkcolor=royalazure,
	citecolor=royalazure,
	linkbordercolor = {white}
}

\usepackage{thmtools}
\usepackage{thm-restate}
\usepackage{cleveref}
\usepackage{enumitem}
\usepackage{multirow}
\usepackage{changepage}
\usepackage{tikz}

\DeclareRobustCommand*\cal{\@fontswitch\relax\mathcal}

\usepackage{pgf,tikz,pgfplots}
\pgfplotsset{compat=1.15}
\usetikzlibrary{arrows}
\usepackage{pgfplots}
\usetikzlibrary{patterns}
\usepgfplotslibrary{fillbetween}
\usetikzlibrary{intersections}

\usetikzlibrary{arrows, decorations.markings}

\tikzstyle{vecArrow} = [thick, decoration={markings,mark=at position
	1 with {\arrow[semithick]{open triangle 60}}},
	double distance=1.4pt, shorten >= 5.5pt,
	preaction = {decorate},
postaction = {draw,line width=1.4pt, white,shorten >= 4.5pt}]
\tikzstyle{innerWhite} = [semithick, white,line width=1.4pt, shorten >= 4.5pt]

\usepackage{perpage}

\makeatletter

\usepackage{accents}
	
\usepackage{natbib}
 \bibpunct[, ]{(}{)}{,}{a}{}{,}%
 \def\bibfont{\small}%
 \def\bibsep{\smallskipamount}%

\setcitestyle{authoryear}

\TheoremsNumberedThrough 
\EquationsNumberedThrough

	\newcommand*{\R}{\mathbb{R}}

	\DeclareMathOperator{\OPT}{\texttt{OPT}}

	\providecommand{\given}{}
	\DeclarePairedDelimiterX{\set}[1]\{\}{\renewcommand\given{\nonscript\:\delimsize\vert\nonscript\:\mathopen{}}#1}
	\let\Pr\relax
	\DeclarePairedDelimiterXPP{\Pr}[1]{\mathbb{P}}[]{}{\renewcommand\given{\nonscript\:\delimsize\vert\nonscript\:\mathopen{}}#1}
	\DeclarePairedDelimiterXPP{\Ex}[1]{\mathbb{E}}[]{}{\renewcommand\given{\nonscript\:\delimsize\vert\nonscript\:\mathopen{}}#1}

\newcolumntype{P}[1]{>{\centering\arraybackslash}c{#1}}

	\newcommand{\RN}[1]{\textcolor{red}{\bf (Rad: #1)}}       
	
	\newcommand{\AS}[1]{\textcolor{blue}{\bf (\em Amin: #1)}}

\newcommand*{\rom}[1]{\expandafter\romannumeral #1}
\newcommand{\Rom}[1]{\uppercase\expandafter{\romannumeral #1\relax}}

\newcommand{\onlineUnweightedRatioText}
{(1-\sfrac{1}{e}+\sfrac{1}{e^2})}
\newcommand{\onlineUnweightedRatio}
{1-\frac{1}{e}+\frac{1}{e^2}}
\newcommand{\onlineWeightedRatio}{0.7613}

\newcommand{\primed}{^\dagger}

\newcommand{\combprob}{\lambda}
\newcommand{\matchprob}{q}
\newcommand{\matchprobi}{\matchprob_{\driver}}
\newcommand{\matchprobs}{\{\matchprob_\driver\}_{\driver\in[\totaldriver]}}
\newcommand{\msprob}{p}

\newcommand{\matchprobHat}{\hat\matchprob}
\newcommand{\msprobHat}{\hat{\msprob}}
\newcommand{\msc}{c}

\newcommand{\weight}{w}
\newcommand{\weightHat}{\hat{\weight}}
\newcommand{\weighti}{\weight_\driver}
\newcommand{\weights}{\{\weight_\driver\}_{\driver\in[\totaldriver]}}

\newcommand{\contribution}{\texttt{OBJ}}
\newcommand{\contributioni}{\contribution^{(\mspair)}}

\newcommand{\rider}{i}
\newcommand{\Rider}{D}
\newcommand{\RiderHat}{\hat \Rider}
\newcommand{\Rideri}{D^{(\mspair)}}
\newcommand{\RideriHat}{\hat D^{(\mspair)}}

\newcommand{\driver}{j}
\newcommand{\Driver}{S}
\newcommand{\DriverHat}{\hat S}
\newcommand{\Driveri}{S^{(\mspair)}}
\newcommand{\DriveriHat}{\hat S^{(\mspair)}}
\newcommand{\Driverfirst}{T}
\newcommand{\DriverfirstHat}{\hat \Driverfirst}
\newcommand{\Driverfirsti}{\Driverfirst^{(\mspair)}}
\newcommand{\Driversecond}{\bar\Driverfirst}
\newcommand{\Driversecondi}{\bar\Driverfirst^{(\mspair)}}
\newcommand{\Driversecondb}{\tilde\Driverfirst}

\newcommand{\OPTHat}{\hat{\OPT}}
\newcommand{\GD}{\texttt{GD}}
\newcommand{\GDi}{\GD^{(\mspair)}}

\newcommand{\MS}{\texttt{MS}}

\newcommand{\MSi}{\MS^{(\mspair)}}

\newcommand{\OPTi}{\OPT^{(\mspair)}}

\newcommand{\N}{\mathbb{N}}

\newcommand{\totaldriver}{n}
\newcommand{\totalmspair}{L}
\newcommand{\totalmspairHat}{\hat L}
\newcommand{\mspair}{l}
\newcommand{\msfirst}{s}
\newcommand{\msfirsti}{\msfirst^{(\mspair)}}

\newcommand{\mssecond}{k}
\newcommand{\mssecondi}{\mssecond^{(\mspair)}}
\newcommand{\msseconds}{\{\mssecondi\}_{\mspair\in[\totalmspair]}}

\newcommand{\msdriver}{n}
\newcommand{\msdriveri}{\msdriver^{(\mspair)}}

\newcommand{\msrider}{m}
\newcommand{\msrideri}{\msrider^{(\mspair)}}

\newcommand{\mspairc}{\msc}
\newcommand{\mspairci}{\mspairc^{(\mspair)}}

\newcommand{\mspairinstance}{\{(\msdriveri,\msfirsti,\msrideri)\}_{\mspair\in[\totalmspair]}}

\newcommand{\mspairinstancenew}{\{(\msdriveri,\msfirsti,\msrideri,\mssecondi)\}_{\mspair\in[\totalmspair]}}

\newcommand{\RaHat}{\hat D_{1}}
\newcommand{\RpHat}{\hat D_{2}}

\newcommand{\weightp}{\bar\weight}
\newcommand{\weightdp}{\tilde\weight}

\newcommand{\driverTilde}{\tilde{\driver}}
\newcommand{\driverBar}{\bar{\driver}}

\newcommand{\val}{v}
\newcommand{\thresh}{t}
\newcommand{\Thresh}{\mathcal{T}}
\newcommand{\price}{p}
\newcommand{\acceptprob}{q}

\newcommand{\assignprob}{x}
\newcommand{\assignprobi}{\assignprob_{\rider\driver}}
\newcommand{\assignprobs}{\{\assignprobi\}_{(\rider,\driver)\in E}}

\newcommand{\assignvertexprob}{y}
\newcommand{\assignvertexprobi}{\assignvertexprob_{\rider}}
\newcommand{\assignvertexprobj}{\assignvertexprob_{\driver}}
\newcommand{\Ra}{D_{1}}
\newcommand{\Rp}{D_{2}}
\newcommand{\Rpa}{\tilde{D}_{2}}
\newcommand{\Fu}{F_{i}}

\newcommand{\RP}{\R_{\geq 0}}
\newcommand{\opt}{\texttt{OPT}}
\newcommand{\optoff}{\texttt{OPT}_{\textrm{offline}}}
\newcommand{\optoffmp}{\texttt{OPT}_{\textrm{offline-MP}}}
\newcommand{\opton}{\texttt{OPT}_{\textrm{online}}}
\newcommand{\maxmatch}{\textsc{Max-WeightMatch}^{\mathbf{w}}}
\newcommand{\maxmatchboth}{\textsc{Max-WeightMatch}^{\mathbf{\hat w}}}
\newcommand{\compratioM}{\Gamma_{\textrm{weighted-matching}}}
\newcommand{\compratioMUW}{\Gamma_{\textrm{matching}}}
\newcommand{\compratioMP}{\Gamma_{\textrm{matching-pricing}}}
\newcommand{\objr}{\texttt{OBJ}^{\Ra\cup\Rp}}
\newcommand{\objd}{\texttt{OBJ}^{S}}

\newcommand{\feasibles}{\mathcal{X}}
\newcommand{\quant}{q}

\newcommand{\secondstageratio}{\beta}
\newcommand{\firststagepercent}{\gamma}

\newcommand{\prob}[2][]{\text{Pr}\ifthenelse{\not\equal{}{#1}}{_{#1}}{}\!\left[{\def\givenn{\middle|}#2}\right]}
\newcommand{\expect}[2][]{\mathbb{E}\ifthenelse{\not\equal{}{#1}}{_{#1}}{}\!\left[{\def\givenn{\middle|}#2}\right]}

\newcommand{\tparen}{\big}
\newcommand{\tprob}[2][]{\text{Pr}\ifthenelse{\not\equal{}{#1}}{_{#1}}{}\tparen[{\def\given{\tparen|}#2}\tparen]}
\newcommand{\texpect}[2][]{\mathbb{E}\ifthenelse{\not\equal{}{#1}}{_{#1}}{}\tparen[{\def\given{\tparen|}#2}\tparen]}

\newcommand{\sprob}[2][]{\text{Pr}\ifthenelse{\not\equal{}{#1}}{_{#1}}{}[#2]}
\newcommand{\sexpect}[2][]{\mathbb{E}\ifthenelse{\not\equal{}{#1}}{_{#1}}{}[#2]}

\newcommand{\indicator}[1]{{\mathbbm{1}\left\{ #1 \right\}}}

\newcommand{\revcolor}[1]{{#1}}
\begin{document}

\TITLE{Two-stage Stochastic Matching and Pricing with Applications to Ride Hailing}

\RUNAUTHOR{Feng, Niazadeh, Saberi}

\RUNTITLE{Two-stage Stochastic Matching and Pricing with Applications to Ride Hailing}

\ARTICLEAUTHORS{%
\AUTHOR{Yiding Feng}
\AFF{Microsoft Research New England, \EMAIL{yidingfeng@microsoft.com}}
\AUTHOR{Rad Niazadeh}
\AFF{University of Chicago Booth School of Business, Chicago, IL, \EMAIL{rad.niazadeh@chicagobooth.edu}}
\AUTHOR{Amin Saberi}
\AFF{Management Science and Engineering, Stanford University, Stanford, CA, \EMAIL{saberi@stanford.edu}}
}

\ABSTRACT{
Matching and pricing are two critical levers in two-sided marketplaces to connect demand and supply. The platform can produce more efficient matching and pricing decisions by batching the demand requests. We initiate the study of  the two-stage stochastic matching problem, with or without pricing, to enable the platform to make improved decisions in a batch with an eye toward the imminent future demand requests. This problem is motivated in part by applications in online marketplaces such as ride hailing platforms.

We design online competitive algorithms for vertex-weighted (or unweighted) two-stage stochastic matching for maximizing supply efficiency, and two-stage joint matching and pricing for maximizing market efficiency. In the former problem, using a randomized primal-dual algorithm applied to a family of ``balancing'' convex programs, we obtain the optimal $3/4$ competitive ratio against the optimum offline benchmark. Using a factor revealing program and connections to submodular optimization, we improve this ratio against the optimum online benchmark to $(1-1/e+1/e^2)\approx 0.767$ for the unweighted and $0.761$ for the weighted case. In the latter problem, we design optimal $1/2$-competitive joint pricing and matching algorithm by borrowing ideas from the ex-ante prophet inequality literature. We also show an improved $(1-1/e)$-competitive algorithm for the special case of demand efficiency objective using the correlation gap of submodular functions. Finally, we complement our theoretical study by using DiDi's ride-sharing dataset for Chengdu city and numerically evaluating the performance of our proposed algorithms in practical instances of this problem.\footnote{\footnotesize{\emph{A preliminary conference version of this work has appeared in the proceeding of the ACM-SIAM Symposium on Discrete Algorithm (SODA'21)~\citep{feng2021two}, which contains only parts of this paper. In the current paper, all of the results in Section~5, Section~6, Appendix~C, and Appendix~D are new and there are several additional insights throughout. Moreover, the current paper presents all the proofs and technical details, most of which were missing in the early conference version.}}}}

\maketitle

\setcitestyle{authoryear}

\let\theHalgorithm=\thealgorithm

\vspace{-4mm}
 \section{Introduction}
 \defcitealias{MSVV-07}{MSVV-07}

The recent growth of two-sided online marketplaces for allocating advertisements, rides, or other goods and services   has led to new interest in matching algorithms and enriched the field with exciting and challenging problems. This is mainly due to the real-time nature of these marketplaces which requires the planner to make matching decisions between the current demand and supply with limited or uncertain information about the future. This dynamic aspect of the matching decision and the uncertainty of the future  are modeled in different ways. One prominent modeling approach in a two-sided matching platform is using online bipartite matching~(\citealp{KVV-90}; \citealp{DJK-13}) or online stochastic bipartite matching (\citealp{FMMM-09}; \citealp{MGS-12}; \citealp{JL-14}). This way of modeling the problem captures scenarios when the matching platform has to match arriving demand to the available supply immediately. This is particularly relevant in the context of online advertising~(\citealp{MSVV-07}; \citealp{BJN-07}; \citealp{FMMS-10}), where arriving search keywords should be matched to available advertisers---with almost no delay.

In applications that allow for some latency, the platform can produce a more efficient matching by accumulating requests and making decisions in a \emph{batch}. For example, 
ride hailing and ridesharing platforms such as Uber~\citep{Uber}, Lyft~\citep{Lyft} and DiDi~\citep{ZHMWZFGY-17} report substantial improvements when shifting from match-as-you-go algorithms to batching, improving the efficiency of the matching. Despite the effectiveness and popularity of the batching paradigm in practice, quantifying its effectiveness from a theoretical lens---and in particular in the context of matching platforms facing future demand uncertainty---is less studied.

The goal of this paper is to extend the batching framework in order to enable the platform to make improved matching decisions in a batch, with an eye toward the immediate future. This is particularly useful when the planner has side information about the compatibility of the imminently arriving demand (e.g., the riders who opened their application and are about to make a ride request) to different available supplies, or has historical information about their demand distributions, maybe in the form of samples from the data. We are also interested in scenarios where the platform not only has an eye toward the immediate future, but also has a handle through pricing, to control the uncertain demand in that future period. We then aim to extend our matching decision making framework to also allow the platform to make improved matching decisions in a batch, jointly with improved pricing decisions for the imminently arriving demand.

As a canonical model to mathematically capture the above scenarios, we consider this two-stage stochastic optimization problem for matching the demand vertices $\Rider$ to the supply vertices $\Driver$:

\mybox{
\setlength{\itemsep}{2pt}
  \setlength{\parskip}{2pt}
  \setlength{\parsep}{2pt}
  \openup 0.6em
\begin{displayquote}
\emph{Two-stage stochastic matching:} We are given a bipartite graph $G=(\Rider,\Driver,E)$,  in which vertices in $\Rider$ are divided into two sets, $\Ra$ and $\Rp$. In the first stage, the algorithm chooses a matching $M_1$ between $\Ra$  and $S$. In the second stage, each vertex $i$ in $\Rp$ arrives independently at random with probability $\pi_i$. The algorithm can choose a matching $M_2$ between the unmatched vertices of $S$ and a subset of $\Rp$ that have arrived. The goal is to maximize $| M_1 \cup M_2 |$ (in expectation).
\end{displayquote}}
\noindent Our techniques apply to a more general version of the problem in which every supply vertex $j \in S$ has a  weight $w_j$ and these weights can be different. The new objective function is to maximize (the expected) sum of the weights of the vertices in $S$ matched by either $M_1$ or $M_2$. 
Naturally, we refer to this problem as the \emph{vertex-weighted two-stage stochastic matching} problem.
\vspace{-2mm}
\subsection{Connections to Ride Hailing Matching and Pricing}
\label{intro:ride hailing}
\noindent The reader may be interested in the formulation of the problem in the context of ride hailing. In that setting, the supply set $\Driver$ represents the set of available drivers in a neighborhood. The first stage demand set $\Ra$ represents the set of riders who have made a request for a ride and need to be matched to a driver. The second stage demand set $\Rp$ represents the set of riders who are entering their destinations in the application and are about to receive a price quote and may or may not make a request. Probability $\pi_i$ captures the rate at which each rider in $\Rp$ will be available for matching in the second stage. \revcolor{The edge set $E$ determines the compatibility of the riders and drivers, most commonly based on their distance and occasionally other factors, such as riders and drivers review scores. Note that in this context, the decision-maker (i.e., the platform) knows the graph in advance, meaning it knows the drivers that will be compatible to each second stage rider, but it does not know whether each second stage rider decides to make a ride request or not.}
\subsubsection{Application to Matching}

In real-world ride hailing applications, the number of available drivers is normally large enough to serve all the riders in $\Ra$; however, the platform still must take into consideration the combination of a long list of objectives to use the supply efficiently. This includes the probability that the supply accepts the ride request, the number of ride requests fulfilled by a nearby supply, or the total traveling distance, among others. The unweighted version of our problem can be seen as a simplification of this objective function in which the goal is to maximize the number of ride requests matched to a compatible supply within a given distance. The vertex weights in our model can be used to increase the priority of the supply nodes who have been idle longer, or to implement policies that prioritize supply nodes with higher ratings or elite status.  Now, the goal is to pick matchings $M_1$ and $M_2$ to maximize the total weight of matched drivers at the end of the second stage---an objective function which we refer to as \emph{supply efficiency}.\footnote{Here we implicitly assume prices are exogenous and known, therefore the matching algorithm knows the rider availability rates $\pi_i$ for each rider in $\Rp$.}





As a remark, one can also consider the edge-weighted version of this problem, in which the weight of an edge captures the degree of compatibility of the two nodes (e.g., as a function of their distance).  This case is uninteresting, at least 
theoretically; the simple algorithm that tosses a fair coin to optimize either for $M_1$ or $M_2$ gets at least $\frac{1}{2}$  of the total weight of the optimum offline and that is the best possible ratio that an algorithm can hope for (see Appendix~\ref{app:general-weight} for details).

\subsubsection{Application to Joint Matching and Pricing} Also important is the price lever available to the platform that can adjust the probability $\pi_i$ of being available for matching in the second stage for vertices $i \in \Rp$. In this scenario, the platform offers a personalized price to each vertex $i \in \Rp$ in the first stage. We consider a Bayesian model where each rider $i \in \Rider$ has an independent value drawn from a known prior distribution for the ride.  Rider $i$ then accepts her price quote -- and hence will be available for matching in the second stage -- if and only if her value for the ride exceeds the price. Therefore, the price offered to rider $i\in \Rp$ determines the probability $\pi_i$ that she is available for matching in the second stage. In the first 
stage, the platform chooses a matching between $\Ra$ and $\Driver$ and offers prices to riders in $\Rp$. In the second stage, it chooses a matching between the unmatched vertices in $\Driver$ and those vertices in $\Rp$ who have accepted the prices. The goal is to maximize a particular objective function through this process. 

The choice of the objective function for this matching/pricing problem requires some care. If we only aim to maximize supply efficiency, the best choice of prices is all zero; this is simply because with increasing prices there will be less riders accepting their price offers to be matched to the drivers. However, from the eye of a market designer who pays attention to both sides of the market, non-zero prices will help with increasing the \emph{demand efficiency}---i.e., increasing total valuation of matched riders or equivalently, their social welfare. At first glance, it might seem intuitive that setting prices to zero will also help with increasing the demand efficiency. This is indeed not the case, as the platform only gets to observe whether riders accepted or declined their price offers, and not their exact valuations. Therefore, non-zero prices will help with demand discovery and filtering out low value buyers from getting matched in the market, and hence can increase the demand efficiency. See Appendix~\ref{apx:onestagepricing} for details of how pricing can help with the demand efficiency of the matching. 

To capture the above trade-off between high prices for demand discovery and low prices for supply efficiency,  we consider maximizing an objective function which we refer to as \emph{market efficiency}. This objective function includes both the expected total weights of the matched drivers in $D$ (i.e., the supply efficiency objective) and the expected social welfare/total values of the matched riders in both $\Ra$ and $\Rp$, given the information of the platform regarding riders in $\Ra$ and $\Rp$ (i.e., the demand efficiency objective). See \Cref{sec:two-stage-matching-pricing-def} for a formal definition.



\subsection{Technical Contributions}
We initiate the study of both two-stage stochastic matching problem and two-stage joint matching and pricing problem with a focus on the worst case competitive analysis of these problems. For measuring the competitive ratio, we consider the  optimum solution in hindsight, which we refer to as the \emph{optimum offline}, and the solution of the computationally-unbounded optimum algorithm, which we refer to as the \emph{optimum online}. We now summarize our main contributions.

\vspace{2mm}
\noindent\textbf{Optimal competitive two-stage matching vs.\ optimum offline:} In \Cref{sec:adversarial}, we  present polynomial-time algorithms with the competitive ratio of $\frac{3}{4}$  compared to the optimum offline for both weighted and unweighted versions of the problem and show that this factor is optimum for both cases of two-stage stochastic matching. Somewhat interestingly, our optimum competitive algorithms do not use any information about the second stage graph (including the edges and availability probabilities $\pi_i$'s). In fact, we show a stronger result and prove that our algorithms achieve the competitive ratio of $\frac{3}{4}$ even when the second stage graph is picked by an oblivious adversary. 

\vspace{1mm}
\textbf{\tikz\draw[black,fill=black] (0,0) circle (.5ex);~Convex programming based weighted-balanced-utilization:} The main idea of both algorithms is to find a balanced allocation of demand $\Ra$ to supply $\Driver$, while also taking the priority weights into consideration in a principled way. We identify a family of convex programs to characterize such a balanced randomized allocation in the first stage---not necessarily of maximum total weight---that hedges against the uncertainty of the  demand in the second stage.  Using the strong duality of the convex programs, we introduce a primal-dual framework that bounds the competitive ratio of randomized integral matchings sampled from the optimal solution of these convex programs.

\vspace{1mm}
\textbf{\tikz\draw[black,fill=black] (0,0) circle (.5ex);~Primal-dual using the convex programming based graph decomposition:} The convex program family we consider also offers a decomposition of the first stage graph into pairs of supply and demand nodes, which in the special unweighted case coincides with the well-studied \emph{matching skeleton} introduced by~\cite{GKK-12}. Notably, this matching skeleton graph decomposition is closely related to the Edmonds-Gallai decomposition~\citep{edm-65,gal-64} of bipartite graphs. In this sense, we generalize the concept of matching skeleton to vertex weighted graphs as a byproduct of our primal-dual approach, which might be of independent interest. We should note that our convex programming based decomposition is constructed using strong duality and incorporating the KKT conditions of the convex program. As a result, it plays a critical role in setting up our improved primal-dual framework for this problem in a way that \emph{diverges from the standard primal-dual framework} in the literature on online bipartite allocations (e.g., \cite{MSVV-07,BJN-07,DJK-13,GNR-14,GU-19,FNS-19,MS-20}). Note that the competitive ratio $(1-1/e)$ is known for the fully online version of the problem, thanks to the classic work of Karp, Vazirani and Vazriani~\citeyearpar{KVV-90} and its extension to vertex-weighted online bipartite matching in \cite{AGKM-11}.
The new algorithmic construct of using convex programming based decomposition in the primal-dual analysis is the key to improve this competitive ratio from $(1-1/e)$ to  $\frac{3}{4}$, which is optimum for our two-stage problem. 


\vspace{2mm}
\noindent\textbf{Complexity of optimum online and an improved competitive algorithm:} In \Cref{sec:hardness-online} and \Cref{sec:stochastic}, we turn our attention to approximating the optimum online. In that case, the problem is in the realm of approximation algorithms and related to  maximizing the sum of monotone submodular and negative linear functions subject to matroid constraints~\citep{CCPV-11,SVW-17}. This connection offers a way to prove that the problem does not admit a Fully Polynomial Time Approximation Scheme (FPTAS) unless $\text{P}=\text{NP}$. Furthermore, it suggests an alternative matching algorithm for the first stage that relies on knowing the second stage graph and availability probabilities $\pi_i$'s. For our first-stage matching, we then output the better (in expectation) of the matchings computed by this alternative algorithm and the $\frac{3}{4}$-competitive algorithm by sampling/simulating the second-stage stochasticity---basically through Monte-Carlo simulation. Using a novel factor revealing program, we show that this approach results in improved competitive ratios of $(1-1/e+1/e^2)\approx 0.767$ and $0.761$ for the unweighted and weighted versions, respectively. 

Perhaps, an important takeaway from the comparison between optimum offline and online as benchmarks is not so much that one can get a better approximation ratio against the latter, but that it reveals so much more of the problem's underlying structure. The $\frac{3}{4}$-competitive algorithm exploits the graph in the first stage effectively through the convex program, but ignores the vertices arriving in the second stage or their probabilities. Nevertheless, it gets the optimum competitive ratio. Only when we switch to the optimum online as a more precise benchmark, can we exploit both the stochastic information and the combinatorial information given to us about the second stage, which in turn help us to design an improved algorithm. On a related note, our improved competitive ratios against optimum online---combined with the tightness of $\frac{3}{4}$-competitive ratio against optimum offline---establish a separation between these two benchmarks. Such a separation is rarely seen in online allocation problems (e.g., see \citealp{anari2019nearly}) and might be of independent interest. 

\vspace{2mm}
\noindent\textbf{Optimal competitive two-stage joint matching and pricing:} In \Cref{sec:matching and pricing}, we discuss the joint matching and pricing problem. There, it is crucial to use both supply efficiency and demand efficiency in the objective function. As mentioned earlier, an algorithm that optimizes only supply efficiency will give a price of zero to all pending riders to allocate as many of them as possible. On the other hand, maximizing demand efficiency will require the platform to price the scarce demand nodes accordingly. To combine these two objectives, we introduce an offline convex optimization which we refer to as \emph{ex ante relaxation}. The optimum solution of this offline program (i) suggests a vector of prices and a randomized matching that is only feasible in-expectation, and (ii) signifies its market efficiency is no less than the optimum offline. Obtaining a feasible two-stage joint matching and pricing policy from the solution of this convex program requires techniques in the prophet inequality literature, originated from the seminal work of \cite{KS-77} and \cite{sam-84}. In particular, we incorporate an algorithmic construct known as the \emph{online contention resolution} scheme~\citep{ala-14,FSZ-16}.  We obtain an optimal $\frac{1}{2}$-competitive algorithm for market efficiency by an adaptation of these techniques to our specific problem. We further show this is the best ratio possible for this objective function. \revcolor{Some of the technical ideas used in this part resemble other work that connect prophet inequalities to more sophisticated online Bayesian allocation problems in revenue management, e.g., see \cite{AHL-12,FNS-19,FNS-20,manshadi2020online,ezra2020online}.} As a remark, our bounds are with respect to the optimum offline benchmark. Whether it is possible to obtain an improved bound against optimum online benchmark is an open problem.

\vspace{2mm}
\noindent\textbf{Optimal competitive single-stage joint matching and pricing:} As a variation of our model, we consider the single-stage joint matching and pricing problem for the objective of demand efficiency in Appendix~\ref{apx:onestagepricing}. Here, we are only considering demand nodes in $\Rp$ and supply nodes in $\Driver$. The goal is to pick a vector of prices for demand nodes in $\Rp$, so that after matching the demand nodes who accept their prices to supply nodes, the generated social welfare for the demand nodes is maximized. By setting up a similar ex ante relaxation as for the market efficiency objective---but this time for only the demand efficiency objective---and incorporating standard tools in the submodular optimization literature such as correlation gap~\citep{ADSY-10,yan-11}, we obtain a $(1-1/e)$ competitive algorithm for demand efficiency.

\revcolor{
\vspace{2mm}
\noindent\textbf{Robustness to approximations and modeling assumptions:} Finally, we study the robustness of our results to various aspects of the two-stage stochastic matching problem in Appendix~\ref{apx:robustness}. In particular, we consider replacing the second stage maximum matching with an approximation algorithm (Appendix~\ref{apx:robustness-approx}) and adding structural assumptions on the first stage graph (Appendix~\ref{apx:robustness-first}).}

We have summarized our results and their comparison to the existing literature in \Cref{tab:summary}.
\vspace{-2mm}
\begin{table}[ht]
\smallskip
\small\addtolength{\tabcolsep}{-2pt}
 	\caption{Competitive ratios for supply efficiency (top) and market efficiency (bottom).}
 	\label{tab:summary}
 	\smallskip
 	\smallskip
\renewcommand\arraystretch{1.8}

\renewcommand{\thefootnote}{\ifcase\value{footnote}\or(\textasteriskcentered)\or(\textasteriskcentered\textasteriskcentered)\or(\textdagger)\or(\textdaggerdbl)\or(\textsection)\or(\textbardbl)\or(\textdollar)\or(\textdollar\textdollar)\or(\pounds)\or(\texteuro)\or(\#\#\#)\or(\#\#\#\#)\or($\infty$)\fi}
 	\small
\centering
 
 	\begin{tabular}{|c|c|c|c|c|c|}
\hhline{~~~---} \multicolumn{3}{c|}{\cellcolor{white!25}}  & \shortstack{\\\footnotesize{Lower-bound}\\ \footnotesize{(previous work)}}& \shortstack{\\\footnotesize{Upper-bound}\\ \footnotesize{(this paper)}}& \shortstack{\\\footnotesize{Lower-bound}\\ \footnotesize{(this paper)}}

 		\tabularnewline\hhline{------}  \multirow{4}{*}{\shortstack{\\ Two-stage stochastic matching \\(\Cref{sec:matching})}} &\multirow{2}{*}{\shortstack{\footnotesize{Optimum offline}}} &\shortstack{\footnotesize{Weighted}} &$(1-\tfrac{1}{e})$~\footnotemark[1] &\shortstack{$\tfrac{3}{4}$}& \shortstack{$\tfrac{3}{4}$}
 			\tabularnewline\hhline{~~----}  & &  \shortstack{\footnotesize{Unweighted}}&$\tfrac{2}{3}$~\footnotemark[2] &\shortstack{$\tfrac{3}{4}$}& \shortstack{$\tfrac{3}{4}$}
 		\tabularnewline\hhline{~-----}  & \multirow{2}{*}{\shortstack{\footnotesize{Optimum online}}} & \shortstack{\footnotesize{Weighted}} &$(1-\tfrac{1}{e})~\footnotemark[1] $ &\shortstack{\footnotesize{FPTAS-hard}}& \shortstack{$\left(1-\tfrac{1}{e}+\tfrac{1}{e^2}\right)\approx 0.767$}
 		\tabularnewline\hhline{~~----}  & & \shortstack{\footnotesize{Unweighted}} &$\tfrac{2}{3}$~\footnotemark[2] &\shortstack{\footnotesize{FPTAS-hard}}& \shortstack{$0.7613$}
 		\tabularnewline\hhline{------}\multicolumn{6}{c}{\cellcolor{white!25}}
 			\tabularnewline\hhline{------}\multicolumn{3}{|c|}{\shortstack{\\Two-stage stochastic joint matching/pricing\\for market efficiency (\Cref{sec:matching and pricing})}}&\shortstack{\\-\\~}&\shortstack{\\$\tfrac{1}{2}$\\~}&\shortstack{\\$\tfrac{1}{2}$\\~}
 			\tabularnewline\hhline{------}\multicolumn{3}{|c|}{\shortstack{\\Single-stage stochastic joint matching/pricing\\for demand efficiency (Appendix~\ref{apx:onestagepricing})}}&\shortstack{\\-\\~}&\shortstack{\\$1-\tfrac{1}{e}$\\(ex-ante)}&\shortstack{\\$1-\tfrac{1}{e}$\\~}
 			\tabularnewline\hhline{------}\multicolumn{3}{|c|}{\shortstack{\\Two-stage stochastic matching\\with general edge weights (Appendix~\ref{app:general-weight})}}&\shortstack{\\-\\~}&\shortstack{\\$\tfrac{1}{2}$~\\~}&\shortstack{\\$\tfrac{1}{2}$~\\~}
 		\\
 		\hhline{------}
 \end{tabular}
 	\smallskip
 	
\noindent\par
 	\begin{minipage}{0.7\textwidth}
 		{
 			\center
 			\footnotesize
 			\footnotemark[1]~{\cite{AGKM-11,DJK-13}};~ 
 			\footnotemark[2]~{\cite{LS-17}};
 		}
 	\end{minipage}
 \end{table}
 \vspace{-8mm}
 \subsection{Data-driven numerical simulations.} To complement our theoretical study of the two-stage stochastic matching problem, we use a ride-sharing dataset from \citet{didi-20} in \Cref{sec:numerical}. This rich dataset contains anonymized trajectory data and ride request data of DiDi Express and DiDi Premier drivers within the “Second Ring Road” of city of Chengdu in China during November 2016. We run numerical simulations to evaluate the performance of our algorithms in practical instances of our problem generated from this data. We evaluate the performance of 
other policies too, and compare
their performance ratios against 
optimum offline using Monte Carlo simulation. See \Cref{sec:numerical} for details. In a nutshell, we conclude that our optimal competitive algorithm against optimum offline outperforms the myopic greedy algorithm,
where both policies require no knowledge
of the second stage to pick their first stage matchings.
We also observe that by utilizing the distributional knowledge 
of the second stage in our improved competitive algorithm against optimum online,
it further outperforms 
both our previous algorithm and the myopic greedy algorithm. Notably, our algorithms obtain near-optimal performances in our experiments. 

\revcolor{
Besides the related literature that we have discussed so far, our work connects to the vast literature on online matching and revenue management. We provide an overview of further related work in Appendix~\ref{appendix-further-related}. We should mention, since the appearance of an early online version of our paper, there has been a few follow up work. Most importantly,  \cite{feng2020batching} extend our result for the two-stage (fractional) matching under adversarial arrival to $K\in\mathbb{N}$ stages using the idea of convex programming-based matching, showing an optimal competitive ratio of $\Gamma(K)=1-(1-\frac{1}{K})^K$ by carefully picking the convex program at each stage. Interestingly, as it has been shown in this paper, there is a uniquely identifiable sequence of polynomials of decreasing degrees, one for each stage, that can be used to regularize the maximum matching linear program to obtain this result. We should add that in our paper we basically consider both two-stage adversarial and stochastic arrivals (to compare our algorithms to optimum online as well), and we focus on the more difficult case of integral matchings.}

 \section{Preliminaries}
 \newcommand{\alg}{\texttt{ALG}}

Suppose we have a bipartite graph $G=(\Rider,\Driver,E)$,  where $\Rider$ denotes the set of demands and $\Driver$ denotes the set of supplies. An edge $(i,j)\in E$ indicates that demand vertex $i\in \Rider$ can be matched to supply vertex $j\in \Driver$. The demand set is partitioned into $\Rider=\Ra\cup\Rp$, where $\Ra$ is the set of first stage demands and $\Rp$ is the set of potential second stage demands. Each demand vertex $i\in\Rider$ has a private value $\val_i\sim\Fu$ for the service, drawn independently from their valuation distribution $\Fu$. We assume valuation distributions $\{\Fu\}_{i\in \Rider}$ are common knowledge.

Each demand vertex $i\in\Rp$ is offered a price $p_i$, and accepts the price offer if their value is at least the price, i.e., $\val_i\geq p_i$. As a result, each demand vertex $i\in\Rp$ will be available in the second stage with probability $\pi_i\triangleq \prob{\val_i\geq p_i}$ and independently from other vertices in $\Rp$. We use $\Rpa\subseteq \Rp$ to denote the (stochastic) subset of second stage available demand vertices. See \Cref{fig:two-stage} for a demonstration.

\begin{figure}[htb]
\centering
\includegraphics[width=0.85\textwidth,trim={2cm 0cm 5cm 0cm},clip]{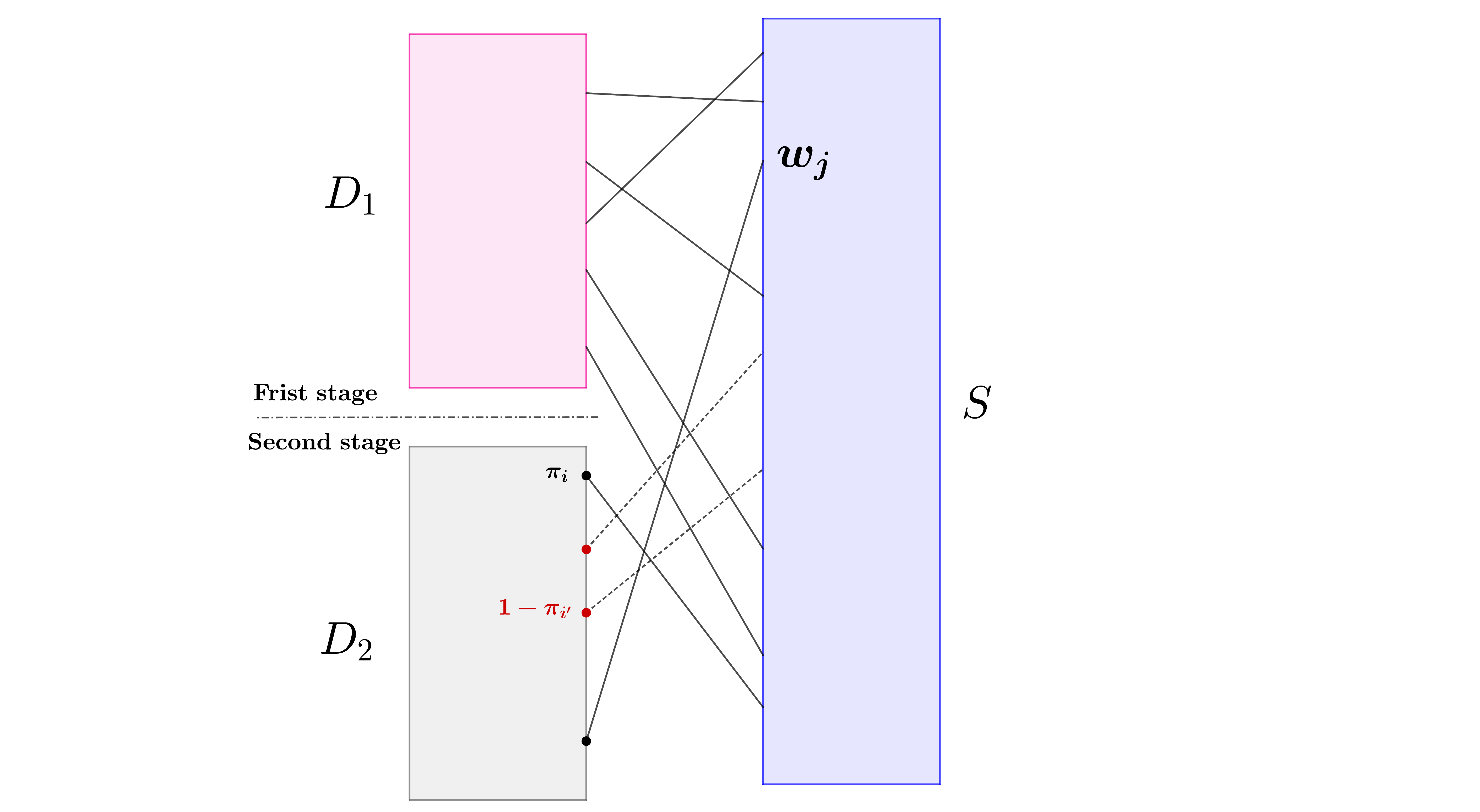}
\qquad
\caption{\label{fig:two-stage}Bipartite graph $G=(\Ra\cup\Rp,\Driver)$ in the (vertex) weighted two-stage stochastic matching problem}
\end{figure}

\begin{definition}
 \label{def:matching}
 Given probabilities $\{\pi_i\}_{i\in\Rp}$, the \emph{matching process} is a two-stage stochastic process, defined as following:
\begin{enumerate}[leftmargin=15pt]
\item[\Rom{1}.]\emph{First stage:}~the platform picks a matching $M_1$ between demand vertices $\Ra$ and supply vertices $\Driver$. 
\item[\Rom{2}.]~\emph{Second stage:}~each vertex $i \in \Rp$ flips an independent coin with success probability $\pi_i$ to determine its availability in this stage. After $\Rpa$ is formed (defined above), the platform picks a matching $M_2$ between $\Rpa$ and remaining unmatched supply vertices from the first stage.
\end{enumerate}

\end{definition}


  
  We will mainly consider the \emph{(vertex) weighted two-stage stochastic matching problem} and its variant, the \emph{two-stage stochastic joint matching and pricing problem}. We describe these two problems below.
\subsubsection*{(i) (Vertex) Weighted two-stage stochastic matching}

Suppose each supply vertex $j\in \Driver$ has a non-negative weight $w_j\in\RP$. This weight indicates the priority level of matching this supply to a demand. For example---in the practice of ride hailing---this weight can depend on driver's location and availability of supply in that area, average utilization so far in the day, or elite status. We assume these weights are given parameters to the platform. We also assume the probabilities $\{\pi_i\}_{i\in \Rp}$ are known and given parameters to the platform. This is a valid assumption if the prices are determined by a fixed formula (e.g., due to regulations) or if they are computed using a different algorithm. Also, this assumption holds in a slightly different model where these probabilities are not determined by prices, and instead are known parameters modeling the demand uncertainty in the second stage. All of our results for the (vertex) weighted two-stage stochastic matching  can be applied to this model too.

In this problem, the goal is to maximize the supply efficiency, that is, the expected total weight of supplies $\Driver$ matched in the matching process of \Cref{def:matching}, where the expectation is over the internal randomness of the matching algorithm and the uncertainty in demand vertices of $\Rp$ being available in the second stage. When $w_j=1$ for all $j\in \Driver$, this objective is the cardinality of the final matching.


\begin{definition}
\label{def:competitive-ratio}
The \emph{competitive ratio} of algorithm $\alg$ against a benchmark $\opt$ in the (vertex) weighted two-stage stochastic matching problem is defined as:
$$
\compratioM(\alg,\opt)\triangleq\underset{G,\{w_j\}_{\driver\in\Driver},\{\pi_i\}_{i\in\Rp}}{\min}\frac{\expect{\sum_{j\in \Driver}w_j\mathbbm{1}\{\textrm{$j$ is matched in $M_1\cup M_2$}\}}}{\opt}~,
$$
where $M_1$ and $M2$ are the matchings picked by $\alg$ in the first and second stage, respectively. For the special unweighted case, we define $\compratioMUW(\cdot,\cdot)$ similarly by setting $w_j=1$ for all $j\in\Driver$.
\end{definition}
After choosing a matching $M_1$ and observing $\Rpa$ in the first stage, in the second stage the platform should always pick a maximum supply-weighted matching $M_2$ between $\Rpa$ and unmatched supplies left by $M_1$. So,  the only question is how it should pick the first stage matching $M_1$.

\paragraph{Benchmarks.} We consider competitive ratios against two candidate benchmarks as $\opt$ for this problem, defined formally below.
\begin{enumerate}
    \item \emph{Optimum offline ($\optoff$)}:  an omniscient offline policy that knows the exact realization of available second stage demands $\Rpa$. This offline policy picks the maximum supply-weighted matching between $\Ra\cup\Rpa$ and $\Driver$. We use $\optoff$ to denote its expected total weight, i.e.,
    \begin{equation}
    \label{eq:opt-off}
    \optoff\triangleq \expect[{\Rpa}]{\maxmatch(\Ra\cup\Rpa,\Driver)},
    \end{equation}
    where $\maxmatch(X,Y)$ denotes the weight of the supply-weighted maximum matching between $X\subseteq \Rider$ and $Y\subseteq \Driver$ given weights $\{w_j\}_{j\in \Driver}$.
    \item \emph{Optimum online ($\opton$)}: an algorithm that searches over all possible matchings $M_1$ between $\Ra$ and $\Driver$, and picks the one that maximizes the expected total weight of supplies matched in the matching process of \Cref{def:matching}. We use $\opton$ to denote the expected total weight of matched supply vertices by this algorithm, i.e.,
    \begin{equation}
    \label{eq:opt-online-def-1}
    \opton\triangleq \underset{M_1}{\max}\left(\sum_{j\in \Driver}{w_j\mathbbm{1}\{\textrm{$j$ is matched in $M_1$}\}}+ \expect[{\Rpa}]{\maxmatch(\Rpa,\Driver\setminus \Driver_1)}\right),
    \end{equation}
    where $\Driver_1$ denotes the set of supply vertices matched by $M_1$.
    
\end{enumerate}


\subsubsection*{(ii) Two-stage stochastic joint matching and pricing}
\label{sec:two-stage-matching-pricing-def}

The setup of this problem is similar to two-stage stochastic matching, with two important modifications: (i)~the platform not only picks matching $M_1$ in the first stage, but it also selects prices $\{p_i\}_{i\in\Rp}$ jointly with $M_1$ for potential demand vertices in $\Rp$, and  (ii)~the platform's goal is to maintain market efficiency, i.e., the expected total value (welfare) of matched demand vertices plus the total weight of matched supply vertices, where the expectation is over the randomness in the algorithm and the valuations of demand vertices.

To define market efficiency mathematically, a few extra definitions are in order. For each demand vertex $i$ and price $p$, define $w_i(p)\triangleq \expect[\val\sim \Fu]{\val|\val\geq p}$. Note that for $i\in\Ra$, as the platform knows this demand vertex has already accepted her price $p_i$, $w_i(p_i)$ is considered to be her contribution to the objective if gets matched by $M_1$, and zero otherwise. This is in contrast to $i\in\Rp$, where her contribution is her actual value $\val_i$ if she accepts her price \emph{and} further gets matched by $M_2$, and zero otherwise.
\begin{definition}
The competitive ratio of algorithm $\alg$ against a benchmark $\opt$ in the two-stage stochastic joint matching and pricing problem is defined as
\begin{align*}
\compratioMP(\alg,\opt)&\triangleq\underset{\substack{G,\{w_j\}_{j\in \Driver},\\ \{\Fu\}_{i\in \Rider}, \{p_i\}_{i\in\Ra}}}{\min}\frac{\expect{\objr+\objd}}{\opt}\\
\objr&\triangleq \sum_{i\in\Ra}{w_i(p_i)\mathbbm{1}\{\textrm{$i$ is matched by $M_1$}\}}+\sum_{i\in\Rp}{\val_i\mathbbm{1}\{\textrm{$\val_i\geq p_i$ and $i$ is matched by $M_2$}\}}\\
\objd&\triangleq \sum_{j\in \Driver}w_j\mathbbm{1}\{\textrm{$j$ is matched in $M_1\cup M_2$}\}
\end{align*}
where $M_1$ and $M_2$ are the matchings picked by $\alg$ in the first and second stage, respectively.
\end{definition}
\paragraph{Benchmarks.} We consider the competitive ratio against the optimum offline benchmark for this problem ($\optoffmp$). This benchmark is essentially an omniscient policy that knows the exact realization of valuations $\{\val_i\}_{i\in\Rp}$, and given this information, selects the maximum vertex-weighted matching between $\Ra\cup\Rp$ and $\Driver$ by considering weights $\hat{w}_i=w_i(p_i)$ for $i\in\Ra$, $\hat{w}_i=\val_i$ for $i\in\Rp$ and $\hat{w}_j=w_j$ for $j\in \Driver$ in order to maximize market efficiency. Formally,
 \begin{equation}
    \label{eq:opt-off-mp}
    \displaystyle\optoffmp\triangleq \expect[\val_i\sim \Fu,i\in\Rp]{\maxmatchboth(\Rider,\Driver)},
    \end{equation}
where $\maxmatchboth(\Rider,\Driver)$ denotes the weight of the vertex-weighted maximum matching in $G$ given weights $\{\hat{w}_i\}_{i\in \Rider\cup \Driver}$ as described above.

 \section{Optimal Competitive Weighted Two-stage Stochastic Matching}
\label{sec:matching}

 
We start by designing an optimal competitive algorithm against optimum offline. The main idea behind this algorithm is to pick a (randomized) matching that distributes the demand vertices $\Ra$ as \emph{balanced} as possible among the supply vertices $\Driver$. This matchings hedges against the second stage uncertainty.

In \Cref{sec:adversarial}, we propose a family of convex programs that  produce fractional matchings for the first stage, according to a certain notion of balancedness that we will define later. Then, we introduce a primal-dual approach to bound the competitive ratio of randomized integral matchings sampled from the optimal solution of these convex programs. Our approach suggests decomposing the first-stage graph into particular parts, and picking a specific fractional matching in each part as our sampling probabilities. We dig deeper into this decomposition in Appendix~\ref{sec:skeleton}. 




\subsection{A \texorpdfstring{$\mathbf{\frac{3}{4}}$}{}-Competitive algorithm}
\label{sec:adversarial}
First, it is easy to observe that no algorithm can obtain a worst-case competitive ratio better than $\tfrac{3}{4}$ against optimum offline. See Appendix~\ref{apx:optoffline-upper} for the proof of this proposition.

\begin{restatable}{proposition}{thmoptofflineupper}
\label{prop:optoffline-upper}
In the two-stage stochastic matching for maximizing supply efficiency, no policy obtains a competitive ratio better than $\tfrac{3}{4}$ against $\optoff$, even when all the weights are equal.
\end{restatable}

Next, we present an algorithm that obtains this competitive ratio. In fact, we show an even stronger result, where our algorithm does not need to know anything about the graph structure between $\Rp$ and $\Driver$, and still obtains the optimal competitive ratio of $\tfrac{3}{4}$ against $\optoff$ in an adversarial scenario. In such a scenario, the second stage graph can be picked arbitrarily by an (oblivious) adversary and not necessarily through our specific stochastic process.

\subsubsection{Convex programming for  balanced supply utilization}
\label{sec:skeleton-weighted}
\newcommand{\xbf}{\mathbf{x}}
\newcommand{\ybf}{\mathbf{y}}
\newcommand{\Lag}{\mathcal{L}}
\newcommand{\alphabf}{\boldsymbol{\alpha}}
\newcommand{\betabf}{\boldsymbol{\beta}}
\newcommand{\lambdabf}{\boldsymbol{\lambda}}
\newcommand{\thetabf}{\boldsymbol{\theta}}
\newcommand{\gammabf}{\boldsymbol{\gamma}}

Intuitively, our goal is to find a (fractional) matching in the first stage that makes a balanced allocation of supply to demand, while also taking supply weights into consideration. To combine these two criteria, and inspired by notions of fairness such as Rawlsian social welfare~\citep{raw-71} and Nash social welfare~\citep{nas-50}, we turn our attention to the following convex minimization program \ref{eq:matching-convex-program} to obtain a desired fractional matching $\{x^*_{ij}\}$ in $G[\Ra, \Driver]$ as our first stage solution:
\begin{align}
\label{eq:matching-convex-program}\tag{$\mathscr{P}^g$}
\begin{array}{llll}
	\{x^*_{ij}\}\in\argmin~~&
	\displaystyle\sum_{j\in \Driver}\frac{1}{w_j}g\left(w_j\left(1-\sum_{i\in N(j)}x_{ij}\right)\right)
	&\text{s.t.}&\\[1em]
	&\displaystyle\sum_{j\in N(i)}x_{ij}\leq 1 &
	 i\in \Ra~,&\\[1em]
	&\displaystyle\sum_{i\in N(j)}x_{ij}\leq 1 &
    j\in \Driver~,&\\[1em]
	&x_{ij}\geq 0&
	 i\in\Ra,j\in \Driver,(i,j)\in E,&
\end{array}
\end{align}
where $g(.):\R\rightarrow \R$ can be \emph{any} differentiable, increasing, and strictly convex function, and $N(i)$ denotes the neighbors of vertex $i$ in $G[\Ra, \Driver]$.
After computing $\{x^*_{ij}\}$, we sample the matching $M_1$ so that 
\begin{equation}
\label{eq:random-matching}
\prob{(i,j)\in M_1}=x^*_{ij},~\forall i\in\Ra,j\in \Driver,(i,j)\in E~.
\end{equation}
Note that because of the integrality of the bipartite matching polytope~\citep{sch-03}, the fractional matching $\{x^*_{ij}\}$ can be written down as a convex combination of integral bipartite matchings in $G[\Ra,\Driver]$. Moreover, this can be done in polynomial time by using standard algorithmic versions of the Carathéodory theorem, resulting in an efficient randomized rounding to sample $M_1$. Alternatively, faster dependent randomized rounding techniques for bipartite matching polytope, e.g., \citealt{gandhi2006dependent}, can be used in a blackbox fashion.

\begin{algorithm}
\caption{\textsc{Weighted-Balanced-Utilization}}
\label{alg:weighted-skeleton}
\begin{algorithmic}[1]
\State{\textbf{input:} bipartite graph $G=(\Rider,\Driver,E)$, non-negative weights $\{w_j\}_{j\in \Driver}$, convex function $g(\cdot)$~.}
\State{\textbf{output:} bipartite matching $M_1$ in $G[\Ra,\Driver]$,  bipartite matching $M_2$ in $G[\Rpa,\Driver]$.}
\vspace{2mm}
\State{Solve convex program \ref{eq:matching-convex-program} to obtain $\{x^*_{ij}\}$.}
\State{Sample matching $M_1$ with edge marginal probabilities $\{x^*_{ij}\}$.}
\State{In the second stage, return the maximum supply-weighted matching $M_2$ between $\Rpa$ and the remaining vertices of $\Driver$.}
\end{algorithmic}

\end{algorithm}

\begin{restatable}{theorem}{optofflinecompratio}
\label{thm:opt-offline-comp-ratio}
For any differentiable, monotone increasing and strictly convex function $g:\R\rightarrow \R$,
$$
\compratioM(\textrm{\Cref{alg:weighted-skeleton}},\opton)\geq
 \compratioM(\textrm{\Cref{alg:weighted-skeleton}},\optoff)\geq\tfrac{3}{4}.
$$
\end{restatable}
\revcolor{
\begin{remark}
\label{remark:uniqueness of g}
The first-stage marginal probabilities $\mathbf{x}^*$ in
\Cref{alg:weighted-skeleton}
do not depend on the choice of the convex function $g$. See Appendix~\ref{apx:identical solution for different}
for a detailed proof. From a computational point of view, however, we may prefer
strongly convex, 
smooth convex functions $g$,
since 
iterative methods for solving convex optimization (e.g., first-order methods)
have generally faster convergence rates
in that case.
\end{remark}}

To prove \Cref{thm:opt-offline-comp-ratio}, we start by identifying structural properties of the optimal solution of the convex program \ref{eq:matching-convex-program} in the form of a particular bipartite graph decomposition. See \Cref{fig:skeleton}.


\begin{figure}[htb]
\centering
\includegraphics[trim={0cm 5cm 0cm 1.5cm},clip,width=0.8\textwidth]{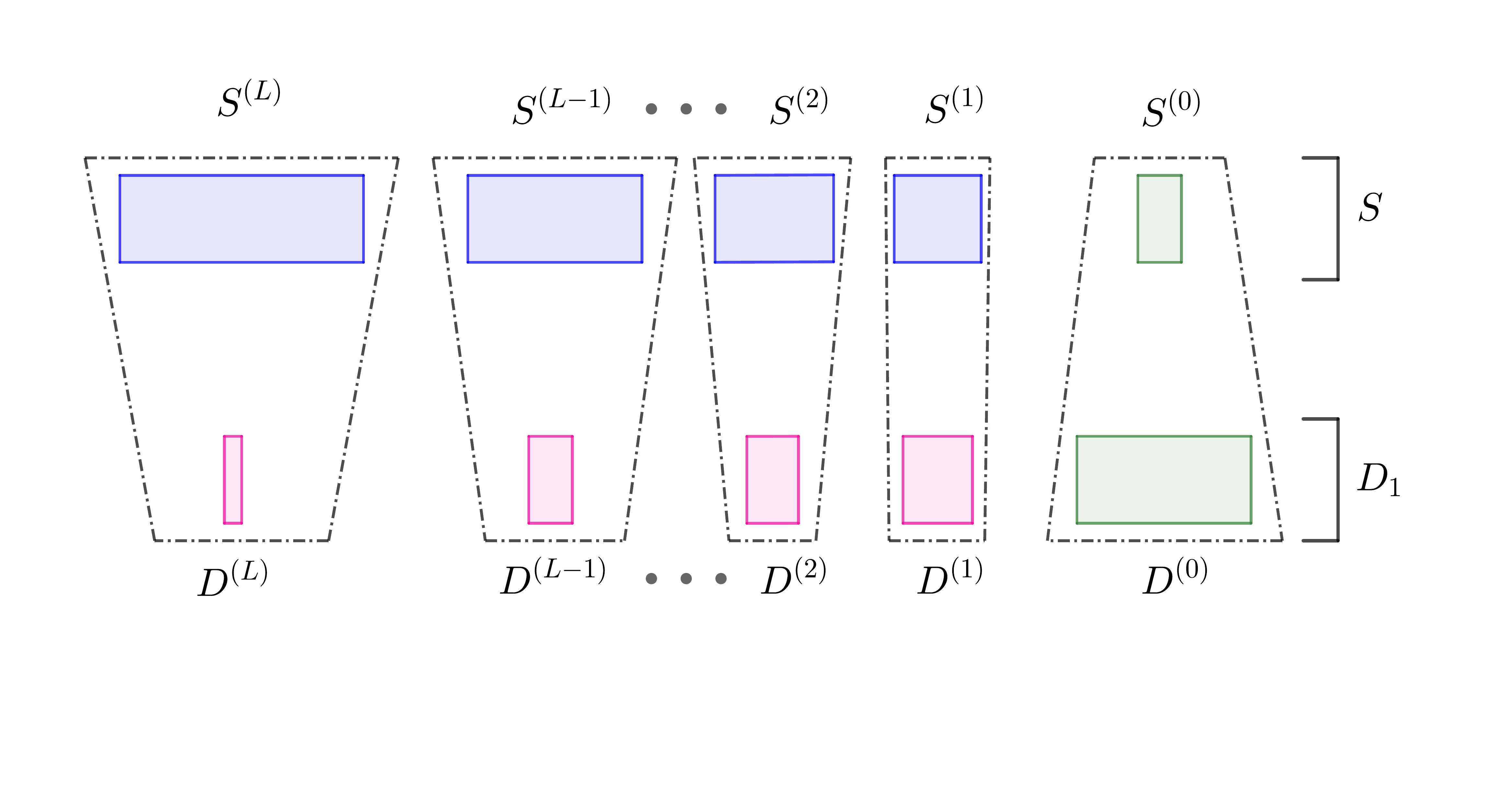}
\caption{Structural decomposition in \Cref{lemma:structure}; the pairs $(\Rider^{(l)},\Driver^{(l})$ are sorted such that $0=c^{(0)}< c^{(1)}<\ldots< c^{(L)}$. Red dashed-lines are forbidden \emph{non-existent} edges, i.e., any edge between $\Rider^{(l)}$ and $\Driver^{(l')}$ when $l<l'$.\label{fig:skeleton}}
\end{figure}
\vspace{-4mm}
\begin{restatable}[\emph{Structural decomposition}]{lemma}
{structural}
\label{lemma:structure}
Let $\xbf^*=\{x^*_{ij}\}$ be the optimal solution of the program \ref{eq:matching-convex-program}. Consider the subgraph $G'=(\Ra,\Driver,E')$ of $G$, where $E'=\{(i,j)\in \Ra\times \Driver:x^*_{ij}>0\}$. Let $\Driver^{(0)}$ be the set of vertices in $\Driver$ fully matched by $\xbf^*$,  and $\Rider^{(0)}$ be the set of demand vertices who are neighbors of $\Driver^{(0)}$ in $G'$, that is, 
$$\Driver^{(0)}\triangleq \{j\in \Driver: \sum_{i\in N(j)}x^*_{ij}=1\}~~\textrm{and}~~\Rider^{(0)}\triangleq \{i\in\Ra:\exists j\in \Driver^{(0)}, x^*_{ij}>0\}~.$$
Moreover, let the pairs $\{\left(\Rider^{(l)},\Driver^{(l)}\right)\}_{l=1}^{L}$ identify the $L\geq 1$ connected components of the induced subgraph $G'[\Ra\setminus\Rider^{(0)},\Driver\setminus\Driver^{(0)}]$ of $G'$. Then:
\begin{enumerate}[label=\roman*.]
    \item \underline{Uniformity}: ~$\forall~ l\in[0:L],j,j'\in \Driveri: w_j\left(1-\sum_{i\in N(j)}x^*_{ij}\right)=w_{j'}\left(1-\sum_{i\in N(j')}x^*_{ij'}\right)\triangleq c^{(l)}$~.
    \item \underline{Monotonicity}:  $\forall~ l,l'\in[0:L]$: there exists no edge in $E$ between $\Rider^{(l)}$ and $\Driver^{(l')}$ if $c^{(l)}<c^{(l')}$. 
    \item \underline{Saturation}: $\forall~ l\in[1:L]$: all vertices in $\Rider^{(l)}$ are fully matched by $\xbf^*$, i.e., $\sum_{j\in N(i)}x^*_{ij}=1, i\in\Rider^{(l)}$~.
\end{enumerate}
\end{restatable}
Note that if $c^{(l)}=c^{(l')}$ for $l\neq l'$, we can simply merge the two pairs $\left(\Rider^{(l)},\Driver^{(l)}\right)$ and $\left(\Rider^{(l')},\Driver^{(l')}\right)$ to $\left(\Rider^{(l)}\cup \Rider^{(l')},\Driver^{(l)}\cup \Driver^{(l')}\right)$, and still our decomposition satisfies the three properties of \Cref{lemma:structure}; these properties are all we need for our technical arguments. Therefore, without loss of generality, we can assume $c^{(l)}$'s are non-identical. 

The proof of \Cref{lemma:structure} is technical and uses convex analysis (applying the Karush–Kuhn–Tucker conditions of the convex program~\ref{eq:matching-convex-program}). We defer this proof to Appendix~\ref{apx:structure}.


\begin{proof}{\emph{Proof of \Cref{thm:opt-offline-comp-ratio}.}}
We prove a stronger statement, which is the desired competitive ratio against $\optoff$ for any realization of $\Rpa$ (essentially any adversarial second stage graph). Fix such a realization.  Consider the linear program of maximum supply-weighted bipartite matching in $G[\Ra\cup\Rpa,\Driver]$ as the primal linear program, and its dual: 
\begin{equation}
\tag{LP-1}
\label{eq:LP-max-weight}
\arraycolsep=1.4pt\def\arraystretch{1}
\begin{array}{llllllll}
\max  &\displaystyle\sum_{i\in\Ra\cup\Rpa}\displaystyle\sum_{j\in N(i)}x_{ij}w_j &~~\text{s.t.}&
&\text{min} &\displaystyle\sum_{i\in \Ra\cup\Rpa }{\alpha_i}+\displaystyle\sum_{j\in \Driver}{\beta_j}&~~\text{s.t.} \\[1.4em]
 &\displaystyle\sum_{j\in N(i)}{x_{ij}}\leq1 &  i\in\Ra\cup\Rpa~,& 
& &\alpha_i+\beta_j\geq w_j& i\in\Ra\cup\Rpa,j\in N(i)~,\\[1.4em]
 &\displaystyle\sum_{i\in N(j)}{x_{ij}}\leq 1 &j\in \Driver~, &
& &\alpha_i\geq 0 &i\in \Ra\cup\Rpa~, \\
 &x_{ij} \geq 0 &i\in\Ra\cup\Rpa,j\in N(i)~. &
& &\beta_j\geq 0  &j\in \Driver~. 
\end{array}
\end{equation}
Note that the optimal objective value of the above LP is what $\optoff$ obtains for this particular realization of $\Rpa$. Now let $M_2$ be the maximum supply-weighted matching between $\Rpa$ and remaining unmatched vertices in $\Driver$ by $M_1$. We construct a randomized dual assignment $\{\alpha_i\},\{\beta_j\}$ for the above dual LP by using matching $M_1$ so that it is feasible in expectation and 
\begin{equation}
\label{eq:primal-dual-ratio}
\expect[M_1]{\sum_{j\in \Driver}w_j\mathbbm{1}\{\textrm{$j$ is matched by $M_1\cup M_2$}\}}= \tfrac{3}{4}\cdot\expect[M_1]{\displaystyle\sum_{i\in \Ra\cup\Rpa }{\alpha_i}+\displaystyle\sum_{j\in \Driver}{\beta_j}}~,
\end{equation}
where $M_1$ is sampled as in \Cref{alg:weighted-skeleton} (with marginal edge probabilities $\xbf^*$). This finishes the proof because $\{\expect{\alpha_i}\},\{\expect{\beta_j}\}$ is a deterministic feasible dual assignment whose objective value is at most $\tfrac{4}{3}$ times the expected weight of the matching $M_1\cup M_2$, and the objective value of any feasible dual assignment is an upper bound on the optimal primal objective value.

Let $\Driver_1$ be the set of matched supply vertices by $M_1$. Consider the linear program for the maximum supply-weighted bipartite matching in the residual graph $G[\Rpa, (\Driver\setminus \Driver_1)]$. Let $\{\hat{\alpha}_i\}_{i\in\Rpa}$, $\{\hat\beta_j\}_{j\in \Driver\setminus \Driver_1}$ be the optimal dual solution of this LP. Therefore:
$$\forall i\in\Rpa,j\in \Driver\setminus \Driver_1, (i,j)\in E:~~~\hat\alpha_i+\hat\beta_j\geq w_j~.$$
Moreover, as $M_2$ is an optimal solution for the primal LP, we have:
$$\sum_{j\in \Driver}w_j\mathbbm{1}\{\textrm{$j$ is matched by $ M_2$}\}=\sum_{i\in\Rpa}\hat\alpha_i+\sum_{j\in \Driver\setminus  \Driver_1}\hat\beta_j~.$$
Now, given the decomposition in \Cref{lemma:structure}, consider this randomized dual assignment for \ref{eq:LP-max-weight}:
\begin{itemize}
\item Demand vertices in $\Ra\cup\Rpa:$
\begin{itemize}
    \item $ i\in\Rpa:~~\alpha_i\leftarrow\hat{\alpha}_i$~,
    \item $ i\in \Rider^{(l)},~l=0,1,\dots,L:~\alpha_i\leftarrow c^{(l)}=w_j(1-\sum_{i'\in N(j)}x^*_{i'j}),$ for any supply vertex $j\in\Driver^{(l)}$\\ (\emph{see the uniformity, \Cref{lemma:structure}}).
\end{itemize}
\item Supply vertices in $\Driver:$
\begin{itemize}
    \item $ j\in \Driver^{(l)},  l=0,1,\dots,L:~~\beta_j\leftarrow w_j\left(\sum_{i\in N(j)}x^*_{ij}\right)^2+\hat\beta_j\mathbbm{1}\{j\in \Driver\setminus \Driver_1\}$~.
\end{itemize}
\end{itemize}
Let $y^*_j\triangleq \sum_{i\in N(j)}x^*_{ij}$ for $j\in \Driver$ denote the probability that $j$ is matched in $M_1$. First, we have:
\begin{align*}
    \sum_{i\in\Ra\cup\Rpa}\alpha_i+\sum_{j\in \Driver}\beta_j&=\sum_{i\in\Ra}\alpha_i+\sum_{j\in \Driver}w_j{y^*_j}^2+\left(\sum_{i\in\Rpa}\hat{\alpha_i}+\sum_{j\in \Driver\setminus \Driver_1}\hat{\beta}_j\right)\\
    &=\sum_{l=0}^{L}\left(\sum_{i\in\Rider^{(l)}}c^{(l)}+\sum_{j\in \Driver^{(l)}}w_j{y^*_j}^2\right)+\sum_{j\in \Driver}w_j\mathbbm{1}\{\textrm{$j$ is matched by $ M_2$}\}\\
    &\overset{(1)}{=}\sum_{l=0}^{L}\left(\sum_{i\in\Rider^{(l)}}\left(\sum_{j\in N(i)}x^*_{ij}\right)c^{(l)}+\sum_{j\in \Driver^{(l)}}w_j{y^*_j}^2\right)+\sum_{j\in \Driver}w_j\mathbbm{1}\{\textrm{$j$ is matched by $ M_2$}\}\\
    &\overset{(2)}{=}\sum_{l=0}^{L}\left(\sum_{i\in\Rider^{(l)}}\sum_{j\in N(i)}x^*_{ij}w_j(1-y^*_j)+\sum_{j\in \Driver^{(l)}}\sum_{i\in N(j)}x^*_{ij}w_j{y^*_j}\right)+\sum_{j\in \Driver}w_j\mathbbm{1}\{\textrm{$j$ is matched by $ M_2$}\}\\
    &=\sum_{l=0}^{L}\sum_{i\in\Rider^{(l)}}\sum_{j\in N(i)}x^*_{ij}w_j+\sum_{j\in \Driver}w_j\mathbbm{1}\{\textrm{$j$ is matched by $ M_2$}\}\\
    &=\sum_{j\in \Driver}w_j\left(\sum_{i\in N(j)}x^*_{ij}+\mathbbm{1}\{\textrm{$j$ is matched by $ M_2$}\}\right)~,
\end{align*}
where (1) holds as either $c^{(l)}=0$ (when $l=0$) or $\sum_{j\in N(i)}x^*_{ij}=1$ (when $l>1$) because of saturation property (\Cref{lemma:structure}), and (2) holds as $c^{(l)}=w_j(1-y^*_j)$ for any $j\in \Driver^{(l)}$ because of uniformity (\Cref{lemma:structure}). Therefore, $\expect[M_1]{ \sum_{i\in\Ra\cup\Rpa}\alpha_i+\sum_{j\in \Driver}\beta_j}=\expect[M_1]{\sum_{j\in \Driver}w_j\mathbbm{1}\{\textrm{$j$ is matched by $ M_1\cup M_2$}\}}$. To check the feasibility, note that the dual constraints in \ref{eq:LP-max-weight} correspond to two types of edges:
\begin{itemize}
    \item $\Big[$\emph{Type \Rom{1}}: $(i,j)\in E, i\in\Ra,j\in \Driver\Big]$ Suppose $i\in\Rider^{(l)}$ and $j\in \Driver^{(l')}$, for $l,l'\in\{0,1,\dots,L\}$. Then:
    \begin{align}
    \label{eq:dual approx feasibility 1}
    \alpha_i+\beta_j\geq c^{(l)}+w_j{y^*_j}^2\overset{(1)}{\geq}c^{(l')}+w_j{y^*_j}^2\overset{(2)}{=}w_j(1-y^*_j)+w_j{y^*_j}^2\overset{(3)}{\geq} \tfrac{3}{4}w_j~,
    \end{align}
    where (1) and (2) hold because of monotonicity and uniformity properties, respectively (\Cref{lemma:structure}), and (3) holds as $y^*_j\in[0,1]$, and $\underset{x\in [0,1]}{\min}1-x+x^2=\tfrac{3}{4}$. So, $\expect[M_1]{\alpha_i+\beta_j}\geq \tfrac{3}{4}w_j$.
     \item $\Big[$\emph{Type \Rom{2}}: $(i,j)\in E, i\in\Rpa,j\in \Driver\Big]$ In this case, we have:
     \begin{align}
      \label{eq:dual approx feasibility 2}
     \begin{split}
     \expect[M_1]{\alpha_i+\beta_j}&= \expect[M_1]{\hat\alpha_i+w_j{y^*_j}^2+\hat\beta_j\mathbbm{1}\{\textrm{$j$ is not matched in $M_1$}\}}\\
     &\geq\expect[M_1]{\left(\hat\alpha_i+\hat\beta_j\right)\mathbbm{1}\{\textrm{$j$ is not matched in $M_1$}\}+w_j{y^*_j}^2}\\
     &\overset{(1)}{\geq} \expect[M_1]{w_j\mathbbm{1}\{\textrm{$j$ is not matched in $M_1$}\}+w_j{y^*_j}^2}= 
     w_j(1-{y_j^*})+w_j{y^*_j}^2\overset{(2)}{\geq}\tfrac{3}{4}w_j~,
     \end{split}
    \end{align}
    where (1) holds as $\hat\alpha_i+\hat\beta_j\geq w_j$ when $i\in\tilde{D}_2$ and $j\in S\setminus S_1$ (as it is optimal dual solution of the residual LP), and again (2) holds as $1-y_j^*+{y_j^*}^2\geq 3/4$ similar to the previous case.
\end{itemize}
Multiplying the current $\{\alpha_i\}$, $\{\beta_j\}$ by $\tfrac{4}{3}$ gives a feasible in expectation dual solution satisfying \Cref{eq:primal-dual-ratio}, which finishes the proof. \hfill\Halmos
\end{proof}
\begin{remark}
Note that \Cref{alg:weighted-skeleton} does not use any information about the vertices in $\Rp$, their probabilities $\{\pi_i\}_{i\in\Rpa}$ of being available in the second stage, or even connections to vertices in $\Driver$ when it chooses $M_1$.  However, it still obtains the optimal competitive ratio of $\tfrac{3}{4}$ against the optimum offline. In \Cref{sec:stochastic}, we show  how to incorporate this information to obtain slightly improved competitive ratio against the optimum online policy.
\end{remark}
\revcolor{
\begin{remark}
\label{remark:robustness main text}
In Appendix~\ref{apx:robustness-approx}, we 
discuss the robustness of \Cref{alg:weighted-skeleton}
(i.e., how the competitive ratio changes)
when we modify the matching decision 
in the second stage, so that instead of outputting
the maximum matching between realized 
demand vertices  
with unmatched supply vertices in the second stage,
the algorithm outputs a  
$\secondstageratio$-approximately optimal
matching. We also discuss how extra assumptions on the first-stage graph can lead to an improved competitive ratio in Appendix~\ref{apx:robustness-first}.
\end{remark}}

\section{Competitive Algorithms for Optimum Online}
\label{sec:improvedCR}
In this section, we change our benchmark to the optimum online. We first discuss the hardness of computing optimum online in \Cref{sec:hardness-online}, and show it can be characterized as a particular submodular maximization. By using approximation algorithms for this submodular maximization~\citep{SVW-17}, we obtain an alternative first stage matching. We then pick the better of this matching and the one returned by the algorithm in \Cref{sec:adversarial}. By introducing a factor revealing program and using the properties of the decomposition and the submodular formulation, we bound the final competitive ratio and show it is strictly better than $3/4$ by a constant.
\subsection{The complexity of computing optimum online}
\label{sec:hardness-online}
\newcommand{\M}{\mathcal{M}}
\newcommand{\I}{\mathcal{I}}
\newcommand{\B}{\mathcal{B}}
\newcommand{\wbf}{\mathbf{w}}
\newcommand{\Cvector}{\mathcal{C}}
 The first observation is that  the optimum online policy is always (weakly) better off by not leaving a supply vertex $j$ unmatched if there is a possibility to match this vertex in the first stage. 

\begin{remark} 
\label{rem:observation}
There exists an optimum online policy picking a matching $M^*_1$ in the first stage, where $M_1^*$ is a maximum \underline{unweighted} matching in the induced subgraph $G[\Ra, \Driver]$.
\end{remark}

The above observation in \Cref{rem:observation} helps us to state the problem of finding the optimum online policy as a \emph{matroid optimization problem}.\footnote{In what follows, we assume the reader is familiar with the general concept of a matroid as a downward-closed set system and its rank function, the concept of dual matroid, and some of the important special cases of matroids such as transversal matroids. For background materials regarding matroidal set systems, please refer to \cite{sch-03}.} Consider the transversal matroid $\M=(\Driver,\I,\B)$, in which independent sets in $\I$ are all subsets of $\Driver$ that can be matched in $G[\Ra, \Driver]$ by a matching, and the bases in $\B$ are all subsets that can be matched in $G[\Ra, \Driver]$ by a maximum matching. Denote the dual of this matroid with $\bar{\M}=(\Driver,\bar{\I},\bar{\B})$, where $T\in\bar{\B}$ if and only if $\Driver\setminus T\in\B$ (or equivalently, $T\in \bar{\I}$ if and only if $\exists T'\subseteq \Driver : T\subseteq T',\Driver\setminus T'\in \B$). 

Given the above definitions, the problem of finding the optimum online policy can be re-framed as finding a base $T_1 \in \B$ to maximize the sum of the weights of the vertices in $T_1$, plus the expected weight of the matching in $G[\Rpa , \Driver \setminus T_1]$. The first term can be written as a linear function. The second term can be written as a convex combination of transversal matroid rank functions~\citep{sch-03}. To see this, consider the weighted rank function  $\rho^\wbf:2^\Driver\rightarrow \R$ in the transversal matroid corresponding to bipartite subgraph $G[\Rpa,\Driver]$:
$$ \forall T\subseteq \Driver:~~~\rho^\wbf(T)\triangleq 
\maxmatch(\Rpa,T).$$

Note that each vertex $i\in\Rp$ will be placed in $\Rpa$ independently with probability $\pi_i$, and therefore $\rho^\wbf(.)$ is a stochastic function. Let 
\begin{equation}
\label{eq:expected-rank-def}
\forall T\subseteq \Driver: f^{\wbf}(T)\triangleq \expect[{\Rpa}]{\rho^\wbf(T)}.
\end{equation}
Given set $T_1$ of matched supply vertices of the first stage, the objective value of any policy  from the matching process (as in \cref{eq:opt-online-def-1}) can be written as $\sum_{j\in (\Driver \setminus \bar T_1)} w_j+f^{\wbf}( \bar T_1)$, where $\bar T_1\triangleq \Driver\setminus T_1$. Putting all the pieces together, the optimum online policy can be characterized as follows.
\begin{proposition}
\label{prop:opt-online-characterization}
If $M^*_1$ is the first-stage matching of the optimum online policy and $\bar T_1$ is the set of supplies that are left unmatched by $M^*_1$, then:
\begin{equation}
\label{eq:opt-sm}
    \bar T_1\in \underset{T\in\bar\B}{\argmax}~ \left(\sum_{j\in \Driver\setminus T}w_j + f^\wbf(T)\right),
\end{equation}
where $\bar\B$ is the set of bases of the dual transversal matroid corresponding to bipartite subgraph $G[\Ra, \Driver]$, and $f^\wbf(\cdot)$ is defined as in \cref{eq:expected-rank-def}. Moreover, $M^*_1$ is the matching that (fully) matches $\Driver\setminus \bar{T}_1$ to $\Ra$.
\end{proposition}

\subsubsection*{Impossibility of FPTAS} Characterization in \Cref{prop:opt-online-characterization} casts optimum online as a set function optimization. Note that computing even one \emph{exact} value query to function $f^\wbf(.)$ requires adding up exponentially many terms. However, we should be able to obtain $\epsilon$-close value queries in time polynomial in $\max(\lvert \Rider \rvert,\lvert \Driver \rvert)$ and $\frac{1}{\epsilon}$ using sampling, which guarantees to obtain a value in $[f^\wbf(T)-\epsilon,f^\wbf(T)+\epsilon]$ for each set $T$ with high probability~\footnote{Without loss of generality, we assume weights are bounded by $1$ for a polynomial sample complexity.}. Can these $\epsilon$-close value queries be used to obtain a near optimal approximation for the optimum online policy in polynomial time?

 We show that computing optimum online does not admit a \emph{Fully Polynomial Time Approximation Scheme (FPTAS)}\footnote{An FPTAS is a $(1-\epsilon)$-approximation algorithm with running time polynomial in the size of the problem and $\tfrac{1}{\epsilon}$.}, even with access to exact value queries to function $f^\wbf$. The proof is based on a reduction from the \emph{max $k$-cover} problem, for which the hardness of approximations is known~\citep{fei-98}. See Appendix~\ref{apx:no-fptas} for proof.
 

\begin{restatable}{theorem}{thmfptas}
\label{thm:no-fptas}
Computing $\opton$ as in \cref{eq:opt-online-def-1}  admits no FPTAS, unless $P=NP$.
\end{restatable}
We highlight that while \Cref{thm:no-fptas} essentially rules out near-optimal approximations for the optimum online, as we established in \Cref{sec:adversarial}, getting a constant approximation with respect to the optimum offline, and hence with respect to the optimum online, is still possible.

\revcolor{\subsubsection*{Connection to submodular maximization} The optimization in \cref{eq:opt-sm} can be rewritten as: 
\begin{equation}
\label{eq:opt-sm2}
 \bar T_1\in \underset{T\in\bar\B}{\argmax}~ \left(f^\wbf(T)-\sum_{j\in T}w_j+\sum_{j\in \Driver}w_j\right)\equiv  \underset{T\in\bar\B}{\argmax}~ \left(f^\wbf(T)-\sum_{j\in T}w_j\right) 
\end{equation}
Moreover, we can show the following simple structures of the function $f^\wbf$.
\begin{lemma} The set function $f^\wbf(.)$ (in \cref{eq:expected-rank-def}) satisfies the following properties,
\begin{itemize}
    \item Monotonicity:~  $\forall T\subseteq T'\subseteq \Driver: f^\wbf(T)\leq f^\wbf(T')~,$
    \item Submodularity:~$\forall T,T'\subseteq \Driver: f^\wbf(T\cup T')+f^\wbf(T\cap T')\leq f^\wbf(T)+f^\wbf(T')~.$
\end{itemize}
\end{lemma}
The proof of the above lemma is immediate, as matroid rank functions are monotone and submodular, and hence a convex combination of them will also be monotone and submodular.
Given this lemma, the formulation in \cref{eq:opt-sm2} casts our problem as maximizing the sum of a non-negative monotone submodular function and a (negative) linear/modular function~\citep{SVW-17}, which immediately implies the following result.} We heavily use this result in \Cref{sec:opt-online} (proof sketch in Appendix~\ref{apx:sm-proof})
\begin{proposition}
\label{prop:cont-greedy}
For every $\epsilon>0$, there exists a greedy-style algorithm returning $\bar{T}\in \bar{\B}$, such that:
$$\forall T\in\bar{\B}: f^\wbf(\bar{T})+\sum_{j\in \Driver\setminus \bar{T}}w_v\geq  (1-\tfrac{1}{e}-\epsilon)f^\wbf(T)+\sum_{j\in \Driver\setminus T}w_j~.$$
Moreover, the algorithm's running time is polynomial in $\lvert \Driver\rvert,\lvert \Rider\rvert$ and $\tfrac{1}{\epsilon}$.
\end{proposition}


\subsection{Improved  competitive ratios against optimum online}
\label{sec:opt-online}
\newcommand{\maxalgo}{\textsc{Hedge-and-Greed}}
\label{sec:stochastic}

In this section, we show how a randomized combination of \Cref{alg:weighted-skeleton} and the greedy-style algorithm in \Cref{prop:cont-greedy} beats the $3/4$ competitive ratio against the optimum online benchmark. 

\begin{algorithm}[htb]
\begin{algorithmic}[1]
\State{\textbf{input:} bipartite graph 
$G=(\Rider,\Driver,E)$, non-negative weights $\{w_j\}$, probabilities  $\{\pi_i\}_{i\in\Rp}$, $\combprob\in[0,1]$.}
\State{\textbf{output:} bipartite matching $M_1$ in $G[\Ra, \Driver]$, bipartite matching $M_2$ in $G[\Rpa,\Driver]$}
\vspace{2mm}
\State{In the first stage, run 
\Cref{alg:weighted-skeleton} (for any increasing strictly convex function $g$)
with probability $\combprob$ and return its matching as $M_1$.}
\State{Otherwise, run the greedy-style algorithm 
in \Cref{prop:cont-greedy} to return $\bar{T}\in \bar{\B}$, and return the maximum unweighted matching $M_1$ between $\Ra$ and $\Driver\setminus \bar{T}$ in the first stage.}
\State{In the second stage, return the maximum weighted matching $M_2$ between $\Rpa$ and the remaining vertices of $\Driver$.}
\end{algorithmic}
\caption{\maxalgo}
\label{alg:online}
\end{algorithm}

\begin{theorem}
\label{thm:opt-online-comp-ratio}
For unweighted supply vertices,
\Cref{alg:online} with 
$\combprob = \sfrac{1}{(e - 1)}$
is $\onlineUnweightedRatioText$-competitive
against the optimum online, i.e.,
\begin{align*}
\compratioMUW(\Cref{alg:online},\opton) &\geq 
\onlineUnweightedRatio
\approx 0.7674. 
\intertext{
For weighted case,
\Cref{alg:online}
with 
$\combprob = 0.7$
is $\onlineWeightedRatio$-competitive
against the optimum online, i.e.,}
\compratioM(\Cref{alg:online},\opton) &\geq 
\onlineWeightedRatio.
\end{align*}
\end{theorem}

In principle, given the second stage arrival probabilities $\pi_i$,
we can estimate the supply efficiency (i.e., the expected total weight of the final matching) of \Cref{alg:weighted-skeleton} and
the greedy-style algorithm with sampling and the Monte Carlo method. We can then run the better of the two. With high enough estimation precision, this
algorithm beats \Cref{alg:online}; nevertheless, \Cref{alg:online}
is sufficient in our analysis to obtain the improved competitive ratio against the optimum online policy.

\paragraph{Overview of the proof.} To prove \Cref{thm:opt-online-comp-ratio}, we use a \emph{factor revealing (FR) program}, that is, a non-linear program whose optimal solution is a lower-bound for the competitive ratio of \Cref{alg:online} against the optimum online. We start from a basic factor revealing program, in which every feasible instance of the 
weighted two-stage stochastic matching problem can be mapped to a feasible instance of this program, and the objective value of the mapped instance will be no more than the ratio $\alg/\opton$ over this instance (where $\alg$, by notation abuse, is the expected total weight of the matching of \Cref{alg:online}). By a series of relaxations, through mapping each feasible solution of the current program to a feasible solution of the next program while only lowering the objective value, we essentially relax our primary program to this simple program:

\begin{align}
\label{eq:online final}
\tag{$\mathscr{P}^{\textrm{FR-final}}$}
\begin{array}{lll}
    \min_{
    \msfirst,
    \mssecond,
    \weightp,
    \weightdp
    }
    &\displaystyle
    1-
    \frac
    {\combprob(
    (1-\mssecond)\msc
    -
    (1-\msfirst-\mssecond)\,\weightdp
    )
    +
    \frac{1}{e}(1 - \combprob)
    \mssecond\weightp}
    {\msfirst + 
    \mssecond\weightp}
    & \text{s.t.} \\
     &\displaystyle
     \msc =
     \frac{1-\msfirst}{
     \msfirst
     +
     \frac{\mssecond}{\weightp}
     +
     \frac{1-\msfirst-\mssecond}{\weightdp}
     } 
     & 
     \\~&~&~\\
     &
     \msfirst,\mssecond\in[0, 1],\
     \msfirst + \mssecond \leq 1
     & 
     \\
     &
     \weightdp \leq 1\leq  \weightp,\
     (1-\msfirst-\mssecond)\,\weightdp \geq (1-\msfirst-\mssecond)\,\msc
     &
\end{array}
\end{align}
Roughly speaking, a feasible solution $\left(\msfirst,
    \mssecond,
    \weightp,
    \weightdp\right)$ of program~\ref{eq:online final} corresponds to a scenario in the weighted two-stage stochastic matching problem, where these three events happen:
\begin{enumerate}
    \item The bipartite graph $(\Ra,\Driver,E)$ is such that $L=1$, $\Rider^{(0)}=\Driver^{(0)}=\emptyset$, and $\msfirst =  \tfrac{|\Rider^{(1)}|}
{|\Driver^{(1)}|}$ after applying the structural decomposition of \Cref{lemma:structure}.
\item All supply vertices can be partitioned into three groups: group $A$, which are supply vertices matched by the optimum online in the first stage; group $B$, which are supply vertices matched in the second stage of the optimum online policy w.p.\ $1$; and group $C$, which are supply vertices that remained unmatched through the process w.p.\ $1$. Moreover, we have
$\mssecond = 
\tfrac{|B|}{|\Driver^{(1)}|}$ (note that $\lvert A\rvert=\lvert \Rider^{(1)}\rvert$, and thus $\msfirst = 
\tfrac{|A|}{|\Driver^{(1)}|}$).
\item all supply vertices in group $A$ (resp.\ $B$, 
$C$)
have the same weight 1 (resp.\ $\weightp$, $\weightdp$)
where $\weightdp \leq 1 \leq 
\weightp$ (for unweighted setting, $\weightp = \weightdp = 1$).
\end{enumerate}
If all the three events happen, the worst-case  ratio of $\alg/\opton$ is exactly what program~\ref{eq:online final} captures. In this sense, the ``possible scenario''~\footnote{To be completely accurate, as we will see in the proof of \Cref{lem:online final} in the appendix, not every feasible solution of the program~\ref{eq:online final} corresponds exactly to a feasible 
weighted two-stage stochastic matching instance; nonetheless, program~\ref{eq:online final} is a valid relaxation.} described above is the worst-case scenario for our original factor revealing program. We finish the proof by solving program~\ref{eq:online final} for weighted and unweighted cases, with our specific values of $\combprob$. The above approach is summarized in following key proposition. \revcolor{In what follows, we provide a formal proof of this proposition for the unweighted setting. We defer the proof of the the weighted version to Appendix~\ref{apx:online final}.
\begin{restatable}{proposition}{lemmaoptonline}
\label{lem:online final}
The competitive ratio of \Cref{alg:online}
with arbitrary 
$\combprob$ 
(resp.\ $\combprob = 0.7$)
for unweighted (resp.\ weighted) setting
against the optimum online policy
is at least
the value of program~\ref{eq:online final}.
\end{restatable}
Before the proof of \Cref{lem:online final} for the unweighted setting, we first introduce the following technical lemma. The proof is immediate by applying the uniformity and saturation properties 
in \Cref{lemma:structure}.
\begin{lemma}
\label{lem:msc closed form}
In the structural decomposition 
(\Cref{lemma:structure}),
$\mspairci = 
\frac{|\Driveri| - |\Rideri|}{\sum_{j\in\Driveri}\frac{1}{\weight_j}}$.
\end{lemma}

\begin{proof}{\emph{Proof of \Cref{lem:online final}
for unweighted setting.}}
The proof is done in three major steps:

\noindent{\textbf{Step 1- writing a factor revealing program:}} Consider the following optimization program, which is parameterized by $\totaldriver,\totalmspair\in\N$, and has variables $ \matchprobs,  \mspairinstance$:
\begin{align}
    \label{eq:online initial unweighted}
    \tag{$\mathscr{P}^{\textrm{FR-first}}$}
\begin{array}{llll}
    \min\limits_{
    \substack{
    \matchprobs,\\
    \mspairinstance}
    }
    &\displaystyle
    \frac{
    \combprob
    \left(
    \MS
    \right)
    +
    (1-\combprob)
    \left(
    \GD
    \right)
    }{
    \OPT
    }
    & \text{s.t.} & \\
    &
    \msfirsti \leq \msdriveri,
    \ 
    \msrideri \leq \msdriveri
    & 
    \mspair\in [\totalmspair] &\textit{\footnotesize{(Feasibility-1)}} \\
    &
    \displaystyle\sum_{\mspair\in[\totalmspair]}\msdriveri = \totaldriver,
    \ 
     \displaystyle\sum_{\mspair\in[\totalmspair]}
    \msfirsti = 
     \displaystyle\sum_{\mspair\in[\totalmspair]}
    \msrideri
    & & \textit{\footnotesize{(Feasibility-2)}}
    \\
    &
    \matchprobi = 0 
    &
    \driver \in \Driverfirst & \textit{\footnotesize{(Feasibility-3)}}
    \\
    ~&~&~&~
    \\
    &
    \mspairci < \mspairc^{(\mspair + 1)}
    &
    \mspair\in [\totalmspair - 1] &  \textit{\footnotesize{(Monotonicity-1)}}
    \\
    &
     \displaystyle\sum_{\mspair' = \mspair}^{\totalmspair}
    \msfirst^{(\mspair')} \leq 
     \displaystyle\sum_{\mspair' = \mspair}^{\totalmspair}
    \msrider^{(\mspair')}
    & 
    \mspair \in [\totalmspair] &  \textit{\footnotesize{(Monotonicity-2)}}
    \\
    &\displaystyle
    \mspairci  
    =
    \frac{\msdriveri - \msrideri}{\msdriveri}
    &
    \mspair \in [\totalmspair] & 
    \\
    ~&~&~&~
    \\
    &
    \msdriveri, \msfirsti, \msrideri\in \N 
    &
    \mspair\in[\totalmspair] & 
    \\
    &
    \matchprobi \in [0, 1]
    &
    \driver\in[n] & 
\end{array}
\end{align}
where the index sets are defined as $$\Driveri  \triangleq\left[\sum_{\mspair' = 1}^{\mspair - 1}\msdriver^{(\mspair')}+1: \sum_{\mspair' = 1}^{\mspair }\msdriver^{(\mspair')}\right]~~~\textrm{and}~~~ \Driverfirst\triangleq \displaystyle\bigcup\limits_{\mspair \in[\totalmspair]}\left[
    \sum_{\mspair' = 1}^{\mspair - 1}\msdriver^{(\mspair')}+1
    :
    \sum_{\mspair' = 1}^{\mspair-1}\msdriver^{(\mspair')}
    +\msfirsti
    \right],$$
and auxiliary variables $\MS,\GD$, and $\OPT$ are defined as:

\begin{equation}
\label{eq:auxilaryvars}
\MS \triangleq\msdriver-\sum_{\mspair\in[\totalmspair]}\sum_{\driver\in\Driveri}\mspairci(1-\matchprobi),~~~~ \GD  \triangleq \left(1 - \frac{1}{e}\right)\sum_{\driver\in[\totaldriver]}\matchprobi+ \sum_{\mspair\in[\totalmspair]}\msfirsti,~~~~\OPT  \triangleq\sum_{\driver\in[\totaldriver]}\matchprobi+\sum_{\mspair\in[\totalmspair]}\msfirsti
\end{equation}

We show that every unweighted 
two-stage matching instance
can be mapped to a feasible 
solution in program~\ref{eq:online initial unweighted} (after settings parameters $\totaldriver$ and $\totalmspair$ appropriately),
and the competitive ratio
of \Cref{alg:online} against 
the optimum online
in this instance is at least the objective value of 
this feasible solution.

\emph{The desired mapping.} Consider any unweighted 
two-stage stochastic matching instance
$\left\{(\RaHat\cup \RpHat, \DriverHat,
\hat E), \{\hat\pi_u\}_{u\in \RpHat}\right\}$.~\footnote{We use notation \^{} to denote the two-stage
matching instance} Let $\DriverfirstHat\subseteq \DriverHat$ be the subset of
supply vertices matched by the optimum online policy 
in the first stage, and $\{(\RideriHat, \DriveriHat)\}_{\mspair = 0}^{ 
\totalmspairHat}$ be the structural decomposition of this instance as in \Cref{lemma:structure} (the pairs of the decomposition are indexed so that $0=\hat\mspairc^{(0)}<\hat\mspairc^{(1)}<\ldots<\hat\mspairc^{(\totalmspairHat)}$).
Let $\totaldriver \leftarrow\lvert\DriverHat \backslash \DriverHat^{(0)}\rvert$ and $\totalmspair\leftarrow\totalmspairHat$. We now construct the following solution 
for program~\ref{eq:online initial unweighted}:
\begin{gather*}
\text{for all }
\mspair\in[\totalmspair]:\quad
\msdriveri \leftarrow
\lvert\DriveriHat\rvert,\quad
\msfirsti \leftarrow
\lvert\DriveriHat \cap \DriverfirstHat\rvert,\quad
\msrideri \leftarrow
\lvert\RideriHat\rvert~.
\end{gather*}
To define the assignment of $\matchprobs$, consider a bijection $\sigma$ from all supply vertices in 
$\DriverHat \backslash \DriverHat^{(0)}$ to $[\totaldriver]$ such that for any two supply vertices $j \in \DriveriHat$ and  $j' \in \DriverHat^{(\mspair')}$,
if either (a) $\mspair < \mspair'$,
or (b) $\mspair = \mspair'$, 
$j \in \DriveriHat \cap \DriverfirstHat$ and $j' \notin \DriveriHat \cap \DriverfirstHat$,
then $\sigma(j) < \sigma(j')$.
Trivially, such a bijection exists. Let $\matchprobHat_j$ be the probability that supply vertex $j$ is matched
in the second stage by the optimum online policy. We finish our construction of a solution by assigning:
\begin{gather*}
\text{for all $j \in \DriverHat \setminus \DriverHat^{(0)}$}: \quad
\matchprob_{\sigma(j)} \leftarrow \matchprobHat_j
\end{gather*}

\emph{Objective value of the constructed solution.}
We first formulate the competitive ratio of \Cref{alg:online} against the optimum online
policy on the original instance. The expected size of the matching of the optimum online policy is 
the number of supply vertices matched in the first stage, plus the expected total number of the supply vertices matched
in the second stage, that is, 
\begin{equation}
\label{eq:opt-bound}
    \lvert\DriverfirstHat\rvert+\sum_{j\in\DriverHat\setminus\DriverfirstHat}\matchprobHat_j=\sum_{\mspair\in[0:\totalmspairHat]}\lvert\DriveriHat\cap\DriverfirstHat\rvert 
    +
    \sum_{j \in \DriverHat \setminus \DriverHat^{(0)}}\matchprobHat_j
    =
    \lvert\DriverHat^{(0)}\rvert
    +\OPT~,
\end{equation}
where $\OPT$ is defined in \eqref{eq:auxilaryvars}. In the above equation, we used the fact that $\DriverHat^{(0)}\subseteq \DriverfirstHat$, simply because vertices in $\DriverHat^{(0)}$ should be matched in every first stage unweighted maximum matching, and that $\matchprobHat_j=0$ for $j\in\DriverfirstHat$. 

By applying \Cref{prop:cont-greedy}, the expected size of the matching suggested by the greedy-style algorithm (which is used in \Cref{alg:online}) is at least 
\begin{equation}
\label{eq:gd-bound}
    \lvert\DriverfirstHat\rvert+\left(1-\frac{1}{e}
    \right)\sum_{j\in\DriverHat\setminus\DriverfirstHat}\matchprobHat_j=\sum_{\mspair\in[0:\totalmspairHat]}\lvert\DriveriHat\cap\DriverfirstHat\rvert 
    +
    \left(1-\frac{1}{e}
    \right)\sum_{j \in \DriverHat}\matchprobHat_j
    =
    \lvert\DriverHat^{(0)}\rvert
    +\GD~,
\end{equation}
where $\GD$ is defined in \eqref{eq:auxilaryvars}, and again we use the facts that $\DriverHat^{(0)}\subseteq \DriverfirstHat$ and $\matchprobHat_j=0$ for $j\in\DriverfirstHat$. 

Let $\msprobHat_j$ be the probability 
that supply vertex $j \in \DriveriHat$ is matched by 
\Cref{alg:weighted-skeleton} in the first stage.
Then, by the structural decomposition in \Cref{lemma:structure},
$\msprobHat_j = 1$ if $\mspair = 0$
and $(1-\msprobHat_j) = 
\hat\mspairc^{(\mspair)}
$ otherwise.
To find a lower-bound for the expected total
number of supply vertices matched in the second stage
of \Cref{alg:weighted-skeleton}, note that
this algorithm indeed finds a maximum 
matching between realized second stage demand vertices and unmatched supply vertices. Such a matching is no smaller than the projection of the matching picked by the optimum online policy onto the supply vertices not matched by \Cref{alg:weighted-skeleton} during the first stage. The expected size of such a projected matching is $\sum_{j \in \DriverHat}
(1 - \msprobHat_j) \matchprobHat_j$, due to the linearity of the expectation. Therefore, the expected total size of the final matching picked by
\Cref{alg:weighted-skeleton} at the end of the second stage will be at least
\begin{equation}
\label{eq:ms-bound}
    \sum_{j\in \DriverHat}\msprobHat_j
  +
  \sum_{j \in \DriverHat}
(1 - \msprobHat_j) 
\matchprobHat_j
= 
  \lvert\DriverHat\rvert
  -
   \sum_{\mspair\in[\totalmspairHat]}
   \sum_{j\in\DriveriHat}
   \hat\mspairc^{(\mspair)}(1-\matchprobHat_j)
   =
   \lvert\DriverHat^{(0)}\rvert
    +\MS~,
\end{equation}
where $\MS$ is defined in \eqref{eq:auxilaryvars}, and we use the fact that every supply vertex in $\DriverHat^{(0)}$ is matched in the first stage matching of \Cref{alg:weighted-skeleton} (and therefore $\msprobHat=1$ for such a supply vertex).

Putting the bounds in \eqref{eq:opt-bound}, \eqref{eq:gd-bound}, and \eqref{eq:ms-bound} together, the competitive ratio 
of \Cref{alg:online}
against the optimum online policy on the original
two-stage matching instance
is at least 
\begin{align*}
    \frac{
    \combprob\left(\MS+\lvert\DriverHat^{(0)}\rvert\right)) + (1-\combprob)\left((\GD+\lvert\DriverHat^{(0)}\rvert\right)
    }{
    \OPTHat} 
    &=
    \frac{
    \lvert\DriverHat^{(0)}\rvert
    +
    \combprob
    \MS
    +
    (1-\combprob)
    \GD
    }{
    |\DriverHat^{(0)}|
    +
    \OPT
    }
    \geq
    \frac{\combprob
    \MS
    +
    (1-\combprob)
    \GD
    }{\OPT}~,
\end{align*}
where the last ratio is the objective value of
our constructed solution in program~\ref{eq:online initial unweighted}.

\emph{Feasibility of the constructed solution:}
Constraints (Feasibility-1), (Feasibility-2), and (Feasibility-3) of program~\ref{eq:online initial unweighted} hold by construction.
Because of \Cref{lem:msc closed form}, 
for any $\mspair \in [\totalmspair]$ and 
any supply vertex $j\in \DriveriHat$,
$\mspairci$
defined in program~\ref{eq:online initial unweighted} equals 
to $\hat\mspairc^{(\mspair)}$ in the original two-stage matching instance.
Therefore, constraint (Monotonicity-1) holds by construction, as $\hat\mspairc^{(\mspair)}<\hat\mspairc^{(\mspair+1)}$.
Constraint (Monotonicity-2) also holds because of the following argument:
$\sum_{\mspair'=\mspair}^{\totalmspair}
\msfirst^{(\mspair')}$ is the 
number of supply vertices in 
$\bigcup_{\mspair'=\mspair}^{\totalmspair}
\DriverHat^{(\mspair')}$
whom are matched in the optimum online policy 
during the first-stage. Moreover,
$\sum_{\mspair'=\mspair}^{\totalmspair}
\msrider^{(\mspair')}$ is the number
of the demand vertices in $\bigcup_{\mspair'=\mspair}^{\totalmspair}
\RiderHat^{(\mspair')}$.
Since ``Monotonicity'' property of the structural decomposition in \Cref{lemma:structure} for the original two-stage matching instance
guarantees that there is no edge from 
the demand vertices in 
$\RiderHat \setminus \bigcup_{\mspair'=\mspair}^{\totalmspair}
\RiderHat^{(\mspair)}
=
\bigcup_{\mspair'=0}^{\mspair-1}
\RiderHat^{(\mspair')}$
to supply vertices in 
$\bigcup_{\mspair'=\mspair}^{\totalmspair}
\DriverHat^{(\mspair')}$,
the number of supply vertices in $
\bigcup_{\mspair'=\mspair}^{\totalmspair}
\DriverHat^{(\mspair')}$
at the first stage
is at most the number of demand vertices 
in $
\bigcup_{\mspair'=\mspair}^{\totalmspair}
\RiderHat^{(\mspair')}$.
Therefore, for every $\mspair\in[\totalmspair]$, we have 
$\sum_{\mspair'=\mspair}^{\totalmspair}
\msfirst^{(\mspair')} \leq 
\sum_{\mspair'=\mspair}^{\totalmspair}
\msrider^{(\mspair')}$.
 \vspace{2mm}
 
\noindent\textbf{Step 2- restricting probabilities $\mathbf{\matchprobs}$:}
In this step, we first argue that program~\ref{eq:online initial unweighted} has an optimal solution where
$\matchprobi \in \{0, 1\}$ for all $\driver \in [n]\backslash\Driverfirst$. To see this, note that for every $\driver\in [n]\backslash\Driverfirst$, the partial derivative of the objective function with respect to $\matchprob_\driver$ has the same sign for all $\matchprobi \in [0, 1]$. Therefore, moving $\matchprob_\driver$ to one of the extreme points $\{0,1\}$ weakly decreases the objective function. Hence, there exists an optimal solution of
program~\ref{eq:online initial unweighted} in which
$\matchprobi \in \{0, 1\}$ for all $\driver\in [n]\backslash\Driverfirst$. 

Having this restriction on $\matchprobs$, we simplify
program~\ref{eq:online initial unweighted} by replacing  $\matchprobs$ 
with new variables $\msseconds$, where $\mssecondi \triangleq
\lvert\{\driver\in\Driveri:\matchprobi=1\}\rvert$, 
and dropping the constraint (Feasibility-3):
\begin{align}
    \label{eq:online all matched unweighted}
    \tag{$\mathscr{P}^{\textrm{FR-second}}$}
\begin{array}{llll}
    \min\limits_{
    \substack{
    \mspairinstancenew}
    }
    &\displaystyle
    \frac{
    \combprob
    \left(
    \MS
    \right)
    +
    (1-\combprob)
    \left(
    \GD
    \right)
    }{
    \OPT
    }
    & \text{s.t.} &\\
    &
    \msfirsti + \mssecondi \leq \msdriveri,
    \ 
    \msrideri \leq \msdriveri
    & 
    \mspair\in [\totalmspair]
    & \footnotesize{\textit{(Feasibility-1)}}
    \\
    &
    \displaystyle\sum_{\mspair\in[\totalmspair]}\msdriveri = \totaldriver,
    \ 
     \displaystyle\sum_{\mspair\in[\totalmspair]}
    \msfirsti = 
     \displaystyle\sum_{\mspair\in[\totalmspair]}
    \msrideri
    & & \footnotesize{\textit{(Feasibility-2)}}
    \\
    ~&~&~&~
    \\
    &
    \mspairci < \mspairc^{(\mspair + 1)}
    &
    \mspair\in [\totalmspair - 1] & \footnotesize{\textit{(Monotonicity-1)}}
    \\
    &
    \displaystyle\sum_{\mspair' = \mspair}^{\totalmspair}
    \msfirst^{(\mspair')} \leq 
     \displaystyle\sum_{\mspair' = \mspair}^{\totalmspair}
    \msrider^{(\mspair')}
    & 
    \mspair \in [\totalmspair] & \footnotesize{\textit{(Monotonicity-2)}}
    \\
    &\displaystyle
    \mspairci  
    \triangleq
    \frac{\msdriveri - \msrideri}
    {\msdriveri}
    &
    \mspair \in [\totalmspair] & 
    \\
    ~&~&~&~
    \\
    &
    \msdriveri, \msfirsti, \msrideri,
    \mssecondi\in \N 
    &
    \mspair\in[\totalmspair] & 
\end{array}
\end{align}
where auxiliary variables $\MS,\GD$, and $\OPT$ are defined as:
\begin{equation}
\label{eq:auxilaryvars-second}
 \MS \triangleq\totaldriver-\sum_{\mspair\in[\totalmspair]}\mspairci(\msdriveri-\mssecondi),~~~~\GD \triangleq\sum_{\mspair\in[\totalmspair]}\msfirsti+(1-\frac{1}{e})\sum_{\mspair\in[\totalmspair]}\mssecondi,~~~~\OPT  \triangleq\sum_{\mspair\in[\totalmspair]}\msfirsti+\sum_{\mspair\in[\totalmspair]}\mssecondi~.
\end{equation}
\vspace{2mm}

\noindent\textbf{Step 3- reduction to the case $\totalmspair=1$:}
In this step, we first argue that it is sufficient
to consider only solutions of program~\ref{eq:online all matched unweighted}, where the
constraints (Monotonicity-2) are tight for all $\mspair \in [\totalmspair]$.
Consider any feasible solution of program~\ref{eq:online all matched unweighted}.
We modify this solution as follows, so that (Monotonicity-2) will be tight and the objective value weakly decreases: suppose $l^*\in[\totalmspair]$ is the largest index such that $\sum_{l=1}^{l^*}(\msdriveri-\msfirsti)\leq \sum_{l=1}^{l^*}\mssecondi$. Now set $\msfirsti_{\textrm{new}}\leftarrow\msrideri$ for $l\in[\totalmspair]$, $\mssecondi_{\textrm{new}}\leftarrow \msdriveri-\msfirsti$ for $l\in[1:l^*]$, $\mssecond_{\textrm{new}}^{l^*+1}\leftarrow \sum_{l=1}^{l^*+1}\mssecondi- \sum_{l=1}^{l^*+1}(\msdriveri-\msfirsti)$, $\mssecondi_{\textrm{new}}\leftarrow 0$ for $l\in[l^*,\totalmspair]$, and keep all other variables unchanged. The feasibility of the modified solution is by construction. Moreover, because of (Feasibility-2) and the fact that $\sum_{l\in[\totalmspair]}\mssecondi_{\textrm{new}}=\sum_{l\in[\totalmspair]}\mssecondi$, quantities $\GD$ and $\OPT$ remain unchanged. As the total mass in $\sum_{l\in[\totalmspair]}\mssecondi$ moves to lower $l$'s in $\{\mssecondi_{\textrm{new}}\}$, $\MS$ weakly decreases due to (Monotonicity-1), and so the objective value.

If (Monotonicity-2) is tight, we can essentially remove the variables $\msrideri$ from program~\ref{eq:online all matched unweighted}(as $\msfirsti=\msrideri$). Now, we can restrict our attention to $L=1$ to have a relaxation. To see this, note that:
\begin{equation}
\label{eq:lower-bound-one}
 \frac{\combprob\MS+(1-\combprob)\GD}{\OPT}= 
 \frac{\sum_{l\in\totalmspair}{\left(\combprob\MS^{(l)}+
    (1-\combprob)
    \GD^{(l)}\right)}
    }{\sum_{l=1}^{\totalmspair}\OPT^{(l)}}\geq \underset{l\in[\totalmspair]}{\min}~ \frac{\combprob\MS^{(l)}+(1-\combprob)\GD^{(l)}}{\OPT^{(l)}}~,
\end{equation}
where $\MS^{(l)}\triangleq\msdriveri-\mspairci(\msdriveri-\mssecondi)$, $\GD^{(l)}\triangleq \msfirsti+(1-\frac{1}{e})\mssecondi$, and $\OPT^{(l)}\triangleq\msfirsti+\mssecondi$. Now, fix a feasible solution $\{\msfirsti,\mssecondi\}$. Suppose, w.l.o.g., the minimum above is attained at $l=1$. Consider the special case of program~\ref{eq:online all matched unweighted} for $\totalmspair=1$ and $\totaldriver=\totaldriver^{(1)}$ (after dropping variable $m$) as our new program:
\begin{align}
    \label{eq:online hypermatchable unweighted}
    \tag{$\mathscr{P}^{\textrm{FR-third}}$}
\begin{array}{lll}
    \min\limits_{
    \substack{
    \msfirst\in\N,
    \mssecond\in\N}
    }
    &\displaystyle
    \frac{
    \combprob
    \left(
    \MS
    \right)
    +
    (1-\combprob)
    \left(
    \GD
    \right)
    }{
    \OPT
    }
    & \text{s.t.} \\
    &
    \msfirst + \mssecond \leq \msdriver
    & 
    \\
    &\displaystyle
    \mspairc  
    =
    \frac{\msdriver - \msfirst}{\msdriver}
    &
\end{array}~,
\end{align}
where $\MS \triangleq\totaldriver-\mspairc(\msdriver-\mssecond)$, $ \GD\triangleq \msfirst+\left(1 -\frac{1}{e}\right)\mssecond$, and $ \OPT \triangleq\msfirst+\mssecond$. After setting $\msfirst\leftarrow\msfirst^{(1)}$ and $\mssecond\leftarrow\mssecond^{(1)}$ (which is clearly feasible in program~\ref{eq:online hypermatchable unweighted}), the objective value of program~\ref{eq:online hypermatchable unweighted} under the assignment $\{\msfirst,\mssecond\}$ is no larger than that of program~\ref{eq:online all matched unweighted} under the assignment $\{\msfirsti,\mssecondi\}$ due to \cref{eq:lower-bound-one}. We finish the proof by relaxing variables $\msfirst, \mssecond$ to be real numbers, and then normalizing them by $\msdriver$ to be in the range $[0,1]$. This proves that program~\ref{eq:online final} is a relaxation to the original program, as desired.\hfill \Halmos
\end{proof}

\begin{proof}{\emph{Proof of \Cref{thm:opt-online-comp-ratio}}.}
We evaluate the 
program~\ref{eq:online final} in both 
unweighted and weighted settings. In the former, we further restrict to $\weightp = \weightdp = 1$
in program~\ref{eq:online final}. Therefore, it becomes
\begin{align*}
    \min_{
    \msfirst,\mssecond\in[0, 1]:
    \msfirst +\mssecond \leq 1
    }
    1-
    \frac{
    \combprob\,\msfirst\,\mssecond
    +
    \frac{1}{e}(1-\combprob)\mssecond
    }{\msfirst+\mssecond}~,
\end{align*}
with the optimal value $\onlineUnweightedRatioText$
at $\msfirst =
\tfrac{\combprob-\frac{1}{e}(1-\combprob)}
{2\combprob},
\mssecond = 
\tfrac{\combprob+\frac{1}{e}(1-\combprob)}{2\combprob}$
for $\combprob=\tfrac{1}{(e-1)}$. In the weighted setting, we numerically solve
the program for
 $\msfirst + \mssecond = 1$
and $\msfirst + \mssecond < 1$.
The optimum value is  $\approx\onlineWeightedRatio$.\hfill\Halmos
\end{proof}

}

 \section{Two-stage Stochastic Joint Matching and Pricing}
 \label{sec:matching and pricing}
In this section, we study the joint matching and pricing problem to maximize market efficiency. Here, potential demand vertices (i.e., demand vertices in $\Rp$) have private valuations for the service (drawn independently from common knowledge distributions). A naive economic intuition might suggest that prices should be set to zero to maximize market efficiency. Interestingly, this intuition breaks here. Note that the platform only gets to see whether each pending demand vertex has accepted her price or not (and not the actual valuations). Therefore, \emph{prices can potentially discover new information about the demand vertices' valuations} to extract more welfare, as the example below also suggests.

\begin{example} Consider an instance with one supply vertex (whose weight is $0$), no first stage demand, and two second stage demand vertices $1$ and $2$. Moreover, let $v_1\sim\textrm{uniform}[0,1]$ and $v_2\sim\textrm{uniform}[0,2]$. If prices are zero, then demand vertex $2$ should always be matched, which gives an expected welfare of $1$. Now suppose $p_1=0$ and $p_2=\tfrac{1}{2}$. Consider a matching policy as follows: if demand vertex $2$ accepts her price, then she gets matched. Otherwise, demand vertex $1$ gets matched. The expected welfare of such a matching is $\tfrac{3}{4}\times\tfrac{5}{4}+\tfrac{1}{4}\times \tfrac{1}{2}=\tfrac{17}{16}>1$.
\end{example}

Given the above intuition, we propose an optimal $\tfrac{1}{2}$-competitive algorithm
for this problem against the optimum offline benchmark ($\optoffmp$). We first introduce a convex program~\ref{eq:ex ante} that essentially is an \emph{ex ante relaxation} to the optimum offline, that is, its optimal solution suggests a vector of prices and a randomized matching that is only feasible in-expectation and its market efficiency is no less than the optimum offline. We then show how to obtain a feasible policy from this optimal solution that only loses a factor $\tfrac{1}{2}$ in the market efficiency. 

We further study a special case of the above problem in Appendix~\ref{apx:onestagepricing}, where it is just a single stage --- or equivalently, there is no demand vertex in the first stage. We propose an optimal $(1-\tfrac{1}{e})$-competitive algorithm against the ex-ante relaxation. The main technique is to reduce the problem to the problem of bounding the correlation gap of submodular functions~\citep{ADSY-10}.


\subsection{Ex ante relaxation}
\label{sec:ex ante relaxation}
First, as a reminder, consider the optimum offline benchmark $\optoffmp$, that is,
 \begin{equation*}
    \displaystyle\optoffmp\triangleq \expect[\val_\rider\sim F_\rider,\rider\in\Rp]{\maxmatchboth(\Rider,\Driver)},
    \end{equation*}
where 
$\weightHat_\rider=w_\rider(\price_\rider)\triangleq
\expect[\val\sim\Fu]{\val|\val \geq \price_\rider}
$ for $\rider\in\Rp$, $\weightHat_\rider=\val_\rider$ for $\rider\in\Ra$ and $\weightHat_\driver=\weight_\driver$ for $\driver\in \Driver$.
For each demand vertex $\rider$, define the threshold function as
$\Thresh_\rider(\acceptprob) \triangleq
\argmax\{\thresh : 1 - F_\rider(\thresh) 
\geq \acceptprob\}$.
Now consider this convex program as an ex ante relaxation for $\optoffmp$:
\begin{align}
\tag{$\mathscr{P}^{\textrm{EAR}}$}
\label{eq:ex ante}
\begin{array}{llll}
	\max~~&
	\displaystyle\sum_{\rider\in \Ra}
	\weightHat_\rider
	\assignvertexprobi
	+
	\displaystyle\sum_{\rider\in \Rp}
	\weight_\rider\left(\Thresh_\rider(\assignvertexprobi)\right)
	\assignvertexprobi
	+
	\displaystyle\sum_{\driver\in \Driver}
	\weightHat_\driver
	\assignvertexprobj
	&\text{s.t.}&\\[1em]
	&\assignvertexprobi = 
	\displaystyle\sum_{\driver\in N(\rider)}\assignprobi\leq 1 &
	 i\in \Rider~,&\\[1em]
	&
	\assignvertexprobj = 
	\displaystyle\sum_{\rider\in N(\driver)}
	\assignprobi\leq 1 &
    \driver\in \Driver~,&\\[1em]
	&\assignprobi\geq 0&
	(\rider,\driver)\in E~~.&
\end{array}
\end{align}
Notably, $\weight_\rider(\Thresh_\rider(\assignvertexprob))\,
\assignvertexprob
= \displaystyle\int_{0}^{\assignvertexprob}\Thresh_\rider(q)dq$, and therefore it is concave as a function of $\assignvertexprob$ (as $\Thresh_\rider$ is monotone non-increasing). 

\begin{lemma}
\label{lem:ex ante}
The optimal value of 
program~\ref{eq:ex ante} is at least 
$\optoffmp$.
\end{lemma}
\begin{proof}
We show this lemma by constructing a feasible assignment
for program~\ref{eq:ex ante} whose objective value is no smaller than $\optoffmp$. Consider the assignment where 
$$\assignprobi = \Pr{(\rider, \driver) \text{ is matched in }
\optoffmp}$$ for all edges $(\rider, \driver)$ in $E$.
Since $\optoffmp$ always returns a feasible matching, the assignment $\assignprobs$ is feasible in program~\ref{eq:ex ante}.
To compare the objective value with $\optoffmp$,
note that the expected contribution of each demand vertex $\rider$ in $\Ra$ is identical in 
program~\ref{eq:ex ante} and $\optoffmp$. The same is true for each supply vertex $\driver$
in $\Driver$.
For demand vertex $\rider$ in $\Rp$,
her contribution in $\optoffmp$ is 
$$\expect{\val_\rider\cdot\mathbbm{1}\{\textrm{$\rider$ is matched in $\optoffmp$}\}}.$$
This is at most 
$\weight_\rider(\Thresh_\rider(\assignvertexprobi))\,\assignvertexprobi$, as vertex $i$ matches in $\optoffmp$ with probability $\assignvertexprobi = \sum_{\driver\in N(\rider)}\assignprobi$.
\end{proof}

\subsection{Optimal two-stage joint matching and pricing  
}
\label{sec:contention resolution}
Let $\assignprobs,
\{\assignvertexprobi,
\assignvertexprobj\}_{\rider \in \Rider,\driver\in\Driver}$ be the optimal solution of the convex program~\ref{eq:ex ante}. We use this solution to design a vector of prices $\{\price_\rider\}_{\rider\in\Rp}$ and matchings in both stages. Assume demand vertices are indexed so that the first stage is earlier in the ordering. First, consider this ``simulation'' process:
\begin{enumerate}[label=(\roman*)]
\item Take a pass over demand vertices $i=1,2,\ldots,\lvert\Rider\rvert$ one by one. For each demand vertex $\rider\in \Ra$, sample a supply vertex $\driver\primed\sim\{\assignprob_{\rider\driver}\}_{j\in D}$ independently and match $\rider$ to $\driver\primed$. For each demand vertex $\rider \in \Rp$, post her the price $\price_\rider = 
\Thresh_\rider\left(
\assignvertexprobi
\right)$. If accepted, sample vertex $\driver\primed\sim\{\tfrac{\assignprob_{\rider\driver}}{\assignvertexprob_{\rider}}\}_{j\in D}$
and match $\rider$ to $\driver\primed$.
\end{enumerate}
Step (\rom{1}) is a lossless randomized rounding in terms of preserving the objective value of the program~\ref{eq:ex ante};
however, while it guarantees that each demand vertex is matched to at most one supply vertex, it only guarantees each supply vertex is matched to at most one demand vertex \emph{in expectation}. Now, we run a separate process for each supply vertex $\driver$, as we run Step~(\rom{1}), in order to maintain the feasibility of the matching. Formally, consider this \emph{discarding} process: 
\begin{enumerate}[label=(\roman*)]
\setcounter{enumi}{1}
    \item When Step~(\rom{1}) recommends a demand vertex $\rider$ to be matched to a supply vertex $\driver$, the discarding algorithm corresponding to supply vertex  $\driver$ finalizes the match, that is, if vertex $\driver$ is already matched to another demand vertex, demand vertex $\rider$ will be left unmatched. Otherwise, demand vertex $\rider$ will be matched to supply vertex $\driver$ with probability $\tfrac{1}{
(2 - \sum_{\rider' < \rider}
\assignprob_{\rider'\driver})
}$.
\end{enumerate}
Step~(\rom{2}) not only guarantees that
each supply vertex is matched with at most one demand vertex, but also it guarantees that the matching happens with probability at least  $\tfrac{1}{2}$, when conditioned on the demand vertex $\rider$ is recommended to be matched to supply vertex $\driver$.

\begin{algorithm}[tbh]
\caption{Ex Ante Pricing with Simulation\&Discarding}
\label{alg:matching and pricing}
\begin{algorithmic}[1]
\State{Compute the optimal solution  $\assignprobs,
\{\assignvertexprob_k\}_{k \in \Rider\cup\Driver}$ of convex
program~\ref{eq:ex ante}. 
} 
\State{Initialize matchings 
$\M_1 \gets \emptyset$,
$\M_2 \gets \emptyset$ 
for both stages.
}
\State{Initialize $\theta_\driver \gets 0$
for all driver $\driver\in\Driver$.
}
\State{\textsl{---At the first stage---}}
\For{each demand vertex $\rider\in \Ra$ 
}
\State{Sample a supply vertex $\driver\primed$ independently from probability distribution 
$\{\assignprob_{\rider\driver}\}_{\driver\in\Driver}$.
}
\If{supply vertex $\driver\primed$
is not matched in $M_1$}
\State{
Flip a coin independently and match $\rider$ to $\driver\primed$,
i.e., $M_1 \gets M_1 \cup \{(\rider,\driver\primed)\}$, w.p. $\tfrac{1}{(2 - 
\theta_{\driver\primed}
)}$.} 
\EndIf
\State{$\theta_\driver\gets \theta_\driver + \assignprob_{\rider\driver}$
for \underline{all supply vertices} $\driver\in\Driver$.
}
\EndFor
\State{Offer the price 
$\price_\rider =
\Thresh_\rider(\assignvertexprobi)$
to each demand vertex $\rider \in \Rp$.
}
\State{\textsl{---At the second stage---}}

\State{Let $\Rpa$ be the set of demand vertices in $\Rp$ who accepts their prices, and $\tilde\Driver$ be the set of available supply vertices.
}
\State{Find the the maximum vertex weighted 
matching  $M_2$ in $G[\Rpa, \tilde\Driver]$,
where each demand vertex $\rider \in \Rpa$ has weight $\weight_\rider(\price_\rider)=\expect[\val\sim\Fu]{\val|\val \geq \price_\rider}$ and each supply vertex $\driver\in\tilde\Driver$ has weight $\weight_\driver$.}
\State{Return $M_1\cup M_2$.}
\end{algorithmic}
\end{algorithm}

\begin{restatable}{theorem}{ocr}
\label{thm:matching and pricing}
$\compratioMP(\Cref{alg:matching and pricing},\optoffmp)\geq \tfrac{1}{2}$, and no other matching/pricing policy can obtain a better competitive ratio against $\optoffmp$.
\end{restatable}
\begin{proof}{\emph{Proof of \Cref{thm:matching and pricing}.}}
Consider a variation of \Cref{alg:matching and pricing} that continues running the simulation and discarding procedures of the first stage during the second stage too, i.e., it runs the Step~(\rom{1}) and Step~(\rom{2}) exactly as in \Cref{sec:contention resolution} (therefore, variables $\{\theta_\driver\}_{\driver\in\Driver}$
also keep updating as we go over the riders of the second stage). Since both algorithms pick the same matching in the first stage, and \Cref{alg:matching and pricing} obtains no smaller market efficiency than this variation in the second stage (as it picks the maximum weighted matching), it is enough to focus on this variation and show its competitive ratio is at least $\tfrac{1}{2}$.

Let $\assignprobs,
\{\assignvertexprob_k\}_{k\in \Rider\cup\Driver}$ be the optimal solution of program~\ref{eq:ex ante}. Since this program gives 
an upper bound on the benchmark $\optoffmp$ 
(\Cref{lem:ex ante}), it is sufficient to show the following
two claims: (a)
whenever a demand vertex $\rider\in\Rp$
is matched during the second stage, her
expected 
value is $\weight_\rider(\Thresh_\rider(\assignvertexprobi))$,
and (b) 
every edge $(\rider, \driver)$
is matched 
with probability of at least $\tfrac{\assignprobi}{2}$.

Claim (a) is guaranteed since \Cref{alg:matching and pricing}
posts the price $\price_\rider = \Thresh(\assignvertexprobi)$
to each demand vertex $\rider\in \Rp$ in the second stage. To show claim (b), we first argue that for each demand vertex $\rider$, any supply vertex $\driver$ is recommended for matching with probability $\assignprobi$. For demand vertices of the first stage, 
this is trivial. For each demand vertex $\rider\in\Rp$,
she accepts her price with probability 
$\assignvertexprobi$,
and hence the supply vertex $\driver$
is selected with probability
$(\tfrac{\assignprobi}{\assignvertexprobi})\cdot \assignvertexprobi = \assignprobi$.
Next, note the following 
invariant of the discarding algorithm that can be proved with a simple reduction: in each iteration, the probability that supply vertex $\driver$ is not yet matched is equal to $1 - \tfrac{\theta_\driver}{2}$.
Thus, conditioning on the event that supply vertex $\driver$
is recommended to demand vertex $\rider$, they will be eventually matched with probability $\tfrac{1}{(2-\theta_{\driver})} \cdot 
(1 - \tfrac{\theta_\driver}{2}) = \tfrac{1}{2}.$ This finishes the argument for claim (b).

To show that the competitive ratio is optimal,
consider the following instance:
there is one supply vertex and two demand vertices.
The supply vertex has weight 0.
Demand vertex 1 is in the first stage with 
$\weightHat_1 = 1$.
Demand vertex 2 is in the second stage with
value $\val = \tfrac{1}{\epsilon}$
w.p. $\epsilon$ and $\val=0$ otherwise.
The expected market efficiency for any policy
is at most 1; however, $\optoffmp$ is $2-\epsilon$.\hfill \Halmos
\end{proof}

\revcolor{
\begin{remark}
\label{remark:online contention resolution}
Technically speaking,
the discarding algorithm in 
Step~(\rom{2}) for each supply vertex $\driver$ is 
an online contention resolution scheme
\citep[cf.][]{FSZ-16}
for $1$-uniform matroid. 
\citet{ala-14}
introduced the so-called ``Magician algorithm'', which is a near-optimal online contention resolution for $k$-uniform matroid. The algorithm in Step (\rom{2}) is a special case for $k=1$.
The overall approach (i.e.,
lossless randomized rounding with an 
inner discarding procedure to guarantee the 
feasibility) has also been used in other 
stochastic online optimization problems and combinatorial variants of prophet inequality
\citep[e.g.,][]{AHL-12,DFKL-17,MSZ-18,FNS-19}. The closest to us in this literature are \cite{AHL-12} and more recently, \cite{ezra2020online}\footnote{\revcolor{Notably, the main technical contribution of \cite{ezra2020online} is their interesting online contention resolution scheme for 
general graphs.
Most of the preceding work 
\citep[e.g.,][]{AHL-12,FNS-19}
only focus on 
bipartite graphs.}}, 
which use this technique for the prophet inequality matching problem under vertex arrival and edge arrival. While the problem studied in these papers is different from ours, the underlying algorithmic solutions are quite similar. We highlight that our contributions here are less on the technical side and mostly on the modeling side and connecting this algorithmic framework to the joint matching and pricing problem for maximizing market efficiency. 
\end{remark}}

\section{Numerical Experiments
for Two-stage Stochastic Matching}

\label{sec:numerical}

In this section, we provide numerical evaluations for the performance of 
\Cref{alg:weighted-skeleton} and
\Cref{alg:online} in practical instances of our two-stage matching problem, estimated from \citet{didi-20} dataset.

\subsection{Experimental setup}
\citet{didi-20}
is a dataset 
that aggregates anonymized 
trajectory data and 
ride request data of DiDi Express and DiDi Premier drivers within the ``Second Ring Road'' of the city of Chengdu in China
during November 2016. See \Cref{fig:chengdu} for a map showing the layout of this city and where DiDi has collected its data in this dataset. In our experiment, 
we focus on trajectory data and ride request data 
between 10:30am and 4:30pm of each day in the interval
from November 1 to November 14.
In \Cref{table:didi} we list all fields used in the 
experiment from \citet{didi-20}.

\begin{figure}[htb]
\centering
\includegraphics[width=0.4\textwidth]{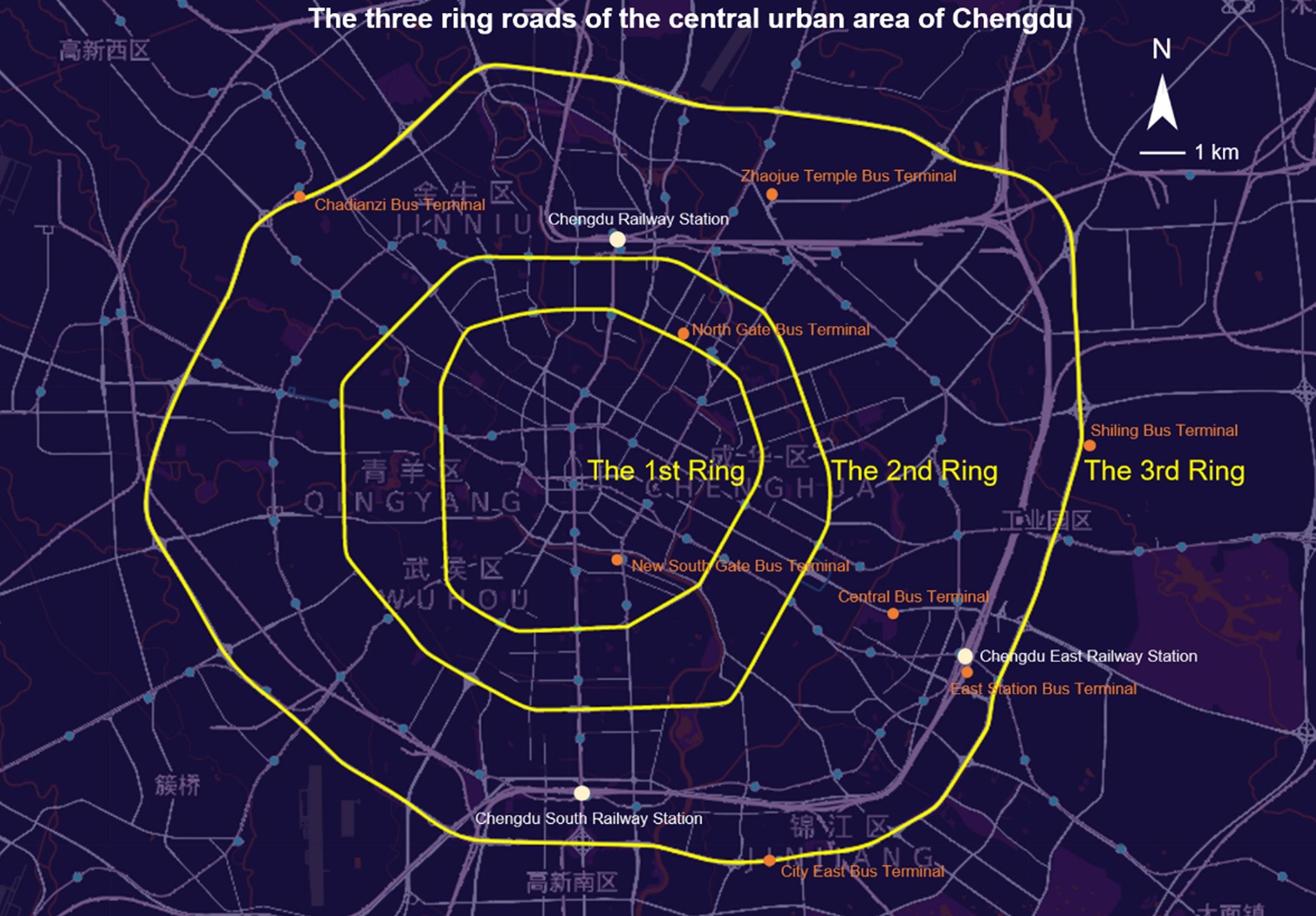}
\caption{\label{fig:chengdu}  Chengdu (source: OpenStreetMap; Chengdu Municipal Bureau of Planning and Natural Resources)}
\end{figure}

\begin{table}[h]

\centering
 \footnotesize
\subfloat[Trajectory Data]{
    \begin{tabular}{|c|c|}
    \hline
    Field & Comment \\
    \hline
      \texttt{Driver ID}  & Anonymized \\
    \hline
      \texttt{Order ID}  & Anonymized \\
    \hline
      \texttt{Time Stamp} & Unix timestamp, in seconds \\
    \hline
      \texttt{Longitude} & \\
    \hline
      \texttt{Latitude} & \\
    \hline
    \end{tabular}}
    \subfloat[Ride Request Data]{ 
    \begin{tabular}{|c|c|}
    \hline
    Field & Comment \\
    \hline
       \texttt{Order ID}  & Anonymized \\
    \hline
       \texttt{Ride Start Time} & Unix timestamp, in seconds \\
    \hline
       \texttt{Ride Stop Time} & Unix timestamp, in seconds \\
    \hline
       \texttt{Pick-up Longitude} & \\
    \hline
       \texttt{Pick-up Latitude} & \\
    \hline
    \end{tabular}
    }
    \normalsize
    \vspace{1mm}
    \caption{\citet{didi-20} fields description}
    \label{table:didi}
\end{table}
\newcommand{\sampletime}{\texttt{TIME}}
\newcommand{\sampleday}{\texttt{DAY}}

We run 500 iterations 
of Monte Carlo simulation to create different instances of our two-stage stochastic matching problem using the data. In each iteration, we sample a day $\sampleday$ uniformly at random from November 1 to November 14 and 
a time stamp $\sampletime$ uniformly at random from 10:30am to 4:30pm.
For this Monte Carlo iteration, the primitives of the instance are determined as follows:\footnote{The parameter choices below are for a better demonstration of the experimental results. We have confirmed that our conclusions are robust to these choices by running several other simulations---which we have omitted for brevity.}

\paragraph{Demand vertices $\Ra$ in the first stage.}
We construct demand vertex set $\Ra$
in the first stage 
with 
all ride requests whose \texttt{Ride Start Time}
is between $\sampletime$ and $\sampletime+10$ 
(i.e., a 10-second time interval)
on day $\sampleday$. 
See \Cref{fig:empirical vertex}
for a histogram of $|\Ra|$.

\paragraph{Demand vertices $\Rp$ and 
stochastic process of $\Rpa$ 
in the second stage.}
For all $d \in [14]$, let $\Rider_{2,d}$ be the set of 
all ride requests with
\texttt{Ride Start Time} between
$\sampletime + 10$ and $\sampletime + 20$ (i.e., forward shifted version of the previous time interval by $10$ seconds) on day $d$.
We set $\Rp = \displaystyle\cup_{d =1}^{14}\Rider_{2,d}$.
In this numerical experiment, 
we deviate from our model for the stochastic process
of $\Rpa\subseteq\Rp$.
In particular, instead of assuming 
that each demand vertex $\driver\in \Rp$
is available (i.e., $\driver\in \Rpa$)
with probability $\pi_\driver$ independently,
we assume $\Rpa = \Rider_{2,d}$ with 
equal probabilities for all $d\in [14]$.
Note that 
both \Cref{alg:weighted-skeleton}
and 
\Cref{alg:online}
are well-defined in this correlated arrival setting, and 
their competitive ratio guarantees are 
preserved under this deviation.
See \Cref{fig:empirical vertex} 
for a histogram of $|\Rpa|$.

\paragraph{Supply vertices $\Driver$.}
By using \texttt{Ride Stop Time} field, we construct the supply vertex set 
$\Driver$
with all the drivers who finish
a ride request between 
$\sampletime-20$ and $\sampletime$
(i.e., a 20-second time interval) 
on day $\sampleday$.
See \Cref{fig:empirical vertex}
for a histogram of $|\Driver|$.

\paragraph{Edge set $E$.}
For each pair of $\rider\in \Rider$
and $\driver \in \Driver$,
we compute the distance 
based on their longitude and latitude
information. 
In particular, recall that 
each $i\in \Rider$ corresponds 
to a ride request and hence
we can use \texttt{Pick-up Longitude}
and \texttt{Pick-up Latitude} 
from Ride Request Data as its location;
and 
each $j\in \Driver$ corresponds 
to a driver and hence
we can use \texttt{Longitude}
and \texttt{Latitude} at time $\sampletime$
on day $\sampleday$
from Trajectory Data as its location.
We add an edge between pair $(\rider,\driver)$
if their distance is below 2500m,
which corresponds to the quantile 82.4\% of the histogram of all pair distances.
See \Cref{fig:empirical distance}
for a histogram of all pair distances,
\Cref{fig:empirical edge}
for a histogram of edges,
and 
\Cref{fig:empirical location}
for the location of vertices (rider and driver) 
in one Monte Carlo iteration.

\newcommand{\idledist}{\mathcal{I}}
\newcommand{\idletime}{I}

\paragraph{Supply weights $\{\weight_\driver\}_{j\in\Driver}$.}
In this experiment, we study both 
scenarios where 
supply vertices 
are (i) unweighted
and (ii) weighted.
For the weighted scenario, 
we define the weight of a supply vertex
(corresponding to a driver) according to its
\emph{average idle time per ride}.
For each driver, 
we compute its 
average idle time per ride as 
the average time difference between
two adjacent rides of this driver
using \texttt{Ride Start Time} 
and \texttt{Ride Stop Time}.\footnote{
We assume that
a driver temporarily abandons the platform 
if she is idle for more than 600 seconds 
after a ride.
Thus, 
when we compute the average idle time per ride,
we only consider idle times below 600 seconds. 
}
See the histogram of the average idle time per ride in \Cref{fig:empirical weight}.
We compute the empirical distribution $\idledist$
of the average idle time per ride among all the drivers.
For a driver $\driver\in\Driver$ with 
average idle time per ride $\idletime$,
we define her weight 
as $\weight_\driver = 1 +
\idledist^{-1}(\idletime)$, where 
$\idledist^{-1}(\cdot)$ is the inverse of
cumulative function 
for the empirical distribution $\idledist$. Briefly speaking, this score aims to give more advantage to low-utilized drivers at the time of decision making.  

\begin{figure}
    \centering
    \subfloat[Histogram of vertices]{
\includegraphics[width=0.4\linewidth]{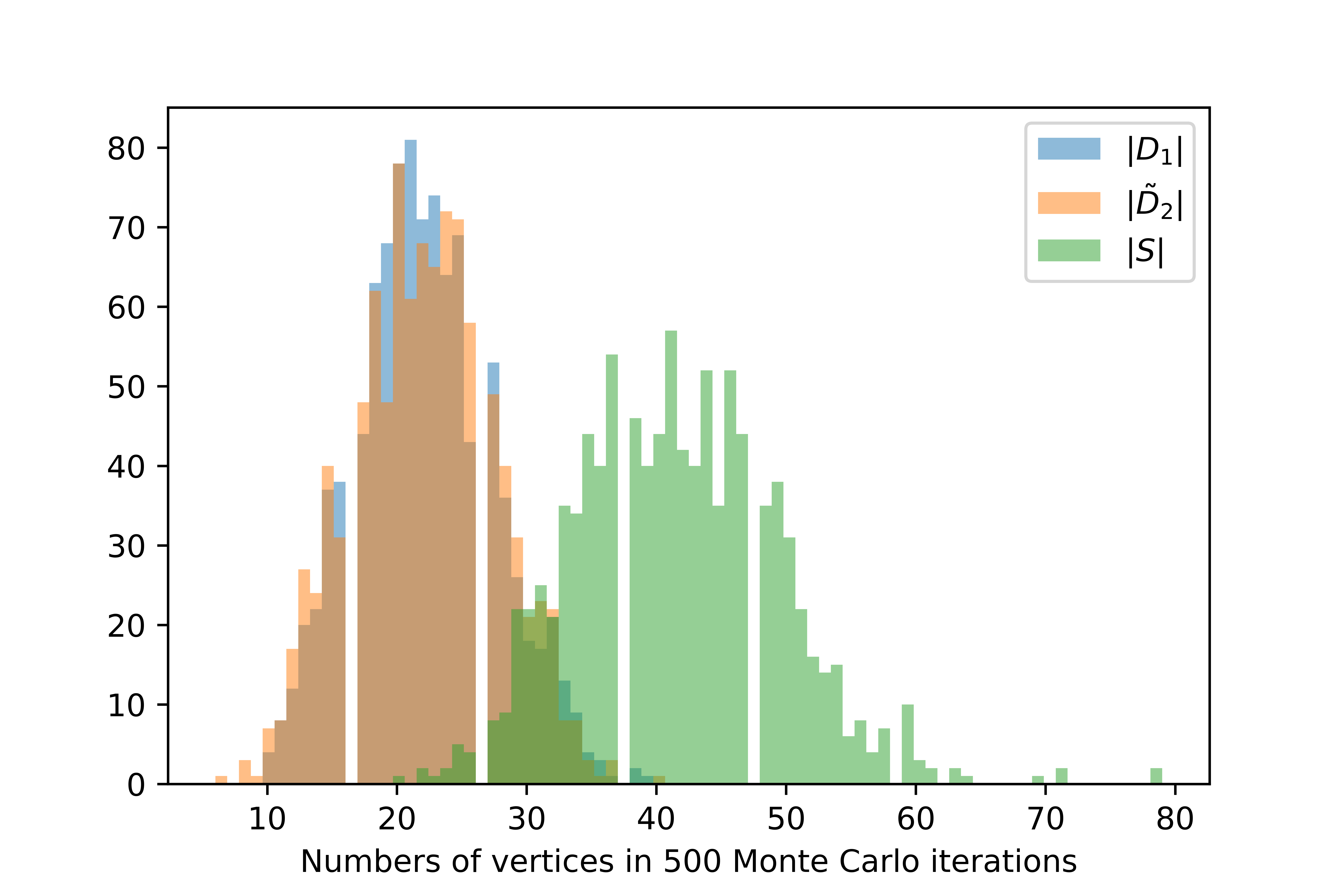}
\label{fig:empirical vertex}
}
\subfloat[Histogram of edges]{
\includegraphics[width=0.4\linewidth]{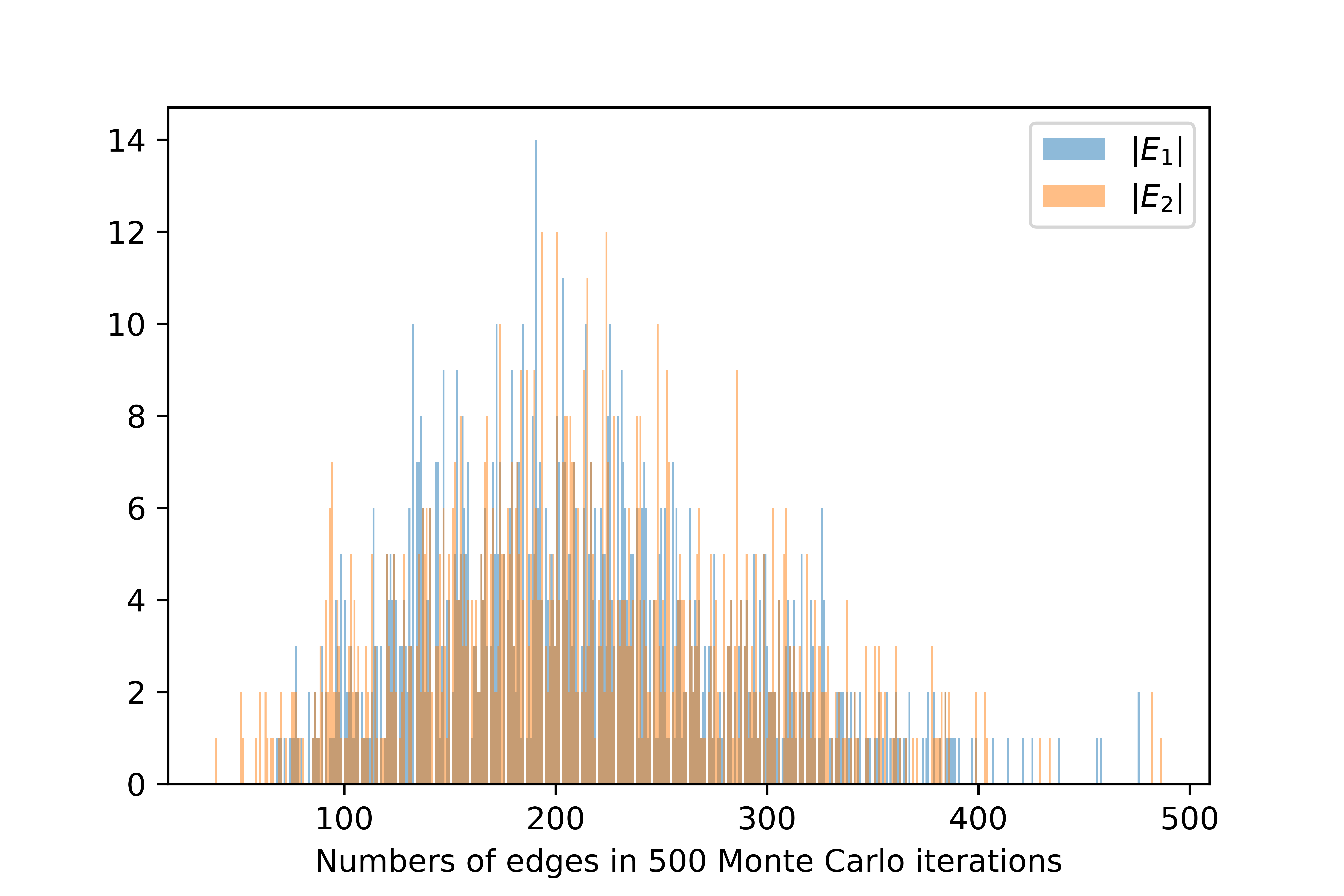}
\label{fig:empirical edge}
}
\\
\vspace{-4mm}
    \subfloat[Histogram of distances]{
\includegraphics[width=0.4\linewidth]{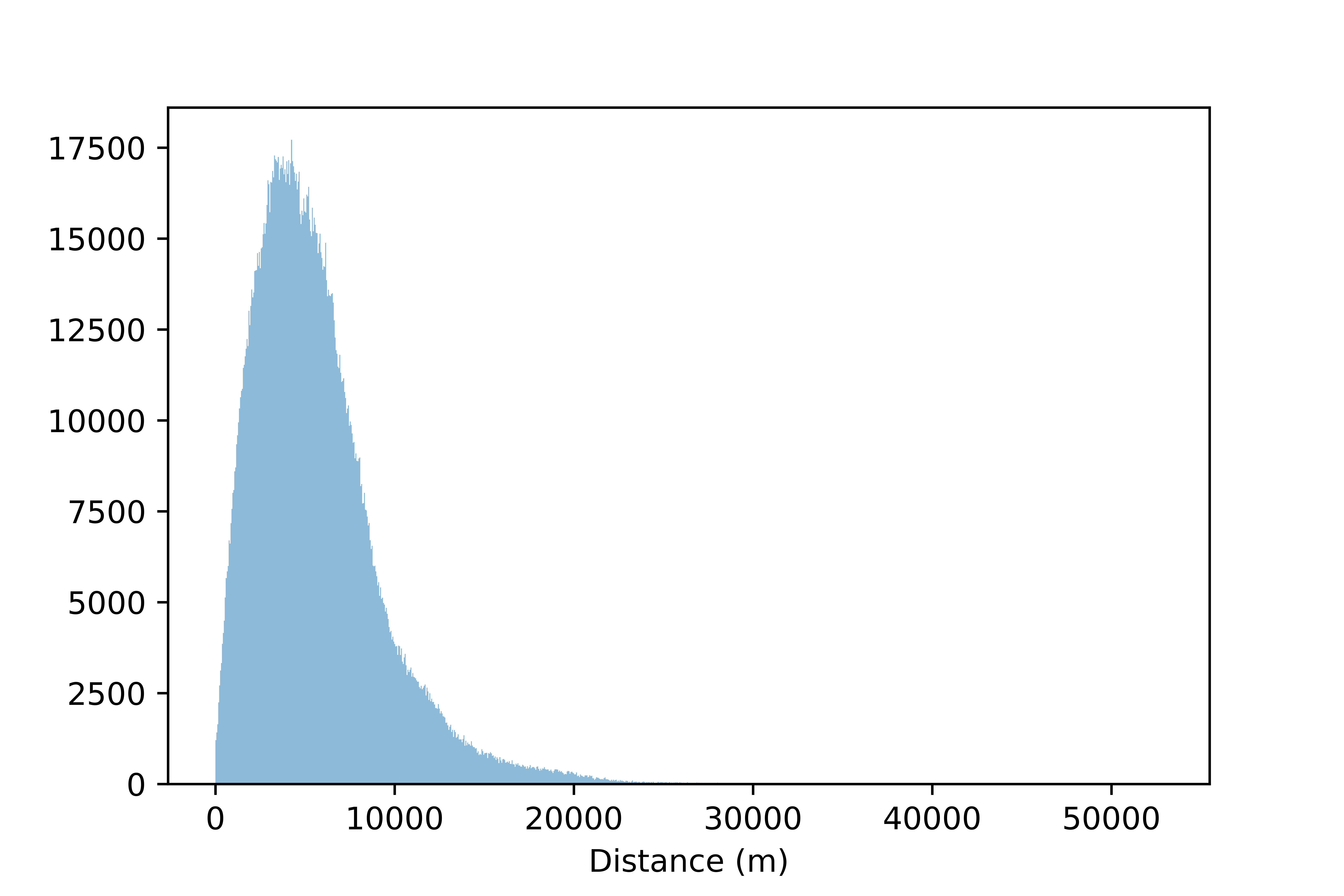}
\label{fig:empirical distance}
}
\subfloat[Histogram of average waiting time per ride]{
\includegraphics[width=0.4\linewidth]{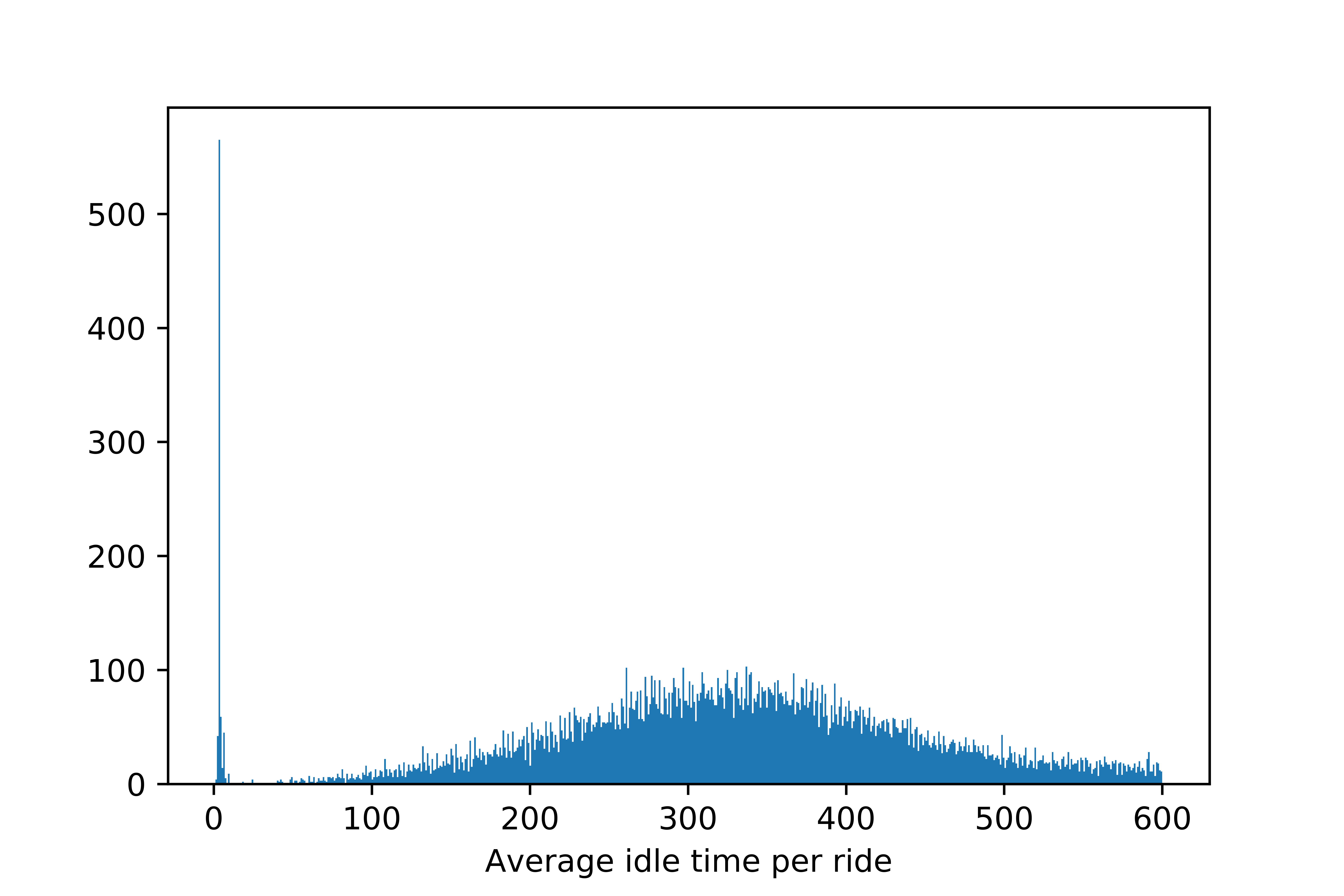}
\label{fig:empirical weight}
}
    \caption{Statistical data of instances used in our numerical experiment.}
\end{figure}
\begin{figure}
\vspace{-3mm}
    \centering
    \includegraphics[width=0.4\linewidth]{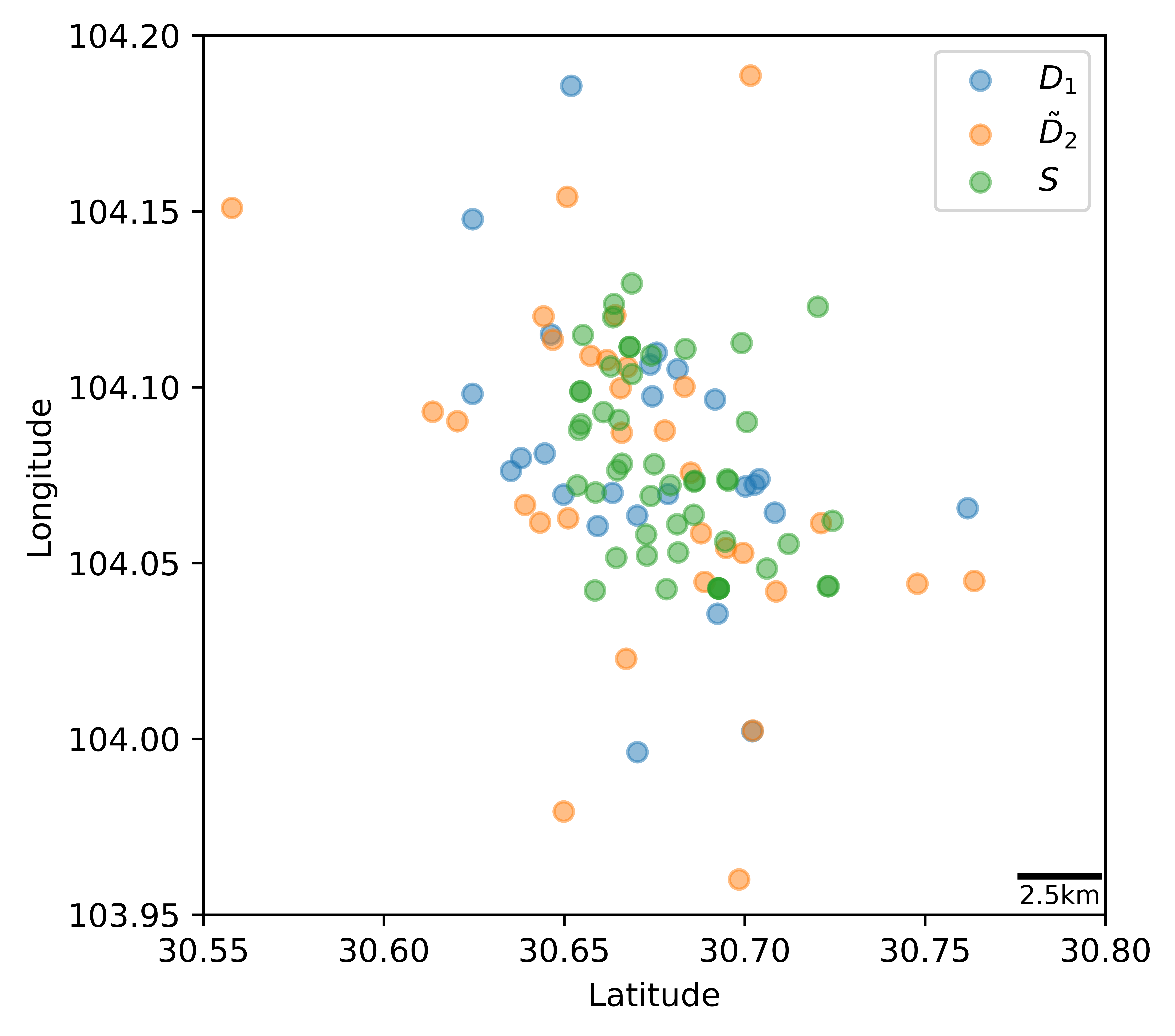}
    \caption{Location of vertices 
in one Monte Carlo iteration.}
\label{fig:empirical location}
\vspace{-3mm}
\end{figure}
\vspace{-1mm}

\subsection{Results and discussion}
To summarize our findings, 
we consider the two scenarios mentioned above
separately, together with five different policies: \Cref{alg:weighted-skeleton} (\texttt{WBU})
\revcolor{with convex function $g(x)=\exp(x)$}, two variants of submodular-maximization based policy (\texttt{SM} and \texttt{SM-Limit}), \Cref{alg:online} (\texttt{HG}) and greedy (\texttt{GR}). See Appendix~\ref{apx:numerical-policies} for details on our policies.  

\paragraph{(i) Supply vertices with no weights.}
In this scenario,
it is clear from 
both \Cref{tab:numerical cc unweighted}
and 
\Cref{fig:empirical cc unweighted}
that all policies 
achieve relatively high
performance ratios against the optimum offline.\footnote{One potential explanation is that \citet{didi-20}
only records ride requests that 
have been served in the platform,
which may be only a subset of 
all possible ones.
This issue downplays
the power of optimum offline in our experiments.}
Both \texttt{WBU} and \texttt{GR} 
output $M_1$ in the first stage 
without any knowledge of the second stage graph. 
The average performance ratio gap between
\texttt{WBU} and
\texttt{GR} is 2\% (\Cref{tab:numerical cc unweighted}).
The performance ratios 
of all 500 iterations of 
the Monte Carlo simulation are less
concentrated for \texttt{GR}
than for \texttt{WBU} and other policies
(\Cref{fig:empirical cc unweighted}).
\texttt{SM}, \texttt{SM-Limit} and 
\texttt{HG} all utilize and rely on the distributional knowledge 
about the second stage graph. 
All three policies achieve better performance ratios
than \texttt{WBU} and \texttt{GR}.

\begin{table}[hbt]
    \centering
    \caption{Comparing performance ratios of different policies against optimum offline; supply vertices have no weights;
    policies with $*$ are proposed or analyzed in this paper.}
    \vspace{1mm}
    \footnotesize
    \label{tab:numerical cc unweighted}
    \begin{tabular}{cc}
    \hline
    policy & average performance ratio (95\% CI) \\
    \hline
    \texttt{WBU}$^*$    &  0.985 (0.966,0.998) \\
    \texttt{SM}$^*$    & 0.998 (0.991,1.000) \\
    \texttt{SM-Limit}$^*$    & 0.994 (0.981,1.000) \\
    \texttt{HG}$^*$    & 0.998 (0.991,1.000) \\
    \texttt{GR}    & 0.965 (0.916,0.997) \\
    \hline
    \end{tabular}
\end{table}
\begin{figure}
\vspace{-9mm}
    \centering
    \subfloat[Supply vertices with no weights]{
\includegraphics[width=0.5\linewidth]{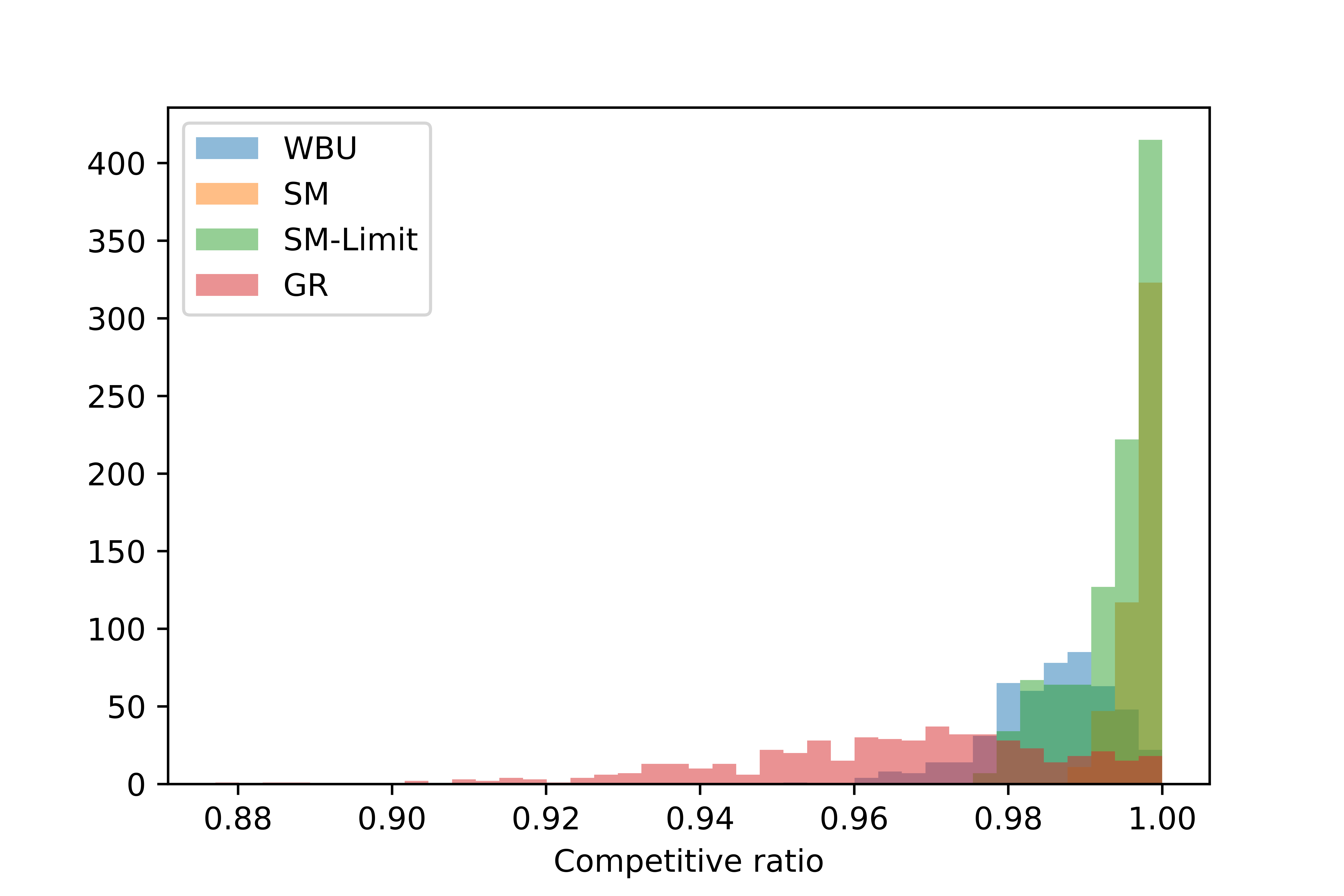}
\label{fig:empirical cc unweighted}
}
\subfloat[Supply vertices with weights]{
\includegraphics[width=0.5\linewidth]{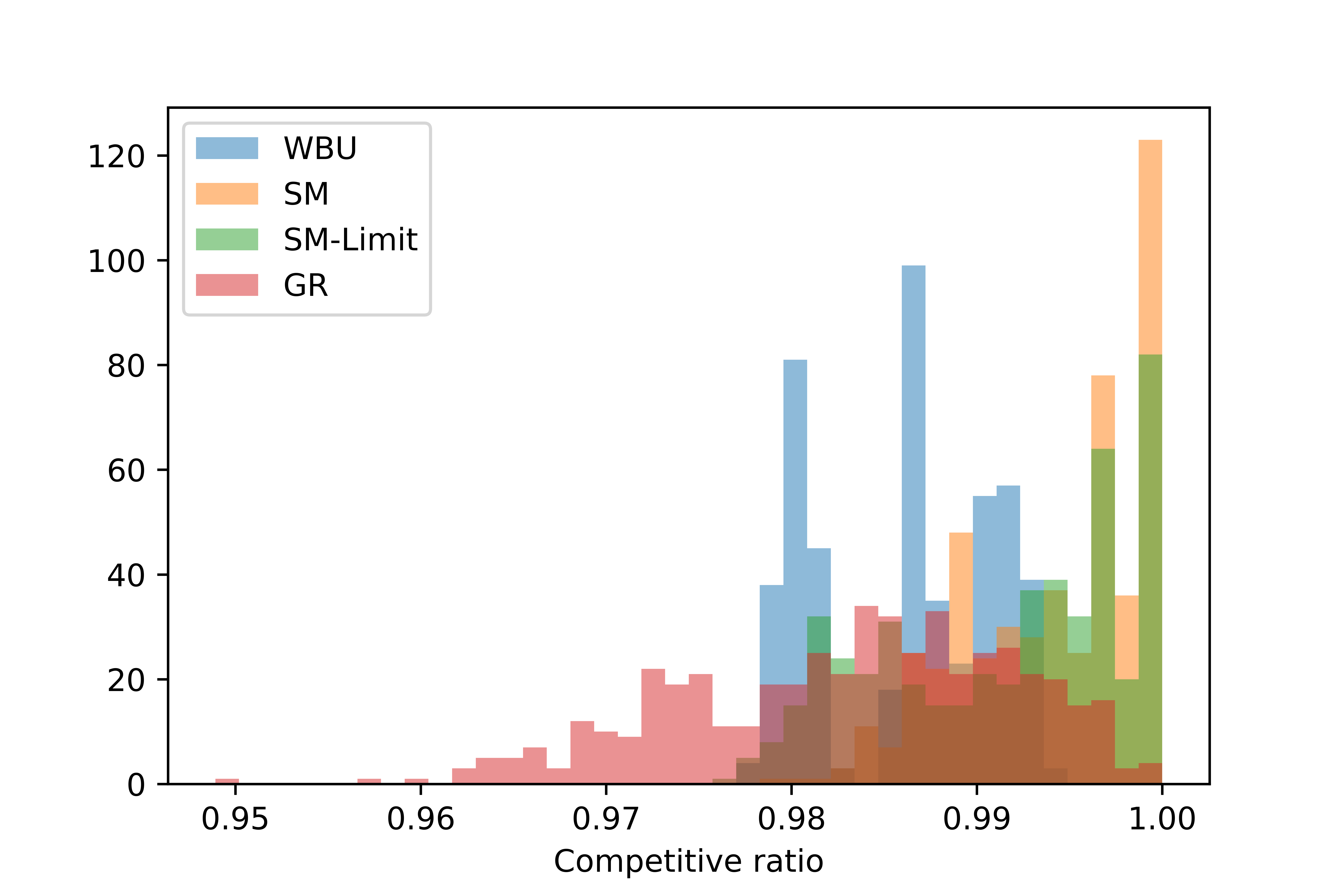}
\label{fig:empirical cc weighted}
}
    \caption{Histogram of the performance ratios of different policies against optimum offline.
    Results are based on 500 iterations of Monte Carlo simulation.}
    \vspace{-2mm}
\end{figure}
\paragraph{(ii) Supply vertices with weights.}
In this scenario, similar to the previous scenario,
it is clear from 
both \Cref{tab:numerical cc weighted}
and 
\Cref{fig:empirical cc weighted}
that all policies 
achieve relatively high
performance ratios against optimum offline.
Again, both \texttt{WBU} and \texttt{GR} 
output $M_1$ 
without any knowledge of the second stage. 
The average performance ratio gap between
\texttt{WBU} and
\texttt{GR} is 0.3\% (\Cref{tab:numerical cc weighted}).
The performance ratios 
of all 500 iterations of 
the Monte Carlo simulation are less
concentrated for \texttt{GR}
than for \texttt{WBU} and other policies
(\Cref{fig:empirical cc weighted}).
\texttt{SM}, \texttt{SM-Limit} and 
\texttt{HG} again achieve better performance ratios
versus \texttt{WBU} and \texttt{GR}.

\begin{table}[hbt]
\vspace{-4mm}
    \centering
    \caption{Comparing performance ratios of different policies against optimum offline; supply vertices has weights;
    policies with $*$ are proposed or analyzed in this paper.}
    \label{tab:numerical cc weighted}
    \footnotesize
    \vspace{1mm}
    \begin{tabular}{cc}
    \hline
    policy & average performance ratio (95\% CI) \\
    \hline
    \texttt{WBU}$^*$    & 0.986 (0.979,0.993) \\
    \texttt{SM}$^*$    & 0.994 (0.984,1.000) \\
    \texttt{SM-Limit}$^*$    & 0.992 (0.979,1.000) \\
    \texttt{HG}$^*$    & 0.994 (0.984,1.000) \\
    \texttt{GR}    & 0.983 (0.964,0.997) \\
    \hline
    \end{tabular}
    \vspace{-4mm}
\end{table}
\begin{remark}
To conduct a numerical study for joint
matching and pricing,
we needed additional data
to fit demand vertices' private values
(willingness to pay),
which is missing in \citet{didi-20}.
\end{remark}

\begin{remark}
See \Cref{sec:conclusion}
for a discussion on how to 
obtain heuristic policies 
based on \Cref{alg:weighted-skeleton}
and \Cref{alg:online}
for multi-stage stochastic matching, with the potentials of being applicable in ride-hailing.
\end{remark}

 \section{Conclusion}
 \label{sec:conclusion}

We have studied the competitive analysis for the weighted two-stage stochastic matching problem for maximizing supply efficiency, and the two-stage stochastic joint matching and pricing problem for maximizing market efficiency. We have also considered the competitive analysis with general edge weights and single-stage stochastic joint matching and pricing for demand efficiency. We developed methods to systematically use graph structure and price information, as well as methods to use the price lever, to obtain optimal competitive ratios for these problems. See \Cref{tab:summary} for more details. 

\paragraph{Future directions.} There are several open problems and directions emerging from this paper. \revcolor{In particular, characterizing the approximation hardness of the optimum online policy and obtaining tight competitive ratios against this benchmark for the supply-weighted two-stage stochastic matching are of interest}. Another important generalization of our $\sfrac{3}{4}$-competitive algorithm is considering the extension to the multi-stage settings.
\revcolor{The recent work of \cite{feng2020batching,feng2021batching} extends \Cref{thm:opt-offline-comp-ratio} to $K$ stages for \emph{fractional matchings} and shows the tight competitive ratio of $1-(1-1/K)^K$. Showing competitive ratio in the integral matching problem for $K>2$ is still an open problem.} 

Note that while our proposed algorithms are defined for two-stage stochastic matching, one can design single lookahead heuristics for the multi-stage version of the same problem based on our algorithms. Such a heuristic algorithm, at each stage $k$, only considers the information about stage $k+1$ (and not other future stages) and declares a matching for stage $k$ by running either of our \Cref{alg:weighted-skeleton} or \Cref{alg:online}. While it is not clear whether such an algorithm is theoretically competitive in the multi-stage problem, the prospect of running this $1$-lookahead algorithm for infinite horizon problems make it appealing for several demand-supply matching applications. Studying their numerical performance in such applications is an important applied future direction.

\setlength{\bibsep}{0.0pt}
\bibliographystyle{plainnat}
\OneAndAHalfSpacedXI
{\footnotesize
\bibliography{refs}}

\clearpage

\ECSwitch
\ECDisclaimer

\renewcommand{\theHsection}{A\arabic{section}}
\renewcommand{\theHsection}{A\arabic{chapter}}



\section{Further Related Work}
\label{appendix-further-related}
 \noindent\emph{Two-stage 
stochastic and non-stochastic combinatorial optimization.}
Two-stage stochastic combinatorial
optimization problems with linear and convex objectives
have been studied extensively in the 
literature
\citep[e.g.,][]{BL-11, CCP-05,SS-04,SS-07,IKMM-04, hanasusanto2015k,hanasusanto2018conic}.
Also, see the excellent book of \cite{kuhn2006generalized} for more context regarding convex multi-stage stochastic programs. For two-stage stochastic matching, 
various models with different objectives
have been studied in 
\citet{KS-06, EGMS-10, KKU-08}. 
Our result for the unweighted version of the two-stage matching is closely related to the beautiful work of \cite{LS-17}, which also uses the matching skeleton of \cite{GKK-12}, but for a different two-stage matching problem. In their model, a batch of new edges arrive adversarially in the second stage, and their algorithm first finds a maximum fractional matching using the matching skeleton and then removes some of the edges in the first stage through a randomized procedure to obtain a $\tfrac{2}{3}$-competitive ratio. Another work that is conceptually related to our two-stage matching problem is the recent work of \cite{housni2020matching}, which is inspired by the idea of robust multi-stage optimization~\citep[e.g.,][]{bertsimas2010optimality,bertsimas2011theory}. They consider a two-stage robust optimization for cost-minimization of the matching, where future demand uncertainty is modeled using a set of demand scenarios (specified explicitly or implicitly). Our model and results diverge drastically from this work, which makes them incomparable. 

\vspace{3mm}
\noindent\emph{Online bipartite allocations.}
Our results for weighted two-stage stochastic matching (\Cref{sec:matching})
is comparable with results in
\emph{online bipartite matching 
with vertex arrival}. In a related line of work, \citet{AGKM-11, DJK-13} generalize the classic RANKING algorithm to the setting where offline vertices are weighted. Also, at high-level, our primal-dual analysis resembles~\cite{BJN-07,devanur2009adwords,FKMMP-09, devanur2012online,HKTWZZ-18,huang2020fully} in spirit. For the online stochastic matching, besides work mentioned earlier, \cite{brubach2016new} slightly improve the previous bound in \cite{JL-14}. Another related line of research is stochastic online packing and convex programming~\citep{FMMS-10,DJSW-11,AD-14,dughmi2021bernoulli}. Finally, another related variant of online bipartite matching is online assortment, both under the adversarial and the Bayesian settings. For example, \citet{CF-09} study a model with non-stationary consumer types (in the language of assortment optimization) and show $\tfrac{1}{2}$ performance guarantee
with respect to the clairvoyant optimum online benchmark. \citet{GNR-14} introduce the ``inventory balancing" algorithm---inspired by the BALANCE algorithm in the seminal work of \cite{MSVV-07} for online ad allocations---and analyze their performance guarantee in the adversarial setting using the primal-dual approach. This analysis is later improved and generalized by \cite{MS-20}, and extended to variants of online assortment by~\cite{FNS-19,goyal2020online,feng2021online}.

\vspace{3mm}
\noindent\emph{Prophet inequality.}
The results in two-stage joint matching/pricing (\Cref{sec:matching and pricing}) resembles some aspects of  ex-ante prophet inequalities for single item and matroid environment~\citep{DFKL-17,LS-18,CHMS-10,KW-12,correa2017posted}, prophet inequality matching \citep{AHL-12}, the Magician's problem~\citep{ala-14}, the pricing with static calendar problem~\citep{MSZ-18} (which also studies static assortment policies without reusable resources), and the volunteer crowdsourcing problem~\citep{manshadi2020online}. Similar ex ante relaxation programs have been used for various stochastic online optimization and mechanism design problems. For example, see \cite{AHL-12,FSZ-16,DFKL-17,LS-18,VB-18,MSZ-18,anari2019nearly}. See \cite{correa2019recent} for a survey on recent developments in the literature on prophet inequalities. At high-level, the main ideas of online contention resolution schemes resemble aspects of the technique introduced in \cite{A-07} for approximate dynamic programming. See \cite{SBAPW-04} for a comprehensive survey.

\vspace{3mm}
\noindent \emph{Continuous-time dynamic stochastic matching/queuing with applications to ride-sharing.} There has been a growing line of work on continuous-time dynamic stochastic matching and dispatching models pertaining to ride-sharing and kidney exchange~\cite[e.g.,][]{ashlagi2018effect,akbarpour2020thickness,ata2020dynamic,hu2021dynamic}. In particular, in the heavy-traffic regime,  \cite{gurvich2015dynamic} study a matching control problem with the objective of minimizing the finite-horizon holding costs. \cite{ozkan2020dynamic} study a dynamic stochastic matching problem with heterogeneous supply and demand types, where demand nodes should be matched immediately, but supply nodes can wait for matching until they depart. More recently, \cite{aveklouris2021matching} study the same bipartite model where both sides can arrive and renege. One takeaway of this line of work is that in the high-traffic regime, reviewing-based batched matching policies can obtain asymptotic near-optimal performance. In the small traffic regime, \cite{aouad2020dynamic} study a continuous time dynamic stochastic matching model where supply and demand have type dependant weights/costs and their arrivals and abandonments are stochastic and heterogeneous with given rates. While our work conceptually resembles some aspects in these work (e.g., the idea of batching improving the performance of matching policy is seen in \citealp{aveklouris2021matching}), it diverges from this line of work on several fronts. In particular, in our two-stage discrete-time model the supply is static, the second stage graph in most of our analysis is adversarial (or can be realized from an arbitrary correlated distribution over possible graphs), and we consider both optimum offline and online benchmarks when supply side has vertex weights. Moreover, our proposed algorithms either do not need any distributional knowledge, or only rely on simulating the second stage graph to compute their matchings in the first stage.


 \section{Missing Proofs of Section~\ref{sec:matching}}
 \label{app:sec3}

\subsection{Omitted Proof for Proposition~\ref{prop:optoffline-upper}}
\thmoptofflineupper*
\label{apx:optoffline-upper}
\begin{proof}{\emph{Proof.}}
Consider an example where $\lvert\Ra\rvert=n$, $\lvert \Driver\rvert=2n$, and each demand vertex in $\Ra$ has an edge to each supply vertex in $\Driver$. Moreover, $\lvert\Rp\rvert=2n$, and each demand vertex in $\Rp$ has exactly one edge to one supply vertex in $D$. Let $\pi_i=1/2$ for each demand vertex $i\in\Rp$. Due to the symmetry, a policy that matches $\Ra$ to any subset of size $n$ of $2n$ supply vertices is optimal. Such a policy has a matching size equal to $n$ in the first stage. Moreover, in the second stage, it matches each of the remaining $n$ supply vertices to their neighbouring demand vertices in $\Rp$ once the neighbouring demand vertex is available, which happens w.p.\ $\tfrac{1}{2}$. Therefore, the expected size of second stage matching is $\tfrac{n}{2}$, and in total the expected size of the matching is $\tfrac{3n}{2}$.

Now, note that $\expect{\lvert\Rpa\rvert}=n$. For any $\epsilon>0$, $\lvert \Rpa \rvert\in [n(1- \epsilon),n(1+\epsilon)]$  with probability at least $1-2\exp(-\tfrac{\epsilon^2 n}{3})$ by applying Chernoff bound. The omniscient policy knows \emph{exactly} which vertices are in $\Rpa$ (and hence which subset of $2n$ supply vertices should be left unmatched for the second stage). Therefore, it obtains a matching of size at least $n+n(1-\epsilon)$ when $\lvert \Rpa \rvert\in [n(1-\epsilon),n(1+\epsilon)]$. The expected size of the final matching picked by omniscient is at least $(2n-\epsilon n)\left(1-2\exp(-\tfrac{\epsilon^2 n}{3})\right)$. By setting $\epsilon=\sqrt{\tfrac{3\log n}{n}}$, the competitive  ratio is at most:
$$
\frac{\tfrac{3n}{2}}{(2n-\sqrt{3n\log n})(1-\tfrac{2}{n})}=\tfrac{3}{4}+o(1)~,
$$
which finishes the proof.\hfill\Halmos
\end{proof}

\subsection{Omitted Proof for Lemma~\ref{lemma:structure}}
\label{apx:structure}
In this subsection, we prove  \Cref{lemma:structure} using convex programming techniques. The key idea to show the structural decomposition of \Cref{lemma:structure} is employing the standard \emph{Lagrangian duality} for the convex program \ref{eq:matching-convex-program}, and then identifying properties of the optimal primal/dual solutions by using \emph{Karush–Kuhn–Tucker (KKT)} conditions (for more details on convex duality, see ~\cite{BV-04}). 

\structural*

\begin{definition}
\label{def:lagrange}
The \emph{Lagrangian dual} $\Lag^g(.)$ for the convex program \ref{eq:matching-convex-program} is defined as:
\begin{align*}
    \Lag^g(\xbf,\lambdabf,\thetabf,\gammabf)\triangleq& \displaystyle\sum_{j\in \Driver}\frac{1}{w_j}g\left(w_j(1-\sum_{i\in \Ra}x_{ij})\right)+\sum_{i\in\Ra}\lambda_i\left(\sum_{j\in N(i)}x_{ij}-1\right)+\sum_{j\in \Driver}\theta_j\left(\sum_{i\in N(j)}x_{ij}-1\right)
    \\
    &-\sum_{i\in\Ra,j\in \Driver:(i,j)\in E}\gamma_{ij}x_{ij}~~,
\end{align*}
where $\{\lambda_i\}$, $\{\theta_j\}$, and $\{\gamma_{ij}\}$ are the \emph{Lagrange multipliers}, also known as the \emph{dual variables}, corresponding to different constraints of \ref{eq:matching-convex-program}.
\end{definition}
\begin{proposition}[\emph{Strong duality/KKT conditions}]
\label{prop:kkt}
Suppose $\xbf^*$ is the optimal solution of the convex program \ref{eq:matching-convex-program}. Consider the Lagrangian dual $\Lag^g$, as in \Cref{def:lagrange}. Then:
\begin{equation*}
    \underset{\xbf}{\min}~\underset{\lambdabf,\thetabf,\gammabf}{\max}~\Lag^g(\xbf,\lambdabf,\thetabf,\gammabf)= \underset{\lambdabf,\thetabf,\gammabf}{\max}~\underset{\xbf}{\min}~\Lag^g(\xbf,\lambdabf,\thetabf,\gammabf)=\underset{\lambdabf,\thetabf,\gammabf}{\max}~\Lag^g(\xbf^*,\lambdabf,\thetabf,\gammabf)~~.
\end{equation*}
Moreover, if $\lambdabf^*$, $\thetabf^*$, and $\gammabf^*$ (together with $\xbf^*)$ are the solutions of the above minmax/maxmin, then:
\begin{itemize}
    \item (Stationarity)~~$\forall i\in\Ra,j\in \Driver,(i,j)\in E:~~\frac{\partial \Lag^g(\xbf^*,\lambdabf^*,\thetabf^*,\gammabf^*)}{\partial x_{ij}}=0,$
    \item (Dual feasibility)~~$\forall i\in\Ra,j\in \Driver,(i,j)\in E:~\lambda^*_i\geq 0,\theta^*_j\geq 0$, and $\gamma^*_{ij}\geq 0$~,
    \item (Complementary slackness)~~$\forall i\in\Ra,j\in \Driver,(i,j)\in E:$ 
    \begin{itemize}
        \item  $\gamma^*_{ij}=0$ or $\left(\gamma^*_{ij}>0~\textrm{and}~ x^*_{ij}=0\right)~,$
        \item $\lambda^*_{i}=0$ or $\left(\lambda^*_{i}>0~\textrm{and}~ \sum_{j\in N(i)}x^*_{ij}=1\right)~,$
        \item  $\theta^*_{j}=0$ or $\left(\theta^*_{j}>0~\textrm{and}~ \sum_{i\in N(j)}x^*_{ij}=1\right)~.$
    \end{itemize}
\end{itemize}
\end{proposition}
With the above ingredients, we prove the structural lemma.
\begin{proof}{\emph{Proof of \Cref{lemma:structure}.}} To show the uniformity property,  first note that for every $j\in \Driver^{(0)}$, $w_j(1-\sum_{i\in N(j)}x^*_{ij})=0=c^{(0)}$. 
Second, for $l\geq1$, as $G'[\Rideri,\Driveri]$ is a connected subgraph,
it is sufficient to consider 
$j,j'\in\Driver^{(l)}$ 
where
there exists a demand vertex $i\in\Rider^{(l)}$, such that $x^*_{ij}>0$ and $x^*_{ij'}>0$. Therefore $\gamma^*_{ij}=\gamma^*_{ij'}=0$ due to the complementary slackness in \Cref{prop:kkt}. Also, $\theta^*_j=\theta^*_{j'}=0$, again because of complementary slackness and the fact that $\sum_{i\in N(j)}x^*_{ij}<1$ and $\sum_{i\in N(j')}x^*_{ij'}<1$.
Now, applying the stationarity condition in \Cref{prop:kkt}, we have:
$$
-\tfrac{1}{w_j}\cdot w_jg'\left(w_j(1-\sum_{i\in N(j)}x^*_{ij})\right)+\lambda^*_i+\theta^*_j-\gamma^*_{ij}=-\tfrac{1}{w_{j'}}\cdot w_{j'}g'\left(w_{j'}(1-\sum_{i\in N(j')}x^*_{ij'})\right)+\lambda^*_i+\theta^*_{j'}-\gamma^*_{ij'}=0~,
$$
and therefore $g'\left(w_j(1-\sum_{i\in N(j)}x^*_{ij})\right)=g'\left(w_{j'}(1-\sum_{i\in N(j')}x^*_{ij'})\right)$. Because $g'(.)$ is strictly increasing (as $g(.)$ is strictly convex), $w_j(1-\sum_{i\in N(j)}x^*_{ij})=w_{j'}(1-\sum_{i\in N(j')}x^*_{ij'})$, finishing the proof of the uniformity property. Also note that $\lambda^*_i=g'\left(w_j(1-\sum_{i\in N(j)}x^*_{ij})\right)>0$, as $g(.)$ is non-decreasing and $g'(.)$ is strictly increasing. Again, by using complementary slackness, we have $\sum_{j\in N(i)}x^*_{ij}=1$, which proves the saturation property.

Now suppose there exists an edge $(i,j)\in E$, where $i\in \Rider^{(l)}$ and $j\in \Driver^{(l')}$, for $l,l'\in\{0,1,\dots,L\}, l\neq l'$. First note that there should exist a supply vertex $j'\in\Driver^{(l)}$ so that $x^*_{ij'}>0$ (because $i$ is either fully matched by $\xbf^*$ if $l\geq 1$, thanks to the saturation property, or $l=0$ and $i$ is connected to a supply vertex  $j'\in\Driver^{(0)}$ with $x^*_{ij'}>0$ by definition). Because of the complementary slackness, we have  $\gamma^*_{ij'}=0$. Now, by writing the stationarity condition of \Cref{prop:kkt} for the edge $(i,j)$, and noting that $w_j(1-\sum_{i'\in N(j)}x^*_{i'j})=c^{(l')}$ because of the uniformity property, we have:
$$
-g'(c^{(l')})+\lambda^*_i+\theta^*_j-\gamma^*_{ij}=0~.
$$
By writing the same condition for the edge $(i,j')$ we have: 
$$
-g'(c^{(l)})+\lambda^*_i+\theta^*_{j'}=0~,
$$
and therefore, 
$$
-g'(c^{(l)})\leq -g'(c^{(l)})+\theta^*_{j'}=-g'(c^{(l')})+\theta^*_j-\gamma^*_{ij}\leq -g'(c^{(l')})+\theta^*_j~.
$$
If $l'=0$, then clearly $c^{(l)}\geq 0= c^{(l')}$. If $l'>0$, then $\sum_{i'\in N(j)}x^*_{i'j}<1$ and $\theta^*_j=0$ because of complementary slackness. Therefore, $g'(c^{(l)})\geq g'(c^{(l')})$. This implies $c^{(l)}\geq c^{(l')}$ because of the convexity of $g(.)$, which finishes the proof of the monotonicity property.\hfill\Halmos
\end{proof}


\subsection{Omitted Proof for Theorem~\ref{thm:no-fptas}}

 In this subsection, we show \Cref{thm:no-fptas} by
 a reduction from the \emph{max $k$-cover} problem, for which the hardness of approximations is known~\citep{fei-98}.

 \begin{definition}[\emph{max $k$-cover}]
 \label{def:max-k-cover}
 Given a collection $\Cvector=\{C_1,C_2,\ldots,C_m\}$ of subsets of $\{1,\ldots,n\}$, \emph{max $k$-cover} is the problem of finding $k$ subsets from $\Cvector$, such that their union has maximum cardinality.  
 \end{definition}
\begin{proposition}[\citealp{fei-98}, Theorem~5.3]
\label{prop:feige-paper}
 For any $\epsilon>0$, max $k$-cover cannot be approximated in polynomial time within a ratio of $(1-\tfrac{1}{e}+\epsilon)$, unless $P=NP$.
\end{proposition}

We use \Cref{prop:feige-paper} to prove the following theorem. 

\label{apx:no-fptas}
\thmfptas*
\begin{proof}{\emph{Proof.}} 
Assume an FPTAS for computing $\opton$ exists, that is, a $(1-\epsilon)$ approximation algorithm for any $\epsilon>0$ with running time polynomial in $\lvert \Driver \rvert$, $\lvert \Rider \rvert$, and $\tfrac{1}{\epsilon}$. Without loss of generality, assume this algorithm picks a maximum unweighted matching in the first stage. By setting an appropriate $\epsilon$, we use this algorithm as a subroutine to obtain a $\gamma$-approximation algorithm for max $k$-cover for $\gamma>1-\tfrac{1}{e}$. \Cref{prop:feige-paper} (Theorem 5.3 in \cite{fei-98}) then implies that P=NP. 

Given an instance of max $k$-cover as in \Cref{def:max-k-cover}, consider an instance of our problem where $\lvert\Ra\rvert=m-k$, $\Rp=\{1,\ldots,n\}$, and $\Driver=\{1,\ldots,m\}$. Add one edge for each pair of vertices in $\Ra\times \Driver$. Also, add an edge $(i,j)\in\Rp\times \Driver$ iff $i\in C_j$ for all $i,j$. Let $\pi_i=\delta, i\in\Rp$ for some $\delta\in(0,\tfrac{1}{2})$, which will be fixed later in the proof. Finally, let all the supply vertices have unit weights. Now, let $i\in\Rp: I_i\sim \textrm{Bernoulli}(\delta)$ be independent Bernoulli random variables indicating whether each demand vertex in $\Rp$ is available in the second stage or not. Then: 
\begin{align}
    f^\wbf(T)&=\expect[\Rpa]{\tilde\rho^\wbf(T)}=\expect[\Rpa]{\tilde\rho^\wbf(T)\mathbbm{1}\left\{\exists i: I_i=1\right\}}\nonumber\\
    &=\delta(1-\delta)^{n-1}\sum_i\expect[\Rpa]{\tilde\rho^\wbf(T)|I_i=1, I_j=0,j\neq i}+ \binom{n}{2}\delta^2(1-\delta)^{n-2}\expect[\Rpa]{\tilde\rho^\wbf(T)|\sum_i I_i>1}\nonumber\\
    \label{eq:cover-function}
    &\leq \delta(1-\delta)^{n-1}\left(\sum_i\expect[\Rpa]{\tilde\rho^\wbf(T)|I_i=1, I_j=0,j\neq i}+ n^3\delta\right)~,
\end{align}
where the last inequality holds as $\tilde \rho^\wbf(T)\leq n$ and $\delta\leq \tfrac{1}{2}$. Note that
$$
\expect[\Rpa]{\tilde\rho^\wbf(T)|I_i=1, I_j=0,j\neq i}=\mathbbm{1}\left\{i\in\bigcup_{j\in T} C_j \right\}~,
$$
and therefore inequality~\ref{eq:cover-function} can be rewritten as
\begin{equation}
\label{eq:upper-lower-cover}
\forall T\subseteq\{1,\ldots,m\}:~~~  \delta(1-\delta)^{n-1}\left\lvert\bigcup_{j\in T} C_j\right\rvert  \leq f^\wbf(T)\leq \delta(1-\delta)^{n-1}\left(\left\lvert\bigcup_{j\in T} C_j\right\rvert+n^3\delta\right)~.
\end{equation}
Suppose $\texttt{OPT}$ is the optimal value of max $k$-cover, i.e., $\texttt{OPT}=\underset{T\subseteq \{1,\ldots,m\}:\lvert T\rvert=k}\max~\left\lvert\bigcup_{j\in T} C_j\right\rvert$, and $\hat T$ is the set of unmatched supply vertices by our FPTAS algorithm for some $\epsilon$ (will be fixed later). This algorithm matches all demand vertices in $\Ra$ to a subset of $m-k$ supply vertices in $D$ during the first stage, and therefore $\lvert \hat T \rvert=k$. By applying \Cref{prop:opt-online-characterization},  we have:
\begin{equation}
\label{eq:fptas}
   \forall T,\lvert T\rvert=k: f^\wbf(\hat{T})+m-k\geq (1-\epsilon)\opton\geq  (1-\epsilon)f^{\wbf}(T)+(1-\epsilon)(m-k)~.
\end{equation}
Combining \cref{eq:upper-lower-cover} and \cref{eq:fptas}, $\forall T:\lvert T\rvert=k$ we have:
\begin{equation}
\delta(1-\delta)^{n-1}\left(\left\lvert\bigcup_{j\in \hat T} C_j\right\rvert+n^3\delta\right)+m-k\geq (1-\epsilon)\delta(1-\delta)^{n-1}\left\lvert\bigcup_{j\in T} C_j\right\rvert+(1-\epsilon)(m-k)~.
\end{equation}
By focusing on the set achieving $\texttt{OPT}$ in the max $k$-cover instance, and rearranging the terms, we have:
\begin{equation}
    \left\lvert\bigcup_{j\in \hat T} C_j\right\rvert\geq (1-\epsilon)\texttt{OPT}-n^3\delta-\frac{\epsilon}{\delta(1-\delta)^{n-1}}(m-k)>\left(1-\epsilon -n^3\delta -\frac{m\epsilon}{\delta(1-\delta)^{n-1}}\right)\texttt{OPT}~,
\end{equation}
where the last inequality holds as $\texttt{OPT}\geq 1$. Finally, set $\delta=\tfrac{c}{n^3}, \epsilon=\tfrac{c^2}{mn^3}$ for an arbitrary small constant $c>0$. Using standard algebraic bounds we get:
\begin{equation*}
\left(1-\epsilon -n^3\delta -\frac{m\epsilon}{\delta(1-\delta)^{n-1}}\right)\texttt{OPT}>\left(1-c^2-c-2c\right)\texttt{OPT}.
\end{equation*}
So, the algorithm outputting $\hat{T}$ has an approximation ratio  strictly greater than $(1-\tfrac{1}{e})$ in the max $k$-cover (for small enough constant $c$), while running time is still polynomial in $m$ and $n$.\hfill\Halmos
\end{proof}

 \revcolor{
\section{Optimal Solution in
Convex Program~\ref{eq:matching-convex-program}}
\label{apx:identical solution for different}

In \Cref{thm:opt-offline-comp-ratio},
we show that for any 
differentiable, monotone increasing 
and strictly convex function $g(\cdot)$,
by solving 
the convex program~\ref{eq:matching-convex-program} with $g(\cdot)$
in the objective function,
\Cref{alg:weighted-skeleton}
obtains an optimal competitive ratio against
the optimum offline.
One natural question 
is how to choose this function $g(\cdot)$.
From a computational point of view we may prefer
strongly convex, 
smooth convex functions $g$,
since 
iterative methods for solving convex optimization (e.g., first-order methods)
have generally faster convergence rates
in that case for 
solving 
the 
convex 
program~\ref{eq:matching-convex-program}.
However, as we will illustrate in this section,
on the theoretical side,
there is no difference as 
all such choices of the function $g$
lead to the same optimal solution 
in the convex program~\ref{eq:matching-convex-program}.

\begin{theorem}
\label{thm:identical solution for different g}
Fix two arbitrary
differentiable,
monotone increasing and strictly convex functions $\hat g(\cdot)$ and $g(\cdot)$.
If $\{x_{ij}^*\}$
is an optimal solution
in the convex program~\ref{eq:matching-convex-program}
with function $\hat g(\cdot)$,
$\{x_{ij}^*\}$
is also an optimal solution
in the convex program~\ref{eq:matching-convex-program}
with function $g(\cdot)$.
\end{theorem}

\begin{proof}{\emph{Proof.}} 
It is sufficient to construct 
a dual assignment 
$\{\lambda_i^*,
\theta_j^*, 
\gamma_{ij}^*\}$
for Lagrangian dual $\Lag^{g}$
that 
satisfies the KKT condition
with $\{x_{ij}^*\}$. 

Let $\{(\Rideri, \Driveri)\}_{\mspair = 0}^{ 
\totalmspair}$ be the structural decomposition
(\Cref{lemma:structure})
under $\{x_{ij}^*\}$ (the pairs of the decomposition are indexed so that $0=\mspairc^{(0)}<\mspairc^{(1)}<\ldots<
\mspairc^{(\totalmspair)}$).
Consider the following dual assignment:
\begin{align*}
\forall
\mspair\in[0:\totalmspair],
\forall
i\in \Rideri:\qquad&
\lambda_i^* \leftarrow
g'(\mspairci)
\cdot \indicator{\mspair \not = 0}
\\
\forall
\mspair\in[0:\totalmspair],
\forall j\in\Driveri:\qquad&
\theta_j^* \leftarrow
g'(\mspairci)
\cdot \indicator{\mspair  = 0}
\\
\forall \mspair\in[0:\totalmspair],
\forall \mspair'\in[0:\mspair],
\forall i\in\Rideri, j \in \Driver^{(\mspair')}, (i,j)\in E:\qquad&
\gamma_{ij} = g'(\mspairci)
\cdot \indicator{\mspair \not = 0} - 
g'(\mspairc^{(l')})
\cdot \indicator{\mspair' \not = 0}
\end{align*}
Note that the ``Monotonicity'' property 
in \Cref{lemma:structure}
ensures that there does not exists edges 
from $\Rideri$ to $\Driver^{(\mspair')}$
where $\mspair' > \mspair$. Hence,
these dual assignments (in particular $\{\gamma_{ij}^*\}$) are well-defined.

First, we consider the dual feasibility.
Since function $g(\cdot)$ 
is monotone increasing, 
$\{\lambda_i^*,\theta_j^*\}$
is non-negative. 
Since $\mspairci \geq \mspairc^{(\mspair')}$
for all $\mspair \geq \mspair'$
and function $g'(\cdot)$
is monotone increasing, 
$\{\gamma_{ij}^*\}$
is non-negative.

Next, we consider the complementary slackness.
The ``Saturation'' property in \Cref{lemma:structure}
ensures the complementary slackness
for $\{\lambda_i^*\}$ for all 
$i\in
\bigcup_{\mspair\in[\totalmspair]}\Rideri$.
The complementary slackness
for $\{\lambda_i^*\}$ for all 
$i\in\Rider^{(0)}$ holds by construction.
The definition of $\Driver^{(0)}$
ensures
the complementary slackness
for $\{\theta_j^*\}$ for all 
$j\in\Driver^{(0)}$.
The complementary slackness
for $\{\theta_j^*\}$ for all 
$j\in
\bigcup_{\mspair\in[\totalmspair]}\Driveri$
holds by construction.
The definition of the
structural decomposition 
$\{(\Rideri, \Driveri)\}_{\mspair = 0}^{ 
\totalmspair}$
guarantees that $x_{ij}^* = 0$
for all $i\in\Rideri,j\in\Driver^{(\mspair')}$
with $\mspair > \mspair'$.
Hence, the complementary slackness 
for $\{\gamma_{ij}^*\}$ holds.

Finally, we consider the stationarity.
Fix an arbitrary edge $(i, j)\in E$.
Suppose $i\in \Rideri$ and $j\in \Driver^{(\mspair')}$.
We then have
\begin{align*}
    &\frac{\partial\Lag^g(\xbf^*,\lambdabf^*,\thetabf^*,\gammabf^*)}{\partial x_{ij}}\\
    =&
    -\tfrac{1}{w_j}\cdot w_jg'\left(w_j(1-\sum_{i\in N(j)}x^*_{ij})\right)+\lambda^*_i+\theta^*_j-\gamma^*_{ij}
    \\
    =&
    -g'(\mspairc^{(\mspair')})
    +
    g'(\mspairci)
\cdot \indicator{\mspair  \not= 0}
+
g'(\mspairc^{(\mspair')})
\cdot \indicator{\mspair'  = 0}
-(
g'(\mspairci)
\cdot \indicator{\mspair \not = 0} - 
g'(\mspairc^{(l')})
\cdot \indicator{\mspair' \not = 0}
)
\\
=&~0
\end{align*}
Invoking KKT condition (\Cref{prop:kkt})
finishes the proof.
\hfill\halmos
\end{proof}}

 \revcolor{
\section{Robustness of \Cref{alg:weighted-skeleton} to Approximations and Modeling Assumptions}
\label{apx:robustness}

In this section, we study
the robustness of \Cref{alg:weighted-skeleton}
(i.e., how its competitive ratio changes) to various modeling aspects of the two-stage stochastic matching problem. In particular, we consider its robustness to:
 
\begin{itemize}
    \item Replacing the matching algorithm in the second stage with an approximation algorithm instead of an exactly optimal algorithm (Appendix~\ref{apx:robustness-approx}). 
\item Imposing additional structural assumptions on the first-stage graph (Appendix~\ref{apx:robustness-first}).
\end{itemize}

\subsection{Approximately 
optimal matching in the second stage}
\label{apx:robustness-approx}
In this sub-section, we restrict our attention to 
unweighted instances of the two-stage stochastic matching problem, as in \Cref{sec:matching}.
 One of the key ingredients in 
\Cref{alg:weighted-skeleton} 
is the randomized matching $M_1$, which is sampled from the solution of the convex program~\ref{eq:matching-convex-program}.
Using this randomized matching $M_1$
in the first stage,
and by outputting
the maximum matching $M_2$
between $\Rpa$ and the remaining unmatched vertices of $\Driver$ in the second stage,
\Cref{alg:weighted-skeleton} attains 
the competitive ratio $\frac{3}{4}$
(\Cref{thm:opt-offline-comp-ratio}).
One natural question is whether 
a similar competitive ratio guarantee exists 
if the platform chooses a $\secondstageratio$-approximately maximum matching between $\Rpa$ and the remaining unmatched vertices of $\Driver$ in the second stage, where $\beta\in[0,1]$. We summarize such an approach in \Cref{alg:weighted-skeleton approx}.
Without any further assumption on 
the program instance, the competitive ratio of any algorithm that uses a $\beta$-approximately maximum matching in the second stage is at most $\secondstageratio$. To see this, consider an instance when there is no demand vertex in the first stage. Therefore, we impose the following assumption on the problem instances, which roughly speaking places a lower-bound on the of contribution of the first-stage graph to the overall maximum matching. In the rest of this section, starting from this assumption, we fix the parameter $\gamma\in[0,1]$.
\begin{assumption}
\label{asp:matching ratio in first stage}
Let $M_1^*$ denote the maximum matching 
between $\Ra$ and $\Driver$,
and $M^*$ denote the maximum matching 
between $\Ra\cup\Rpa$ and $\Driver$.
For each realization $\Rpa$,
$\frac{|M_1^*|}{|M^*|} \geq \firststagepercent$.
\end{assumption}

\begin{algorithm}
\revcolor{
\begin{algorithmic}[1]
\State{\textbf{input:} bipartite graph $G=(\Rider,\Driver,E)$, convex function $g(\cdot)$~.}
\State{\textbf{output:} bipartite matching $M_1$ in $G[\Ra,\Driver]$,  bipartite matching $M_2$ in $G[\Rpa,\Driver]$.}
\vspace{2mm}
\State{Solve convex program \ref{eq:matching-convex-program} to obtain $\{x^*_{ij}\}$.}
\State{Sample matching $M_1$ with edge marginal probabilities $\{x^*_{ij}\}$.}
\State{In the second stage, return a $\secondstageratio$-approximately
maximum matching $M_2$ between $\Rpa$ and the remaining unmatched vertices of $\Driver$.}
\end{algorithmic}
\caption{\label{alg:weighted-skeleton approx}\textsc{\revcolor{Weighted-Balanced-Utilization 
with $\secondstageratio$-approximation in second stage}}}}
\end{algorithm}

\begin{theorem}
\label{thm:opt-offline-comp-ratio approx second stage}
For any differentiable, monotone increasing and strictly convex function $g:\R\rightarrow \R$,
under \Cref{asp:matching ratio in first stage} with $\firststagepercent\in[0, 1]$,
\Cref{alg:weighted-skeleton approx}
is 
$(\firststagepercent+
\secondstageratio(1-\firststagepercent)^2)$-competitive 
against $\optoff$.
\end{theorem}

\begin{remark}
\label{remark:robust approx 1}
Note that there always exists $\underline{\gamma}$, so that for $\gamma\in(\underline{\gamma},1]$ the obtained competitive ratio in \Cref{thm:opt-offline-comp-ratio approx second stage} is strictly larger than $\beta$. As an implication, 
when $\secondstageratio = 1-\frac{1}{e}$
(i.e., the approximation guarantee 
from the RANKING algorithm in \citealp{KVV-90}),
the competitive ratio of
\Cref{alg:weighted-skeleton approx}
is strictly greater than $1-\frac{1}{e}$
for all $\firststagepercent>
\frac{e-2}{e-1}\approx 0.418$. This implies that using \Cref{alg:weighted-skeleton} in the first stage and RANKING in the second stage, one can obtain a competitive ratio strictly larger than $1-\frac{1}{e}$ even in the case where second stage vertices arrive in an online fashion (under the assumption that $\gamma>\tfrac{e-2}{e-1}$). With no assumption on $\gamma$, running RANKING in both stages will result in a weaker competitive ratio of $1-\frac{1}{e}$.
\end{remark}
\begin{remark}
\label{remark:robust approx 2}
As another implication for the case when $\beta=1$, the worst-case competitive ratio is $\tfrac{3}{4}$ and happens at $\gamma=\tfrac{1}{2}$. For any other value of $\gamma$, either closer to $0$ or $1$, the obtained competitive ratio is strictly larger than $\tfrac{3}{4}$, and converges to $1$ as $\gamma$ converges to either $0$ or $1$. This implies that an \emph{unbalanced} instance, in which the first-stage graph either has a very large or a very small contribution to the final maximum matching, is an easier instance for our convex-programming based matching algorithm (\Cref{alg:weighted-skeleton}), and hence we can obtain an improved competitive ratio for such an instance. 
\end{remark}
In the proof of \Cref{thm:opt-offline-comp-ratio approx second stage},
we use a similar technique 
as the one in the proof of \Cref{thm:opt-online-comp-ratio}
--
we use a {factor revealing (FR) program}, that is, a non-linear program whose optimal solution is a lower-bound on the competitive ratio of \Cref{alg:weighted-skeleton approx}
against the optimum offline.

\begin{proof}{\emph{Proof of \Cref{thm:opt-offline-comp-ratio approx second stage}}.}
The proof is done in three major steps:

\newcommand{\led}{^{(\mspair)}}
\newcommand{\msopt}{o}
\newcommand{\msopti}{\msopt\led}
\newcommand{\Driveropt}{O}

\noindent{\textbf{Step 1- writing a factor revealing program:}} Consider the following optimization program, which is parameterized by $\totalmspair\in\N$, and has variables 
$\{(\msdriveri, \msopti,
\mssecondi, \msrideri)\}
_{\mspair\in[\totalmspair]}$:
\begin{align}
    \label{eq:offline approx second stage}
\begin{array}{llll}
    \min\limits_{
\{(\msdriveri, \msopti,
\mssecondi, \msrideri)\}
_{\mspair\in[\totalmspair]}
    }
    &
    \frac{
    \MS
    }{
    \OPT
    }
    & \text{s.t.} & \\
    &
    \mssecondi\leq \msopti \leq \msdriveri,
    \ 
    \msrideri \leq \msdriveri
    & 
    \mspair\in [\totalmspair] &\textit{\footnotesize{(Feasibility)}} \\
    ~&~&~&~
    \\
    &
    \displaystyle\sum_{\mspair\in[\totalmspair]}
    (\msopti - \msrideri)
    = 
    (1-\firststagepercent)
    \displaystyle\sum_{\mspair\in[\totalmspair]}
    \msopti
    \ 
    & 
     &\textit{\footnotesize{(Approximation)}} \\
    &
    \mspairci < \mspairc^{(\mspair + 1)}
    &
    \mspair\in [\totalmspair - 1] &  \textit{\footnotesize{(Monotonicity-1)}}
    \\
    &
     \displaystyle\sum_{\mspair' = \mspair}^{\totalmspair}
     (
     \msopt^{(\mspair')} -
    \mssecond^{(\mspair')}
    )
    \leq 
     \displaystyle\sum_{\mspair' = \mspair}^{\totalmspair}
    \msrider^{(\mspair')}
    & 
    \mspair \in [\totalmspair] &  \textit{\footnotesize{(Monotonicity-2)}}
    \\
    &
    \mspairci  
    =
    \frac{\msdriveri - \msrideri}{\msdriveri}
    &
    \mspair \in [\totalmspair] & 
    \\
    ~&~&~&~
    \\
    &
    \msdriveri,\msopti,
    \mssecondi, \msrideri\in \N 
    &
    \mspair\in[\totalmspair] & 
\end{array}
\end{align}
where the auxiliary variables $\MS$
and $\OPT$ are defined as:
\begin{equation}
\label{eq:auxilaryvars approx}
\MS \triangleq
\firststagepercent\sum_{\mspair\in[\totalmspair]}
\msopti
+
\secondstageratio
\sum_{\mspair\in[\totalmspair]}
\mspairci\mssecondi,
~~~~
\OPT  \triangleq
\sum_{\mspair\in[\totalmspair]}\msopti
\end{equation}
Fixing an arbitrary unweighted 
two-stage matching instance
and an arbitrary realization $\Rpa$,
we show how to map it
to a feasible 
solution of program~\eqref{eq:offline approx second stage},
and we show the competitive ratio
of \Cref{alg:weighted-skeleton approx} against 
the optimum offline
in this realized instance 
is at least the objective value of 
the obtained feasible solution of program~\eqref{eq:offline approx second stage}.
The mapping is similar to the one in
the proof of \Cref{lem:online final}.

\emph{The desired mapping.} 
Consider any unweighted 
two-stage stochastic matching instance
$\left\{(\Ra\cup \Rp, \Driver,
E), \{\pi_u\}_{u\in \Rp}\right\}$.
Fix an arbitrary realization $\Rpa\subseteq \Rp$.
Let $\Driveropt \subseteq \Driver$ be 
the subset of
supply vertices matched by the optimum offline policy,
and $\Driversecond\subseteq \Driveropt$
be the subset of
supply vertices matched by the optimum offline policy
in the second stage.
Let $\{(\Rideri, \Driveri)\}_{\mspair = 0}^{ 
\totalmspair}$ be the structural decomposition of this instance as in \Cref{lemma:structure} (the pairs of the decomposition are indexed so that $0=\mspairc^{(0)}<\mspairc^{(1)}<\ldots<
\mspairc^{(\totalmspair)}$).
We now construct the following solution 
for program~\eqref{eq:offline approx second stage}:
\begin{gather*}
\text{for all }
\mspair\in[\totalmspair]:\quad
\msdriveri \leftarrow
\lvert\Driveri\rvert,\quad
\msopti \leftarrow
\lvert\Driveri\cap \Driveropt\rvert,\quad
\mssecondi \leftarrow
\lvert\Driveri \cap \Driversecond\rvert,\quad
\msrideri \leftarrow
\lvert\Rideri\rvert~.
\end{gather*}

\emph{Objective value of the constructed solution.}
We first formulate the competitive ratio of \Cref{alg:weighted-skeleton approx} 
against the optimum offline
policy on the original instance
with the fixed realization $\Rpa$. 
The size of the matching of 
the optimum offline policy is 
\begin{equation}
\label{eq:opt-bound approx}
    \sum_{\mspair\in[0:\totalmspairHat]}
    \lvert\Driveri\cap\Driveropt\rvert 
    =
    \lvert\Driver^{(0)}\rvert
    +\OPT~,
\end{equation}
where $\OPT$ is defined 
in \eqref{eq:auxilaryvars approx}. 
In the above equation, we used the fact that $\Driver^{(0)}\subseteq \Driveropt$, 
simply because vertices in $\Driver^{(0)}$ 
should be matched in the optimum offline policy. 

Now consider the performance of \Cref{alg:weighted-skeleton approx}.
By \Cref{asp:matching ratio in first stage}
and 
\Cref{rem:observation},
the total number of supply vertices 
matched in the first stage of 
\Cref{alg:weighted-skeleton approx}
is at least $|\Driver^{(0)}| 
+ \firststagepercent\sum_{\mspair\in[\totalmspair]}
|\Driveri\cap\Driveropt|$.
To find a lower-bound for the expected total
number of supply vertices matched in the second stage
of \Cref{alg:weighted-skeleton approx}, 
first consider the 
maximum 
matching between realized second stage demand vertices and unmatched supply vertices. Such a matching is no smaller than the subset of edges of the matching picked by the optimum offline policy corresponding to the supply vertices not matched by \Cref{alg:weighted-skeleton approx} during the first stage. 
By the structural decomposition in \Cref{lemma:structure},
for every supply vertex in $\Driveri$,
the probability that 
it is not matched 
in the first stage
is $\mspairci$.
Thus, 
the expected size of such a projected matching 
is $\sum_{\mspair\in[\totalmspair]}
\mspairci
|\Driveri \cap \Driversecond|$, 
due to the linearity of the expectation.
Combining with the fact that 
\Cref{alg:weighted-skeleton approx}
picks a $\secondstageratio$-approximately
maximum matching in the second stage,
we can conclude that the expected total size of the final matching picked by
\Cref{alg:weighted-skeleton approx} at the end of the second stage will be at least
\begin{equation}
\label{eq:ms-bound approx}
   |\Driver^{(0)}| 
+ \firststagepercent\sum_{\mspair\in[\totalmspair]}
|\Driveri\cap\Driveropt|
+
\secondstageratio
\sum_{\mspair\in[\totalmspair]}
\mspairci 
|\Driveri \cap \Driversecond|
   =
   \lvert\Driver^{(0)}\rvert
    +\MS~,
\end{equation}
where $\MS$ is defined in \eqref{eq:auxilaryvars approx}.

Putting the bounds in \eqref{eq:opt-bound approx}, and \eqref{eq:ms-bound approx} together, the competitive ratio 
of \Cref{alg:weighted-skeleton approx}
against the optimum offline policy on the original
two-stage matching instance
under an arbitrary fixed realization $\Rpa$
is at least 
\begin{align*}
    \frac{
    \lvert\Driver^{(0)}\rvert
    +
    \MS
    }{
    |\Driver^{(0)}|
    +
    \OPT
    }
    \geq
    \frac{
    \MS
    }{\OPT}~,
\end{align*}
where the 
ratio 
on the right-hand side
is the objective value of
our constructed solution in program~\eqref{eq:offline approx second stage}.

\emph{Feasibility of the constructed solution:}
Constraint \emph{(Feasibility)} of
program~\eqref{eq:offline approx second stage} 
holds by construction.
Note that there exists a maximum matching 
between $\Ra\cup\Rpa$ and $\Driver$
where all demand vertices in
$\bigcup_{\mspair\in[\totalmspair]}\Rideri$
are matched. 
Hence, constraint \emph{(Approximation)}
holds due to \Cref{asp:matching ratio in first stage}.
Because of \Cref{lem:msc closed form}
and ``Monotonicity'' property of the structural decomposition in
\Cref{lemma:structure},
constraint \emph{(Monotonicity-1)} holds by construction.
Constraint \emph{(Monotonicity-2)} also holds because of the following argument:
$\sum_{\mspair'=\mspair}^{\totalmspair}
(\msopt^{(\mspair')} - \mssecond^{(\mspair')})$ 
is the 
number of supply vertices in 
$\bigcup_{\mspair'=\mspair}^{\totalmspair}
\Driver^{(\mspair')}$
whom are matched in the optimum offline policy 
during the first-stage. Moreover,
$\sum_{\mspair'=\mspair}^{\totalmspair}
\msrider^{(\mspair')}$ is the number
of the demand vertices in $\bigcup_{\mspair'=\mspair}^{\totalmspair}
\Rider^{(\mspair')}$.
Since ``Monotonicity'' property of the structural decomposition in \Cref{lemma:structure} for the original two-stage matching instance
guarantees that there is no edge from 
the demand vertices in 
$\Rider\setminus \bigcup_{\mspair'=\mspair}^{\totalmspair}
\Rider^{(\mspair)}
=
\bigcup_{\mspair'=0}^{\mspair-1}
\Rider^{(\mspair')}$
to supply vertices in 
$\bigcup_{\mspair'=\mspair}^{\totalmspair}
\Driver^{(\mspair')}$,
the number of supply vertices in $
\bigcup_{\mspair'=\mspair}^{\totalmspair}
\Driver^{(\mspair')}$
at the first stage
is at most the number of demand vertices 
in $
\bigcup_{\mspair'=\mspair}^{\totalmspair}
\Rider^{(\mspair')}$.
Therefore, for every $\mspair\in[\totalmspair]$, we have 
$\sum_{\mspair'=\mspair}^{\totalmspair}
(\msopt^{(\mspair')} - \mssecond^{(\mspair')})
\leq 
\sum_{\mspair'=\mspair}^{\totalmspair}
\msrider^{(\mspair')}$.
 \vspace{2mm}
 
 \noindent\textbf{Step 2- restricting $\{\msfirsti\}_{\mspair\in[\totalmspair]}$:}
In this step, we 
argue that it is sufficient
to consider only solutions of program~\eqref{eq:offline approx second stage}, 
where the
constraints \emph{(Monotonicity-2)} are tight 
for all $\mspair \in [\totalmspair]$. 
Namely,
$\mssecondi = \msopti - \msrideri$
for all $\mspair \in [\totalmspair]$.
Consider any feasible solution of program~\eqref{eq:offline approx second stage}.
We modify this solution as follows, 
so that constraint \emph{(Monotonicity-2)} will be tight and the objective value weakly decreases:
Among all indices where 
constraints \emph{(Monotonicity-2)} 
are satisfied with strict inequality, 
let
$\mspair^*$ 
be the index that minimizes 
$\sum_{\mspair=\mspair^*}^{\totalmspair}
(\msrideri-\msopti + \mssecondi)$.
Denote 
$\sum_{\mspair=\mspair^*}^{\totalmspair}
(\msrideri-\msopti + \mssecondi)$
as $\Delta$.
Let $\mspair\primed \in [\mspair^* - 1]$
be the largest index where
the constraint \emph{(Monotonicity-2)}
is satisfied with equality
($\mspair\primed = 0$ if no such index exists).
Now set $\mssecond^{(\mspair^*)}_{\textrm{new}}
\gets 
\mssecond^{(\mspair^*)} -
\Delta$,
$\mssecond^{(\mspair\primed)}_{\textrm{new}}
\gets 
\mssecond^{(\mspair\primed)} +
\Delta$
if $\mspair\primed > 0$,
and keep all other variables unchanged.
The feasibility of the modified solution is by construction. Moreover, 
$\OPT$ remains unchanged. As the total mass in $\sum_{l\in[\totalmspair]}\mssecondi$ moves to lower $l$'s in $\{\mssecondi_{\textrm{new}}\}$, $\MS$ weakly decreases due to \emph{(Monotonicity-1)}, and so the objective value.

Therefore, by replacing $\mssecondi$
with $\msopti-\msrideri$ for 
all $\mspair\in[\totalmspair]$
and dropping constraint \emph{(Monotonicity-1)}
and
constraint \emph{(Monotonicity-2)}, 
we simplify program~\eqref{eq:offline approx second stage}
as
\begin{align}
    \label{eq:offline approx second stage 2}
\begin{array}{llll}
    \min\limits_{
\{(\msdriveri, \msopti, \msrideri)\}
_{\mspair\in[\totalmspair]}
    }
    &
    \frac{
    \firststagepercent
    \displaystyle\sum_{\mspair\in[\totalmspair]}
    \msopti
    +
    \secondstageratio
    \displaystyle\sum_{\mspair\in[\totalmspair]}
    \frac{\msdriveri-\msrideri}{\msdriveri}
    (\msopti - \msrideri)
    }{
    \displaystyle\sum_{\mspair\in[\totalmspair]}
    \msopti
    }
    & \text{s.t.} & \\
    &
    \mssecondi\leq \msopti \leq \msdriveri,
    \ 
    \msrideri \leq \msdriveri
    & 
    \mspair\in [\totalmspair] &\textit{\footnotesize{(Feasibility)}} \\
    ~&~&~&~
    \\
    &
    \displaystyle\sum_{\mspair\in[\totalmspair]}
    (\msopti - \msrideri)
    = 
    (1-\firststagepercent)
    \displaystyle\sum_{\mspair\in[\totalmspair]}
    \msopti
    \ 
    & 
     &\textit{\footnotesize{(Approximation)}} \\
    \\
    &
    \msdriveri,\msopti,
    \mssecondi, \msrideri\in \N 
    &
    \mspair\in[\totalmspair] & 
\end{array}
\end{align}

\noindent\textbf{Step 3- final evaluation:}
Here we find a lower-bound on the optimal value in 
program~\eqref{eq:offline approx second stage 2}.
Note that 
for any feasible solution,
we can lower bound the objective as 
\begin{align*}
    &\frac{
    \firststagepercent
    \displaystyle\sum_{\mspair\in[\totalmspair]}
    \msopti
    +
    \secondstageratio
    \displaystyle\sum_{\mspair\in[\totalmspair]}
    \frac{\msdriveri-\msrideri}{\msdriveri}
    (\msopti - \msrideri)
    }{
    \displaystyle\sum_{\mspair\in[\totalmspair]}
    \msopti
    }
    \overset{(a)}{\geq}
    \frac{
    \firststagepercent
    \displaystyle\sum_{\mspair\in[\totalmspair]}
    \msopti
    +
    \secondstageratio
    \displaystyle\sum_{\mspair\in[\totalmspair]}
    \frac{\msopti-\msrideri}{\msopti}
    (\msopti - \msrideri)
    }{
    \displaystyle\sum_{\mspair\in[\totalmspair]}
    \msopti
    }\\
    =&
    \firststagepercent + 
    \secondstageratio\frac{
    \displaystyle\sum_{\mspair\in[\totalmspair]}
     \frac{(\msopti-\msrideri)^2}{\msopti}
    }
    {
    \displaystyle\sum_{\mspair\in[\totalmspair]}
    \msopti
    }
    \overset{(b)}{\geq}
    \firststagepercent + 
    \secondstageratio\frac{
    \frac{\left(
    \displaystyle\sum_{\mspair\in[\totalmspair]}\msopti-\msrideri
    \right)^2
    }{
    \displaystyle\sum_{\mspair\in[\totalmspair]}\msopti
    }
    }{
    \displaystyle\sum_{\mspair\in[\totalmspair]}
    \msopti
    }
    \overset{(c)}{=}
    \firststagepercent
    +\secondstageratio(1-\firststagepercent)^2
\end{align*}
where inequality~(a) uses
$\msopti\leq \msdriveri$
in the constraint \emph{(Feasibility)},
inequality~(b) uses the Cauchy–Schwarz inequality
that 
$
\left(\sum_{\mspair\in[\totalmspair]}
     \left(\frac{\msopti-\msrideri}{\sqrt{\msopti}}\right)^2
     \right)
     \left(
     \sum_{\mspair\in[\totalmspair]}
     \left((\sqrt{\msopti}\right)^2
     \right)
     \geq 
     \left(
    \displaystyle\sum_{\mspair\in[\totalmspair]}\msopti-\msrideri
    \right)^2$,
    and 
    equality~(c) uses the constraint \emph{(Approximation)}.
    \hfill\halmos
\end{proof}

\subsection{More refined competitive ratio 
with dependence on the first stage graph}
\label{apx:robustness-first}

Here we present the more refined 
competitive ratio 
guarantee of \Cref{alg:weighted-skeleton}
which depends on
$\{y_\driver\}_{\driver \in \Driver}$,
i.e., 
the probability that each supply vertex $\driver\in\Driver$
is matched in the first stage in \Cref{alg:weighted-skeleton}.
Note that this result is general, and holds even for supply-weighted instances of the two-stage stochastic matching problem.


\begin{theorem}
\label{thm:opt-offline-comp-ratio refine}
For any differentiable, monotone increasing and strictly convex function $g:\R\rightarrow \R$,
suppose $\{y_\driver\}_{\driver \in \Driver}$,
is 
the probability that each supply vertex $\driver\in\Driver$
is matched in the first stage in \Cref{alg:weighted-skeleton},
then
\Cref{alg:weighted-skeleton}
is 
$\left(\min_{\driver\in\Driver}
1-y_\driver + y_\driver^2
\right)$-competitive 
against optimum offline benchmark, i.e., $\optoff$.
\end{theorem}

As we discussed in 
\Cref{sec:adversarial},
the convex program
\ref{eq:matching-convex-program}
identifies 
a structural decomposition 
$\{(\Rideri,\Driveri)\}_{\mspair\in[\totalmspair]}$
for the first-stage graph (\Cref{lemma:structure}).
The ``Uniformity'' property in 
\Cref{lemma:structure}
pins down $\mspairci$
for each pair $(\Rideri,\Driveri)$.
For problem instances with unweighted
supply vertices,
for each $\driver\in\Driveri$,
$1-y_\driver = \mspairci$.
For problem instances with weighted supply vertices,
for each $\driver\in\Driveri$,
$1-y_\driver = \frac{\mspairci}{\weight_\driver}$.

As an observation, note that 
the function $1-y + y^2$
is first decreasing and then increasing,  
with the smallest value $\frac{3}{4}$ 
attained
at $y = \frac{1}{2}$.
Hence,
\Cref{thm:opt-offline-comp-ratio refine}
suggests that 
when 
the first-stage randomized matching $M_1$
in \Cref{alg:weighted-skeleton}
has less randomness,
i.e.,
when 
$y_j$'s
are all close to either 0 or 1,
the final competitive ratio 
against optimum offline 
gets closer to 1.

The proof of \Cref{thm:opt-offline-comp-ratio refine}
follows from the same primal-dual argument 
as the one 
used to prove \Cref{thm:opt-offline-comp-ratio},
with the additional observation that 
the dual feasibility 
constraint
is
$\left(\min_{\driver\in\Driver}
1-y_\driver + y_\driver^2
\right)$-approximately satisfied for each edge; 
see inequalities~\eqref{eq:dual approx feasibility 1} and \eqref{eq:dual approx feasibility 2}.
To avoid redundancy, we omit the formal proof 
of \Cref{thm:opt-offline-comp-ratio refine}.

}
 
 \section{A combinatorial interpretation and generalized matching skeleton}
 \label{apx:goel}
 \label{sec:skeleton}
There is a natural connection between the unweighted version of our structural decomposition in \Cref{lemma:structure} and the concept of matching skeleton introduced in \cite{GKK-12}. To observe this connection, first consider a  bipartite graph $(P,Q,E)$ for which $P$ is \emph{hypermatchable}, that is, for every vertex $v\in Q$, there exists a perfect matching of the $P$ side that does not include $v$. By a purely combinatorial proof, it turns out that such a graph admits a nice decomposition similar to Lemma~\ref{lemma:structure}.
\begin{lemma}[\citealp{GKK-12}]
\label{lem:goel}
If $P$ is hypermatchcable, then the bipartite graph $G=(P,Q,E)$ can be decomposed into a collection of pairs $\left\{\left(P^{(l)},Q^{(l)}\right)\right\}_{l=1}^{L}, P^{(l)}\subseteq P, Q^{(l)}\subseteq Q$, such that:
\begin{itemize}
    \item $\{P^{(l)}\}_{l=1,\ldots,L}$ and $\{Q^{(l)}\}_{l=1,\ldots,L}$ partition $P$ and $Q$, respectively. 
    \item $\lvert P^{(l)}\rvert=\alpha_l\cdot \lvert Q^{(l)}\rvert$ for $\alpha_l\in(0,1]$, and the induced subgraph $G[P^{(l)}, Q^{(l)}]$ has a (fractional) matching that saturates each vertex in $Q^{(l)}$ exactly $\alpha_l$ times and each vertex in $P^{(l)}$ exactly once.
    \item There exists no edge in $E$ between $P^{(l)}$ and $Q^{(l')}$ if $\alpha_l>\alpha_{l'}$.
\end{itemize}
\end{lemma}
When $P$ is not  hypermatchable, \citet{GKK-12} extend the construction of \Cref{lem:goel} to the general case using the Edmonds-Gallai decomposition~\citep{edm-65,gal-64}. This construction essentially allows us to partition the vertices of $G$ into sets $P_{\textrm{up}}$, $P_{\textrm{center}}$, $P_{\textrm{down}}$,
$Q_{\textrm{up}}$, $Q_{\textrm{center}}$, and $Q_{\textrm{down}}$, such that $P_{\textrm{up}}$ is hypermatchable to $Q_{\textrm{down}}$, $Q_{\textrm{up}}$ is hypermatchable to $P_{\textrm{down}}$, and there is a perfect matching between $P_{\textrm{center}}$ and $Q_{\textrm{center}}$. Moreover, this decomposition guarantees that there is no edge from $P_{\textrm{center}}$ or $P_{\textrm{down}}$ to $Q_{\textrm{down}}$, and there is no edge from $P_{\textrm{down}}$ to $Q_{\textrm{center}}$. Using this partition, the final skeleton matching decomposition is constructed by applying \Cref{lem:goel} to each of the pieces $P_{\textrm{up}}\cup Q_{\textrm{down}}$ (from $P$ side) and $P_{\textrm{down}}\cup Q_{\textrm{up}}$ (from $Q$ side), and finally attaching everything to $P_{\textrm{center}}\cup Q_{\textrm{center}}$. See Sections~3.1 and 3.2 of \cite{GKK-12} for more details.

\paragraph{Generalization of the matching skeleton.} To see why the matching skeleton corresponds to the special case of \Cref{lemma:structure} for unweighted supplies, first note that the pair $\left(\Rider^{(0)},\Driver^{(0)}\right)$ plays the role of $\left(P_{\textrm{center}}\cup P_{\textrm{down}},Q_{\textrm{center}}\cup Q_{\textrm{up}}\right)$ as described above. Moreover, the collection of pairs $\left\{\left(\Rider^{(l)},\Driver^{(l)}\right)\right\}_{l=1}^L$ corresponds to the defined collection of pairs $\left\{\left(P^{(l)},Q^{(l)}\right)\right\}_{l=1}^L$ in \Cref{lem:goel} (applied to the induced subgraph $G[\Rider\setminus \Rider^{(0)},\Driver\setminus \Driver^{(0)}]$). The fractional matching $\xbf^*$ is the union of fractional matchings promised in the second bullet of \Cref{lem:goel}. Moreover, note that in the special case of \Cref{lemma:structure}, when all the weights are equal, the uniformity property implies that $\sum_{i\in N(j)}x^*_{ij}=\sum_{i\in N(j)}x^*_{ij'}$ for $j\in \Driver^{(l)}$, that is, fractional matching $\xbf^*$ (fractionally) saturates each vertex $j\in \Driver^{(l)}$ exactly the same amount. It also saturates each vertex in $\Rider^{(l)}$ once. This is indeed equivalent to the second bullet of \Cref{lem:goel}. Finally, the monotonicity is equivalent to the third bullet of \Cref{lem:goel}, simply because $c^{(l)}\leq c^{(l')}$ if and only if  $\sum_{i\in N(j)}x^*_{ij}\geq \sum_{i\in N(j)}x^*_{ij'}$ for $j\in \Driver^{(l)},j'\in D^{(l')}$, due to the uniformity property for equal weights. All in all, our structural decomposition generalizes the matching skeleton of \cite{GKK-12} to the vertex-weighted case. 

\paragraph{Algorithmic implications.} From an algorithmic point of view, if all the weights are equal, we can run \cref{alg:weighted-skeleton} without solving a convex program. To do so, we first find the matching skeleton decomposition (see \citealp{GKK-12} for fast and streaming algorithms to perform this task). Then, we find the maximum fractional matching for each of the pairs of the collection in \Cref{lemma:structure}, and finally we sample a bipartite matching from each of these fractional matchings (for efficient methods see, e.g., \citealp{gandhi2006dependent} or \citealp{sch-03}).

\section{Missing Proofs of Section~\ref{sec:improvedCR}}

\subsection{Omitted Proof Sketch for \Cref{prop:cont-greedy}}
\label{apx:sm-proof}
Consider maximizing  $g(T)=f^\wbf(T)-\sum_{j\in T}w_j$ subject to $T\in \bar{\B}$. By using the algorithm introduced in \cite{SVW-17}, i.e., applying a modified version of continuous greedy algorithm on the multilinear extension of $g(\cdot)$ (estimated through sampling) and then using pipage rounding, we output a set $\bar{T}$ in time $\textrm{poly}(\lvert \Driver\rvert,\lvert \Rider\rvert,\tfrac{1}{\epsilon})$ such that:\footnote{See Theorem~3.1 and Section~4 of \cite{SVW-17} for more details.} 
$$\forall T\in\bar{\B}: f^\wbf(\bar{T})-\sum_{j\in \bar{T}}w_j\geq  (1-\tfrac{1}{e}-\epsilon)f^\wbf(T)-\sum_{j\in T}w_j~. $$
Adding $\sum_{j\in \Driver}w_j$ to both sides of the above inequality finishes the proof.\hfill\Halmos

\subsection{Omitted Proof for Proposition~\ref{lem:online final} (Weighted Setting)}
\label{apx:online final}
We start by a technical lemma.
\begin{lemma}
\label{lem:program restriction}
Given space $X$, function $f:X \rightarrow \R$,
and subspace $X_1, X_2$
such that $X_1 \subseteq X_2 \subseteq X$,
suppose $\gamma_1 \triangleq \min_{x\in X_1} f(x)$
and $\gamma_2 \triangleq \min_{x \in X_2} f(x)$
are well-defined.
$\gamma_1 = \gamma_2$ if
there exists a mapping $\phi:X_2 \rightarrow X_1$
such that for all $x \in X_2$, either of the following
two
statements is satisfied:
\begin{enumerate}
    \item there exists $\gamma^{\dagger} < \gamma_1$,
     such that $f(x) \leq f(\phi(x)) \leq \gamma^{\dagger}$
    or $\gamma^{\dagger} \leq f(\phi(x)) \leq f(x)$; or
    \item there exists $\gamma^{\ddagger} > \gamma_1$, 
    such that 
    $f(\phi(x)) \leq f(x) \leq \gamma^{\dagger}$
    or $\gamma^{\dagger}\leq f(x) \leq f(\phi(x))$.
\end{enumerate}
\end{lemma}
\begin{proof}{\emph{Proof of \Cref{lem:program restriction}.}}
    We prove by contradiction.
    Suppose $\gamma_1 \not=\gamma_2$, 
    then $\gamma_1 > \gamma_2$.
    Consider the optimal solution 
    $x^* = \argmin_{x\in X_2}f(x)$.
    If statement-1 is satisfied for $x^*$,
    either 
    $\gamma_1 \leq f(\phi(x^*)) \leq \gamma^{\dagger} <
    \gamma_1$
    or
    $\gamma_1 \leq f(\phi(x^*)) \leq f(x^*) = \gamma_2 <
    \gamma_1$, 
    where both lead to contradiction.
    (b) If statement-2 is satisfied for $x^*$,
    either 
    $\gamma_1 \leq f(\phi(x^*)) \leq f(x^*) = \gamma_2 <
    \gamma_1$
    or 
    $\gamma_2 = f(x^*) \geq \gamma^{\ddagger} > \gamma_1$,
    where both lead to contradiction.\hfill\Halmos
\end{proof}
\lemmaoptonline*
\begin{proof}{\emph{Proof of \Cref{lem:online final}
for weighted setting.}}
we start by introducing a technical lemma.

Having the technical lemma, we start the proof of the weighted setting. This proof
is identical to the proof of the unweighted setting, up to the point where we consider
a program similar to program~\ref{eq:online all matched unweighted}
with additional variables $\weights$. We then need a different argument for the reduction to $\totalmspair = 1$, and an additional argument for a reduction to the identical weights.

\vspace{2mm}

\noindent\textbf{Summary of Step~1 and Step~2:} Consider 
the following optimization program,
which is parameterized by 
$\totaldriver,\totalmspair\in\N$,
and has variable
$\weights,
\{(\msdriveri,\msfirsti,\msrideri,\mssecondi)\}_{\mspair\in[\totalmspairHat]}$:
\begin{align}
    \label{eq:online all matched}
    \tag{$\mathscr{P}^{\textrm{FR-weighted-second}}$}
\begin{array}{llll}
    \min\limits_{
    \substack{
    \weights,\\
    \mspairinstancenew}
    }
    &
    \frac{
    \combprob
    \left(
    \MS
    \right)
    +
    (1-\combprob)
    \left(
    \GD
    \right)
    }{
    \OPT
    }
    & \text{s.t.} \\
    &
    \msfirsti + \mssecondi \leq \msdriveri,
    \ 
    \msrideri \leq \msdriveri
    & 
    \mspair\in [\totalmspair]
    &
    \footnotesize{\textit{(Feasibility-1)}}
    \\
    &
    \sum_{\mspair\in[\totalmspair]}\msdriveri = \totaldriver
    &
    &
    \footnotesize{\textit{(Feasibility-2)}}
    \\
    &
    (\msdriveri-\msfirsti-\mssecondi)
    \mspairci \leq
    (\msdriveri-\msfirsti-\mssecondi)\weighti
    &
    \mspair\in[\totalmspair], \driver \in\Driveri
    &
    \footnotesize{\textit{(Feasibility-3)}}
    \\
    &
    \mspairci < \mspairc^{(\mspair + 1)}
    &
    \mspair\in [\totalmspair - 1] 
    &
    \footnotesize{\textit{(Monotonicity-1)}}
    \\
    &
    \msfirst^{(\totalmspair)} \leq 
    \msrider^{(\totalmspair)}
    &
    &
    \footnotesize{\textit{(Monotonicity-2)}} 
    \\
    &
    \mspairci  
    =
    \frac{\msdriveri - \msrideri}{\sum_{\driver\in\Driveri}\frac{1}{\weighti}}
    &
    \mspair \in [\totalmspair]
    &
    \\
    &
    \msdriveri, \msfirsti, \msrideri,
    \mssecondi\in \N 
    &
    \mspair\in[\totalmspair] &\\
    &
    \weighti \in \RP
    &
    \driver\in[\totaldriver] &
\end{array}
\end{align}~,
where index sets are defined as 
$$
    \Driveri  
    \triangleq
    \left[
    \sum_{\mspair' = 1}^{\mspair - 1}\msdriver^{(\mspair')}+1
    :
    \sum_{\mspair' = 1}^{\mspair}\msdriver^{(\mspair')}
    \right],~~
    \Driverfirsti  
    \triangleq
    \left[
    \sum_{\mspair' = 1}^{\mspair - 1}\msdriver^{(\mspair')}+1
    :
    \sum_{\mspair' = 1}^{\mspair-1}\msdriver^{(\mspair')}
    +\msfirsti
    \right],~~
    \Driverfirst 
    \triangleq
    \displaystyle\bigcup_{\mspair\in[\totalmspair]}
    \Driverfirsti,
    $$
    $$
    \Driversecondi  
    \triangleq
    \left[
    \sum_{\mspair' = 1}^{\mspair}\msdriver^{(\mspair')}
    -\mssecondi + 1
    :
    \sum_{\mspair' = 1}^{\mspair}\msdriver^{(\mspair')}
    \right],~~
    \Driversecond 
    \triangleq
    \displaystyle\bigcup_{\mspair\in[\totalmspair]}
    \Driversecondi~,$$
and auxiliary variables $\MS, \GD,$ and $\OPT$
are defined as:
\begin{align*}
        \MS 
    \triangleq
    \sum_{\driver\in[\totaldriver]}
    \weighti
    -
    \sum_{\mspair\in[\totalmspair]}
    \mspairci(\msdriveri-\mssecondi),~~~~
    \GD  
    \triangleq
    \sum_{\driver\in\Driverfirst}
    \weighti
    +
    \left(1-\frac{1}{e}\right)
    \sum_{\driver\in\Driversecond}
    \weighti,~~~~
    \OPT  
    \triangleq
    \sum_{\driver\in\Driverfirst}
    \weighti
    +
    \sum_{\driver\in\Driversecond}
    \weighti~.
\end{align*}
Comparing
with 
program~\eqref{eq:online all matched unweighted}
for unweighted setting, 
program~\eqref{eq:online all matched}
introduces a new constraint (Feasibility-3)
which holds trivially for all solution for unweighted setting;
and a weaker constraint (Monotonicity-2) for $\mspair=\totalmspair$ only.

\vspace{2mm}
\noindent\textbf{Step 3- reduction to the case $\totalmspair=1$:}
In this step, we argue that it is sufficient
to consider only solutions of
program~\eqref{eq:online all matched} by arguing where
$\totalmspair = 1$
and constraint (Monotonicity-2) is tight.
To do this,  
consider the ``objective value'' in each
$\mspair\in[\totalmspair]$ defined as follows,
\begin{align*}
    \contributioni \triangleq
    \frac{\combprob\MSi + (1-\combprob)\GDi}{\OPTi}~,
\end{align*}
where 
$\MSi\triangleq\sum_{\driver\in\Driveri}\weighti - \mspairci\left(\msdriveri-\mssecondi\right)$,
$\GDi\triangleq\sum_{\driver\in\Driverfirsti}\weighti+\left(1-\sfrac{1}{e}\right)\sum_{\driver\in\Driversecondi}\weighti$,
and 
$\OPTi\triangleq\sum_{\driver\in\Driverfirsti}\weighti+\sum_{\driver\in\Driversecondi}\weighti$.
Notice that $\MS = \sum_{\mspair\in[\totalmspair]}\MSi$,
$\GD = \sum_{\mspair\in[\totalmspair]}\GDi$, and
$\OPT = \sum_{\mspair\in[\totalmspair]}\OPTi$.

We first argue that it is sufficient to consider only solutions where
$\contributioni$ is weakly decreasing in $\mspair$.
Consider an arbitrary feasible solution
of program~\eqref{eq:online all matched}
whose $\contributioni$
is not weakly decreasing in $\mspair$.
Let
$\mspair \in [\totalmspair - 1]$ 
be the index 
in this solution
such that $\contributioni < 
\contribution^{(\mspair+1)}$.
We modify this solution 
into another feasible solution
as follows:
set a scalar $\eta = \sfrac{\mspairc^{(\mspair+1)}}{\mspairc^{(\mspair)}}$.
By (Monotonicity-1), 
the scalar $\eta$ is strictly larger than 1.
Holding all other variables fixed,
if $\contributioni$ is weakly smaller than
the objective value of this solution,
we 
scale $\weighti$ up by 
scalar $\eta$
(i.e.\ $\weighti^\text{new} \gets \eta\,\weighti$)
for all $i\in \Driverfirsti$;
otherwise (i.e.\ $\contributioni$ and
$\contribution^{(\mspair+1)}$ are both strictly larger than
the objective value),
we scale $\weighti$ down by scalar
$\sfrac{1}{\eta}$ 
(i.e.\ $\weighti^\text{new} \gets \sfrac{\weighti}{\eta}$)
for all $i\in \Driverfirst^{(\mspair+1)}$.
The feasibility 
of the modified solution is by construction.
After the modification,
all $\rider \in\Driverfirsti\cup\Driverfirst^{(\mspair+1)}$
has the same $\msc_\text{new}$, and 
thus, $\Driverfirsti_\text{new} = \Driverfirsti$,
and $\Driverfirst^{(\mspair' - 1)}_\text{new} =\Driverfirst^{(\mspair')}$ for $\mspair' = \mspair +2 ,\dots \totalmspair$.
The objective value of the modified solution
is strictly smaller by construction.
Since 
$\totalmspair$ strictly decreases by one 
(i.e.\ $\totalmspair_\text{new} = \totalmspair - 1$)
through the modification,
for any feasible solution whose $\contributioni$
is not weakly decreasing in $\mspair$,
there exists 
a feasible solution
(which can be generated 
by iterating the modification 
finite times)
whose $\contributioni$
is weakly decreasing in $\mspair$
with strictly smaller objective value.

Next, we argue that it is sufficient to consider 
solutions with $\totalmspair = 1$.
Consider an arbitrary feasible solution 
of program~\eqref{eq:online all matched}
whose $\contributioni$
is weakly decreasing in $\mspair$.
We modify this solution as follows,
$\totaldriver_{\text{new}} \gets \msdriver^{(\totalmspair)}$;
$\totalmspair_{\text{new}} \gets 1$;
$\weighti^{\text{new}} \gets \weight_{\driver\primed}$
for all $\driver\in[\totaldriver_{\text{new}}]$ and 
$\driver\primed = \driver + \sum_{\mspair\in[\totalmspair-1]}\msdriveri$;
$\msdriver^{(1)}_{\text{new}} \gets \msdriver^{(\totalmspair)}$;
$\msfirst^{(1)}_{\text{new}} \gets \msfirst^{(\totalmspair)}$;
$\msrider^{(1)}_{\text{new}} \gets
\msrider^{(\totalmspair)}$;
and 
$\mssecond^{(1)}_{\text{new}} \gets 
\mssecond^{(\totalmspair)}$.
Since $\contributioni$ is weakly decreasing
in $\mspair$ in the original solution,
the modified solution has weakly smaller 
objective value.

Finally, we argue that it is sufficient to 
consider solutions with 
$\totalmspair_{\text{new}} = 1$,
and tight (Monotonicity-2).
We introduce the following modification
for arbitrary solutions with $\totalmspair = 1$
to another feasible solution with 
$\totalmspair = 1$ and tight
(Monotonicity-2),
and argue that the optimal objective value
for program~\eqref{eq:online all matched}
remains unchanged even if we restrict 
program
to these modified
solutions (by \Cref{lem:program restriction}). 
For notation simplicity,
since $\totalmspair = 1$, we 
drop the superscripts of 
$\{\msdriveri,\msrideri,\msfirsti,\mssecondi\}$.
The modification works as follows:
set 
$\msrider_\text{new} \gets \min\{\msrider,\totaldriver - \mssecond\}$,
$\msfirst_\text{new} \gets \min\{\msrider,\totaldriver - \mssecond\}$,
and hold all other variables unchanged.
The feasibility of modified solution is by construction.
Notice that $\msc_\text{new}$ might increase as
$\msrider_\text{new}$ decreases, however,
(Feasibility-3) is still satisfied trivially since
$\msfirst = \msdriver-\mssecond$ in such case.
We argue the effects of the objective value 
on 
$\msrider_\text{new}$
and $\msfirst_\text{new}$, separately.
For the effect on $\msrider_\text{new}$, since it is weakly 
smaller than $\msrider$,
$\msc_\text{new}$ weakly increases
and the objective value weakly decreases.
For the effect on $\msfirst_\text{new}$,
the denominator of the objective function
increases by 
$\sum_{\driver = \msfirst + 1}^{\msrider_\text{new}}\weighti$,
and the numerator increases by
$(1-\combprob)\sum_{\driver = \msfirst + 1}^{\msrider_\text{new}}\weighti$.
Thus, this modification 
makes
the objective value closer 
to value $1-\combprob$.
As we will prove at the end,
the value of program~\eqref{eq:online all matched}
(restricting to solution under this modification)
is
$0.7613$
with $\combprob = 0.7$,
which are larger than $1-\combprob$.
Thus, \Cref{lem:program restriction} implies that 
the value of program~\eqref{eq:online all matched}
restricting to solution under this modification
is the same as 
the value of this program without any restriction
on solutions.

Therefore,
the value of 
program~\eqref{eq:online all matched}
remains unchanged by restricting to those modified
solution with $\totalmspair=1$
and tight (Monotonicity-2).
We convert 
program~\eqref{eq:online all matched} parameterized
by $\totaldriver$ and $\totalmspair$
into
program~\eqref{eq:online hypermatchable} 
parameterized by $\totaldriver$
as follows,
\begin{align}
    \label{eq:online hypermatchable}
    \tag{$\mathscr{P}^{\textrm{FR-weighted-third}}$}
\begin{array}{llll}
    \min\limits_{
    \substack{
    \weights\\
    \msfirst,
    \mssecond\in\N}
    }
    &
    \frac{
    \combprob
    \left(
    \MS
    \right)
    +
    (1-\combprob)
    \left(
    \GD
    \right)
    }{
    \OPT
    }
    & \text{s.t.} &\\
    &
    \msfirst + \mssecond \leq \msdriver,
    & 
    &
    \footnotesize{\textit{(Feasibility-1)}} 
    \\
    &
    (\msdriver-\msfirst-\mssecond)\,\msc 
    \leq 
    (\msdriver-\msfirst-\mssecond)\,\weighti
    &
    \driver\in [\totaldriver]
    &
    \footnotesize{\textit{(Feasibility-3)}} 
    \\
    &
    \mspairc  
    =
    \frac{\msdriver - \msfirst}{\sum_{\driver\in[\totaldriver]}\frac{1}{\weighti}}
    &
    \mspair \in [\totalmspair]~,
\end{array}
\end{align}
where $\MS\triangleq
    \sum_{\driver\in[\totaldriver]}
    \weighti
    -
    \mspairc(\msdriver-\mssecond)$,
    $\GD\triangleq 
    \sum_{\driver\in[\msfirst]}
    \weighti
    +
    \left(1 - \frac{1}{e}\right)
    \sum_{\driver=\msdriver-\mssecond+1}^\msdriver
    \weighti$,
    and
    $\OPT  
    \triangleq
    \sum_{\driver\in[\msfirst]}
    \weighti
    +
    \sum_{\driver=\msdriver-\mssecond+1}^\msdriver
    \weighti$.

\vspace{2mm}
\noindent\textbf{Step 4- reduction to the case with identical weights:}
In this step, we show that it is sufficient
to consider the identical weight 
1, $\weightp$, $\weightdp$ 
such that $\weightdp \leq 1 \leq \weightp$
for all 
$\driver \in \{1, \dots, \msfirst\} (\triangleq\Driverfirst)$,
$\driver \in \{\msdriver-\mssecond + 1, \dots, \msdriver\}
(\triangleq\Driversecond)$,
and
$\driver \in  
\{\msfirst+1, \dots, \msdriver - \mssecond\}
(\triangleq\Driversecondb)$),
respectively.

We first argue that it is 
sufficient to consider solutions where 
$\weight_{\driverTilde} 
\leq \weight_{\driver}
\leq \weight_{\driverBar}$ 
for all $\driverTilde\in \Driversecondb$, 
$\driver\in\Driverfirst$
and 
$\driverBar\in\Driversecond$.
Consider an arbitrary feasible solution for 
program~\eqref{eq:online hypermatchable}.
If there exists $\driverTilde \in \Driversecondb$
and $\driver \in \Driverfirst$, such that
$\weighti < \weight_{\driverTilde}$,
by swapping weights 
$\weighti$ and $\weight_{\driverTilde}$
and holding all other variables 
unchanged,
the denominator of 
objective function increases by
$\weight_{\driverTilde} - \weighti$,
and the numerator 
increases by 
$(1-\combprob)(\weight_{\driverTilde} - \weighti)$, which makes 
the objective value closer to 
$1-\combprob$. 
By \Cref{lem:program restriction},
it is sufficient to consider 
solutions after
this modification.
Similarly,
if there exists $\driverBar \in \Driversecond$
and $\driver \in \Driverfirst$, such that
$\weighti > \weight_{\driverBar}$,
by swapping weights 
$\weighti$ and $\weight_{\driverBar}$
and holding all other variables 
unchanged,
the objective value decreases, since
the denominator of 
objective function remains unchanged,
and the numerator 
decreases by 
$\sfrac{1}{e}\,(1-\combprob)
(\weighti - \weight_{\driverBar})$.
Hence, it is sufficient to consider solutions after
this modification.

Now, we show it is sufficient to consider
solutions with identical weights $\weighti$
among $\Driversecondb$.
Consider the modification which
sets $\weight_{\driverTilde}^\text{new} \gets \frac{1}{|\Driversecondb|}\sum_{\driverTilde'\in \Driversecondb}\weight_{\driverTilde'}$
for all $\driverTilde \in \Driversecondb$
and holds all other variables unchanged.
The objective 
value weakly decreases since $\msc$ weakly
increases. 
By standard algebra, $\msc_\text{new} \geq \weight_{\driverTilde}^\text{new}$
for all $\driverTilde \in \Driversecondb_\text{new}$.
If $\Driversecondb_\text{new} \not=\emptyset$,
the assumption that $\weight_{\driverTilde}^\text{new} \leq \weight_\driver^\text{new}$
for all $\driverTilde \in \Driversecondb_\text{new}$
and $\driver\in \Driverfirst_\text{new}\cup\Driversecond_\text{new}$
ensures that $\msc_\text{new} \leq \weight_\driver^\text{new}$
for all $\driver \in
\Driverfirst_\text{new}\cup\Driversecond_\text{new}$.
Hence, (Feasibility-3) is satisfied for the modified solution.
The other constraints are also satisfied by construction.

To show it is sufficient to
consider identical weights $\weighti$
among $\Driversecond$.
Consider an arbitrary feasible solution. 
Holding all other variables unchanged,
we set $\weight_{\driverBar}^\text{new}$ for all 
$\driverBar\in \Driversecond$
with identical weight 
such that the induced $\msc_\text{new}$ 
in modified solution is the same
as $\msc$ in the original solution.
Notice that the 
summation of $\weight_{\driverBar}^\text{new}$
over all $\driverBar \in \Driversecond$
weakly decreases,
and thus, the objective value 
moves further away from 
$1 - \sfrac{1}{e}\,(1-\combprob)$
after this modification.
By \Cref{lem:program restriction},
the optimal objective value
remains unchanged after this modification.
To show identical weights $\weight_i$
among $\Driverfirst$,
the same argument as the one for 
$\Driversecond$ applies.

Therefore, it is sufficient to consider solutions
of program~\eqref{eq:online hypermatchable}
with identical weights over $\Driverfirst$,
$\Driversecond$
and $\Driversecondb$ respectively.

We finish the proof by 
normalizing the weights of 
all $\driver$ in $\Driverfirst$
(resp.\ $\Driversecond, \Driversecondb$) 
as 1
(resp.\ $\weightp, \weightdp$);
relaxing variables $\msfirst, \mssecond$ to be real numbers; and then normalizing them by $\msdriver$ to be in the range $[0,1]$. This proves that program~\ref{eq:online final} is a relaxation to the original program, as desired.
\hfill\Halmos

\end{proof}

 \section{Single-stage Joint Matching and Pricing
for Demand Efficiency}
\label{apx:onestagepricing}
 In this section, we study a variant of the joint matching 
and pricing problem. We diverge from other results in this paper by essentially considering a single stage problem.
The setup of this problem can be thought of an special case of the two-stage joint matching/pricing problem (\Cref{sec:two-stage-matching-pricing-def,sec:matching and pricing}), where:
(i) the set of first stage demand vertices $\Ra$ is empty, thus, at the first stage
the platform only needs to post prices to demand vertices in $\Rp$;
and
(ii) the weights of supply vertices are zero,
thus, the goal of the platform is to
maintain
\emph{demand efficiency}, i.e., 
maximizing the expected total value of matched demand vertices.
We show \Cref{alg:matching and pricing} introduced in
\Cref{sec:matching and pricing} is a $\left(1-\sfrac{1}{e}\right)$-competitive for this single-stage
problem.

To show this competitive ratio, we have an observation as follows. The matching generated by 
\Cref{alg:matching and pricing} is the demand-weighted maximum matching between $\Rpa$ and $\Driver$. The weight of matched demand vertices in such a matching is equal to the weighted rank function of a transversal matroid defined on the demand side (for the bipartite graph $G[\Rp,\Driver]$), evaluated
for demand vertices $\Rpa$ who
accept their prices. Note that
each demand vertex in $\Rp$
is selected independently to be put into $\Rpa$.

At the same time, the optimal solution
in program~\ref{eq:ex ante} can be thought of as the expected weight of 
a distribution over bipartite matchings in $G[\Rp,\Driver]$, which is equal to the value of the same weighted rank function for demand vertices $\Rpa\primed$; however, this time
demand vertices $\Rpa\primed$ are drawn in a correlated fashion from $\Rp$, but
have the same marginal probabilities as in $\Rpa$.

This is tightly related to the concept of \emph{correlation gap}, introduced by \citet{ADSY-10} (\Cref{def:correlation gap}): given a set system and a set function,
correlation gap quantifies the maximum ratio between the expected
value of the set function on a randomized set whose elements are drawn
independently, and the expected value on another randomized 
set whose elements are drawn correlated with the same marginal. 
For submodular functions, \citet{ADSY-10} show the correlation
gap is at most $1-\sfrac{1}{e}$ (\Cref{lem:correlation gap}).
This technique has also been used 
in other stochastic optimization problems
\citep[e.g.,][]{yan-11}.

\begin{definition}
[correlation gap]
\label{def:correlation gap}
For a set system $([n], \feasibles)$
with non-negative weights
$\mathbf\weight=\{\weight_i\}_{i=1}^n$,
the \emph{weighted rank function}
$\rho^{\mathbf\weight}(S,\feasibles)$ is the maximum
of $\sum_{i\in T}\weight_i$ over
all $T\subseteq S$.
The \emph{correlation gap for 
the weighted rank function}
under feasibility $\feasibles$
is 
$\sup_{
\mathbf\weight,
\mathcal D}
\frac
{
\expect[S\sim \mathcal D]
{\rho^{\mathbf\weight}(S,\feasibles)}
}
{
\expect[S\sim \mathcal I(\mathcal D)]
{\rho^{\mathbf\weight}(S,\feasibles)}
}
$,
where 
$\mathcal{D}$ is a
distribution over $2^{[n]}$
with marginal probability
$\{\quant_i\}_{i=1}^n$,
and $S \sim \mathcal I(\mathcal D)$
 denotes
that each $i\in[n]$ is included
in $S$ independently and with probability $\quant_i$.
\end{definition}
Weighted rank functions of matroids are known to be submodular~\citep{sch-03}.
\begin{lemma}[\citealp{ADSY-10}]
\label{lem:correlation gap}
The correlation gap of any submodular function, including the weighted rank
function of a matroid, is 
at most $1-\sfrac{1}{e}$.
\end{lemma}

\begin{theorem}
\label{thm:single stage pricing}
In single-stage joint matching and 
pricing for demand efficiency,
\Cref{alg:matching and pricing}
is $(1-\sfrac{1}{e})$-competitive.
\end{theorem}

\begin{proof}{\emph{Proof of \Cref{thm:single stage pricing}.}}
Let $\assignprobs,
\{\assignvertexprobi,
\assignvertexprobj\}_{\rider \in \Rider,\driver\in\Driver}$ be the optimal solution of the convex program~\ref{eq:ex ante}.
Note that the expected weight of 
the matching generated by \Cref{alg:matching and pricing}
and the objective value of program~\ref{eq:ex ante} are equal to 
the numerator and denominator of the correlation
gap for the weighted rank function of the transversal matroid defined on the demand side,
with weights
$\weight_\rider\primed \triangleq \weight_\rider(\Thresh_\rider(\assignvertexprobi))$,
and marginal probabilities $\quant_\rider\primed\triangleq \assignvertexprobi$
for all demand vertices $\rider\in\Rp$.
Invoking \Cref{lem:correlation gap} finishes
the proof.\hfill\Halmos
\end{proof}
 \section{Two-stage Matching with General Weights}
 \label{app:general-weight}
 
In this section, we consider two 
variations of the two-stage stochastic matching problem with
weights $\{\weight_\rider\}_{\rider\in \Rider}$
on demand vertices and its generalization with 
weights $\{\weight_e\}_{e\in E}$
on edges. We show that a simple algorithm, i.e., 
\Cref{alg:greedy with discard},
attains the optimal
competitive ratio $\sfrac{1}{2}$
against the optimum offline policy.

\begin{algorithm}
\begin{algorithmic}[1]
\State{\textbf{input:} bipartite graph $G=(\Rider,\Driver,E)$, non-negative edge weights $\{w_e\}_{e\in E}$.}
\State{\textbf{output:} bipartite matching $M_1$ in the graph $G[\Ra, \Driver]$.}
\State{Let $M_1$ be the maximum weighted 
matching in the graph $G[\Ra , \Driver]$.}
\State{Discard each edge $e \in M_1$
from matching $M_1$ with probability $\sfrac{1}{2}$.
}
\State{Return $M_1$.}
\end{algorithmic}
\caption{Greedy with Random Discard}
\label{alg:greedy with discard}
\end{algorithm}

\begin{lemma}
\label{lem:prior free weighted edge}
In the two-stage matching problem with weighted
edges, 
$\Cref{alg:greedy with discard}$
is $\sfrac{1}{2}$-competitive against
the optimum offline policy.
\end{lemma}
\begin{proof}{\textsl{Proof.}}
We consider the expected total weights of
edges matched in the first stage
and the second stage, separately.
At the first stage, since 
$\Cref{alg:greedy with discard}$
finds the maximum weighted matching 
and randomly discards each edge with probability
$\sfrac{1}{2}$, 
its final $M_1$ is at least 
$\sfrac{1}{2}$-approximation to the 
total weights of edges matched in the optimum
offline policy at the first stage.
Since each supply vertex is 
matched 
at the first stage 
with probability at most $\sfrac{1}{2}$,
the maximum weighted matching with all 
remaining supply vertices and new demand vertices
at the second stage
is at least
$\sfrac{1}{2}$-approximation to the 
total weights of edges matched in the optimum
offline policy at the second stage.
Combining the approximation for both stages,
the proof is finished.\hfill\Halmos
\end{proof}

\begin{example}
\label{example:Bayesian weighted rider}
Consider the following instance of two-stage
matching with weighted
demand vertices:
There is
one supply vertex and two demand vertices.
Demand vertex 1 is at the first stage 
with weight $1$.
Demand vertex 2
with weight $\sfrac{1}{\epsilon}$ appears at the second stage
with probability $\epsilon$.
Both demand vertices can be matched to the 
supply vertex.
\end{example}

\begin{lemma}
\label{lem:Bayesian weighted rider}
In the two-stage matching problem with weighted
demand vertices, no policy can obtain a competitive ratio better than $\tfrac{1}{2}$ against the optimum offline
policy.
\end{lemma}
\begin{proof}{\textsl{Proof.}}
Consider the instance in \Cref{example:Bayesian weighted rider}.
For any policy, the expected total weights
of matched demand vertices is at most $1$.
However, in the optimum offline policy 
the expected total weights of matched demand vertices
is $\epsilon\cdot\sfrac{1}{\epsilon} + (1 - \epsilon)\cdot 1 = 2 - \epsilon$.\hfill\Halmos
\end{proof}

Clearly, the above result also implies that in the two-stage matching problem with weighted
edges, no policy can obtain a competitive ratio better than $\tfrac{1}{2}$ against the optimum offline policy.

 \section{Missing Discussions from \Cref{sec:numerical}}
\subsection{Policies used in \Cref{sec:numerical}}
\label{apx:numerical-policies}
We evaluate and compare the 
performances of six different policies/benchmarks. The performance is defined to be the expected total weight of all matched drivers at the end of the second stage:
\paragraph{--Weighted balanced utilization (\texttt{WBU}).}
This policy is \Cref{alg:weighted-skeleton}
with convex function $g(x)=\exp(x)$
discussed in \Cref{sec:matching}.
It outputs $M_1$ in the first stage
without any knowledge of 
the second stage.

\paragraph{--Submodular maximization with local search (\texttt{SM})}. 
This policy is the greedy-style algorithm 
discussed in 
\Cref{prop:cont-greedy}.
As the choice of the approximation algorithm for our particular submodular maximization problem, we implement the ``Non-Oblivious Local Search'', which is fast and easy to implement; see \citet{SVW-17} for more details. We set the precision parameter to $\epsilon = 0.001$.

\paragraph{--Submodular maximization with limited iterations of local search
(\texttt{SM-Limit})}.
This policy is the same as the local search based
algorithm \texttt{SM} for submodular maximization, with only one difference: we limit the maximum number of iterations to 5. We also set the precision parameter to $\epsilon = 0.001$ as before.
Note that in practical implementation,
it is common to restrict 
the maximum number of iterations
of local search to control the running time.

\paragraph{--Hedge and greedy (\texttt{HG}).}
This policy is \Cref{alg:online}
discussed in \Cref{sec:improvedCR}.
It first finds the better of 
\texttt{WBU} and \texttt{SM} in terms of the expected total weight of the final matching (by running several rounds of Monte Carlo simulation to simulate the second stage). It then follows the winning algorithm.

\paragraph{--Myopic greedy policy (\texttt{GR}).}
In both stages, 
this policy outputs the maximum weighted matching
between current demand vertices 
and  available supply vertices.
This algorithm also outputs $M_1$ in the first stage
without any knowledge of 
the second stage.

\paragraph{--Optimum offline.}
This is the
omniscient 
offline policy that knows the exact
realization of available second stage demands $\Rpa$,
and picks the maximum weighted matching 
between $\Ra\cup\Rpa$ and $\Driver$.




\end{document}